\documentclass[10pt,journal,compsoc]{IEEEtran}

\usepackage[T1]{fontenc}
\ifCLASSOPTIONcompsoc

\usepackage{cite}

\usepackage{bm}
\usepackage{amsmath}
\usepackage{amsthm}
\usepackage{float}
\usepackage{setspace}
\usepackage{amssymb}
\usepackage{stfloats}
\usepackage{cite}
\usepackage{ragged2e}
\usepackage[ruled,vlined,linesnumbered]{algorithm2e}
\usepackage{amsfonts}
\usepackage{mathrsfs}
\usepackage{amsmath,amsthm}
\usepackage{array,booktabs}
\usepackage{subfigure}
\usepackage{multirow}
\usepackage{cuted}
\usepackage{multicol}
\usepackage{graphicx}
\usepackage{subfigure}
\usepackage{graphicx,xcolor,bm}
\usepackage{hyperref}
\usepackage{threeparttable}
\usepackage{dcolumn}
\usepackage{setspace}
\usepackage{makecell}
\usepackage{lipsum}
\usepackage{enumerate}
\usepackage{mathrsfs}

\newtheorem{Defn}{Definition}
\newtheorem{lem}{Lemma}

\newtheorem{thm}{Theorem}

\usepackage{cite,bm}
\graphicspath{{figures/}}
\begin{document}
	
	\title{Bridge the Present and Future:\\ A Cross-Layer Matching Game in Dynamic Cloud-Aided Mobile Edge Networks}
	
	\author{Houyi Qi*, Minghui Liwang*, \IEEEmembership{Member}, \IEEEmembership{IEEE}, Xianbin Wang,~\IEEEmembership{Fellow}, \IEEEmembership{IEEE}, Li Li, \IEEEmembership{Member}, \IEEEmembership{IEEE}, \\Wei Gong, \IEEEmembership{Member}, \IEEEmembership{IEEE}, Jian Jin, Zhenzhen Jiao
		\thanks{H. Qi (qihouyi@stu.xmu.edu.cn) and M. Liwang (minghuilw@xmu.edu.cn) are with the School of Informatics, Xiamen University, Fujian, China. X. Wang (xianbin.wang@uwo.ca) is with the Department of Electrical and Computer Engineering, Western University, Ontario, Canada. L. Li (lili@tongji.edu.cn) and W. Gong (weigong@tongji.edu.cn) are with the Department of Control Science and Engineering, and also with the Shanghai Research Institute for Intelligent Autonomous Systems, Tongji University, Shanghai, China. J. Jin (jin.jian@caict.ac.cn) is with the Research Institute of Industrial Internet of Things, CAICT, China. Z. Jiao (jiaozhenzhen@teleinfo.cn) is with the iF-Labs, Beijing Teleinfo Technology Co., Ltd., CAICT, China.
			
		Corresponding author: Minghui Liwang
		
		*H. Qi and M. Liwang contributed equally to this work.
	}}

	\IEEEtitleabstractindextext{
		\begin{abstract}
			\justifying
\textbf{C}loud-\textbf{a}ided \textbf{m}obile \textbf{e}dge \textbf{n}etworks (CAMENs) allow edge servers (ESs) to purchase resources from remote cloud servers (CSs), while overcoming resource shortage when handling computation-intensive tasks of mobile users (MUs). Conventional trading mechanisms (e.g., onsite trading) confront many challenges, including decision-making overhead (e.g., latency) and potential trading failures. This paper investigates a series of cross-layer matching mechanisms to achieve stable and cost-effective resource provisioning across different layers (i.e., MUs, ESs, CSs), seamlessly integrated into a novel hybrid paradigm that incorporates futures and spot trading. In futures trading, we explore an \textbf{o}verbooking-driven \textbf{a}forehand \textbf{c}ross-\textbf{l}ayer \textbf{m}atching (OA-CLM) mechanism, facilitating two future contract types: contract between MUs and ESs, and contract between ESs and CSs, while assessing potential risks under historical statistical analysis. In spot trading, we design two backup plans respond to current network/market conditions: determination on contractual MUs that should switch to local processing from edge/cloud services; and an \textbf{o}n\textbf{s}ite \textbf{c}ross-\textbf{l}ayer \textbf{m}atching (OS-CLM) mechanism that engages participants in real-time practical transactions. We next show that our matching mechanisms theoretically satisfy stability, individual rationality, competitive equilibrium, and weak Pareto optimality. Comprehensive simulations in real-world and numerical network settings confirm the corresponding efficacy, while revealing remarkable improvements in time/energy efficiency and social welfare.
		\end{abstract}
		\begin{IEEEkeywords}
			 Matching game, cloud-aided mobile edge networks, futures and spot trading, overbooking, cross-layer, risk control
		\end{IEEEkeywords}
	}
	
\maketitle
\IEEEdisplaynontitleabstractindextext
%
\IEEEpeerreviewmaketitle

\section{Introduction}
\IEEEPARstart{T}{he} past decade has witnessed a leap in the proliferation of Internet of Things (IoT), where the ever-growing loT applications that are computation-intensive and delay-sensitive have raised great challenges to resource- and capability-constrained mobile users (MUs, e.g., intelligent vehicles, smartphones) \cite{survey1}. To alleviate heavy workloads of MUs, applying cloud computing (CC) technology \cite{CC1,CC2,CC3,R2_2} can facilitate the migration of certain pressure of MUs to remote clouds, which, however, may incur unacceptable delays and heavy burdens on backhaul networks \cite{CCbad}. To this end, mobile edge computing (MEC) is emerging as a popular solution to offer responsive and cost-effective computing services near the network edge \cite{MEC1,MEC2,R2_3,R2_4,R2_5,R1}. Nevertheless, limited resources of edge servers (ESs) impose significant challenges in meeting the ever-growing service demands, especially during rush hour periods of data traffic \cite{8626532,9763875}. Fortunately, the \textbf{c}loud-\textbf{a}ided \textbf{m}obile \textbf{e}dge \textbf{n}etworks (CAMENs) offer a hybrid computing platform that flexibly leverages both CC and MEC, where cloud servers (CSs) can provide reliable backup resources to ESs, thereby attracting more MUs and thus financial incomes \cite{9763875}.

\subsection{Motivation}
	Ensuring the needed resources in CAMENs often calls for a certain resource trading mechanism \cite{9154594} with incentives, where \textit{requestors} (e.g., MUs and ESs) can offload a certain amount of task data to \textit{servers} (e.g., ESs and CSs) for processing, while paying for the acquired computing services. Nowadays, mainstream cloud and edge service providers, such as Amazon Web Services and Google Cloud Platform, commonly employ fixed-price strategies \cite{R2_6}, which offer reliable and available services to meet users' demands. Nevertheless, the rapid development of innovative technologies in computing, communication, and artificial intelligence calls for a more flexible manner of service provisioning (e.g., negotiable service prices and trading resources) for numerous mobile devices. Among the recent literature, fluctuant and variable service prices have raised wide attention \cite{9763875, 9154594}.
	When it comes to getting the resources needed for task execution, existing solutions on resource provisioning have primarily centered on one of the following trading strategies\cite{9771321,9453820}: \textit{i)} \textit{spot resource trading}, where requestors and servers reach an agreement relying on the current network/market conditions (e.g., an onsite trading mode); \textit{ii)} \textit{futures resource trading}, where servers and requestors pre-sign \textit{forward contracts}, which can be fulfilled for guiding the resource provisioning procedure in the future.
Note that a practical \textit{transaction} involves a trading event between resource servers and requestors, where servers provide computing services to requestors in exchange for payment. To bridge the gap between personalized service demands and supplies, matching theory offers an effective technique\cite{1990Two}, which has been extensively explored in existing literature, e.g., \cite{9154594} and \cite{9416305}.

\subsubsection{Hybrid resource trading: bridge present and future} 
To facilitate resource provisioning with proper incentives, spot trading has been widely deployed, allowing resource buying and selling among requestors and servers \cite{9453820,9416305} according to the current network/market information during each transaction. However, such an onsite mode can face significant drawbacks, especially for excessive overhead \cite{9763875,8854324,7476840,9357934} and potential trading failures \cite{9763875,8854324}. For example, the onsite decision-making procedure, e.g., looking for stable matching groups, can consume excessive time and energy, and thereby leave a negative impact on the quality of experience (QoE) of computing services. Moreover, limited computing resources can prevent some requesters from obtaining their required service, although they have spent time/energy during the trading decision-making process.

To cope with the aforementioned drawbacks, futures trading has been introduced to encourage participants to negotiate forward contracts ahead of future practical transactions, relying on analyzing past data/information \cite{9771321}. Nevertheless, various issues may still be caused. For example, improper contract terms can result in negative participants' utilities (e.g., an overlarge price can bring losses to MUs), and unsatisfying trading experience (e.g., an ES fails to deliver on its promise of computing service due to limited resource supply)\cite{9763875}. The above considerations collectively form the key motivation behind this work, in which we investigate a hybrid trading market that encompasses the advantages of both futures and spot trading. More importantly, to confront the dynamic resource demand and supply in CAMENs, we further introduce the concept of \textit{overbooking}, which plays an interesting and valuable role in facilitating seamless computing services.

\subsubsection{Overbooking: a helpful tool in dealing with dynamics}
Futures trading generally relies on the presale of resources, e.g., resources can be booked before practical demands. Conventional booking strategy encourages the amount of booked resources not to exceed the theoretical resource supply, which, however, confronts difficulties in handling fluctuant resource demand/supply. To immigrate such challenges, various commercial sectors, including airlines\cite{MA2019192}, hotels\cite{haynes2020perceptions}, and telecom companies\cite{9149184}, employ the strategy of \textit{overbooking} their resources (e.g., flight tickets, hotel rooms, and spectrum). For instance, in an effort to maximize their revenue, airlines frequently sell more tickets than the available seats on a flight. This practice helps ensure that flights depart with minimal vacant seats, as otherwise, a flight often takes off with 15\% or more unoccupied seats, resulting in economic losses\cite{6175013}. Thus, overbooking offers an effective strategy for handling dynamic and unpredictable resource demand/supply in a trading market \cite{9763875}. Accordingly, we adopt overbooking in CAMENs, where each server can pre-sign contracts with more requestors than its actual resource supply, to support timely service provisioning while improving resource utilization.

Motivated by the above discussions, we explore a hybrid service provisioning market that integrates both futures and spot trading modes in dynamic CAMENs with multiple MUs, ESs, and CSs. To facilitate responsive and cost-effective resource trading among them, we investigate a cross-layer matching game to bridge diverse resource demands and supplies, by obtaining proper mappings among MUs, ESs, and CSs.
	For futures trading mode, we employ the \textbf{a}forehand \textbf{c}ross-\textbf{l}ayer \textbf{m}atching (OA-CLM) mechanism which allows two contract types: \textit{i)} contract between MUs and ESs, and \textit{ii)} contract between ESs and CSs, which will be fulfilled during future resource transactions. To better assess the dynamics of CAMENs, we adopt interesting concepts of overbooking and risk management that are popularly concerned with daily economic behaviors, which is essential for successfully navigating the market dynamics over a long-term view. In the designed spot trading mode (i.e., during each practical transaction), we implement two contingency plans according to the present network/market conditions, including:
	\textit{i)} volunteer selection, where ESs make choices regarding \textit{volunteers}, who should temporarily waive their services from ESs/CSs due to resource shortage. These volunteers, in return, can receive compensation based on pre-signed forward contracts; and \textit{ii)} when there exist MUs with unreached demands, while some ESs/CSs still have surplus resources, we develop an \textbf{o}n\textbf{s}ite \textbf{c}ross-\textbf{l}ayer \textbf{m}atching (OS-CLM) mechanism. This mechanism facilitates resource trading among MUs, ESs, and CSs by evaluating the current network/market conditions, e.g., the remaining resource supply.

\subsection{Investigation}

This section conducts an extensive investigation into the research background and questions, focusing on the resource sharing in a multi-layer network architecture and the application of matching theory to facilitate these mechanisms.

\textbf{Matching game.} Our research seeks to obtain effective mappings a between requestors and servers to ensure responsive and cost-effective resource provisioning. This research topic is popular and can be found in various computing contexts, including MEC, and IoT domains. Specifically, the application of matching theory becomes popular in building a bridge between resource requestors and servers \cite{9154594,9127810,9851813,9108577,9200548}.
In \cite{9154594}, \textit{Wang et al.} modeled a distributed many-to-many matching model between mobile tasks and edges under diverse resource requirements and availabilities. 
In \cite{9127810}, \textit{Sharghivand et al.} considered quality of service in terms of service response time and proposed a two-sided matching between cloudlets and IoT applications.
In \cite{9851813}, \textit{Fang et al.} proposed a many-to-one matching algorithm between low earth orbit (LEO) satellites and geostationary orbit satellites to minimize delay and energy overhead in dual-layer satellite networks.
In \cite{9108577}, \textit{Fantacc et al.} matched ESs with applications, to minimize both the average system response time and the number of dropping requests in industrial IoT.
In \cite{9200548}, \textit{Raveendran et al.} investigated a many-to-many matching model to schedule resources of fog nodes according to the requirements of subscribers in IoT fog computing networks.
Although the aforementioned works have made certain contributions, they mainly study matching models regarding single-layer architectures, i.e., resource sharing between two parties. Apparently, involving more parties will further complicate the problem, e.g., a three-tier CAMEN. More importantly, having significant properties of conventional matching hold, e.g., stability, the problem becomes even more challenging \cite{8815852}. Thus, it is urgent and critical to design feasible cross-layer matching mechanisms when involving more parties into the market, where the amount of trading resources between MUs and ESs, as well as that between ESs and CSs, can impact each other.
 
\textbf{Resource allocation in three-tier network architectures.} Our study focuses on resource provisioning within three-tier network architectures, a field buoyed by the advancements in communication and computing technologies. To this end, concerted efforts have been made to refine task and resource scheduling processes, as highlighted by literature \cite{9344666,9214500,9195499,9838921,9616429}.
In \cite{9344666}, \textit{Tang et al.} solved an energy-aware task offloading problem in a cloud-edge integrated computing architecture in LEO satellite networks.
In \cite{9214500}, \textit{Aazam et al.} proposed a three-tier IoT-fog-cloud model to support the high scalability of IoT services and global energy efficiency.
In \cite{9195499}, \textit{Li et al.} jointly optimized data caching and task offloading in a three-tier MEC system.
In \cite{9838921}, \textit{Zhang et al.} considered a network with a remote CS, multiple mobile ESs, and users, aiming to minimize the overall cost in terms of energy and delay.
In \cite{9616429}, \textit{Tian et al.} conducted a response ratio offloading strategy centered on user preference and real-time nature to optimize the number of MUs that can be served by ESs or the CS. 
While the aforementioned studies have made valuable contributions, they primarily focus on onsite decision-making methods (e.g., spot trading). Nevertheless, such approaches can be susceptible to prolonged delays, heavy energy consumption, and the potential of trading failures. We are thereby motivated to develop a hybrid resource trading market within the framework of a three-tier network architecture in this paper. Our methodology offers several key unique features over existing approaches as summarized below.
	
\noindent $\bullet$ \textbf{Our considered resource trading mode:} Our study delves into a dynamic and uncertain market over CAMENS, and explores a resource trading paradigm that combines both futures and spot trading modes, which greatly distinguish our consideration with the traditional single mode. By implementing futures trading (i.e., our proposed OA-CLM game), limitations (e.g., extended delays, excessive energy consumption, and potential trading inaccuracies) associated with spot trading can be effectively mitigated, which enhances trading efficiency while benefiting environmental sustainability. Since uncertainties can impose risks, we further introduce spot trading (i.e., OS-CLM) as a backup plan to ensure reliable and timely resource provisioning. These two modes can greatly complement each other.

\noindent $\bullet$ \textbf{Overbooking and risk evaluation:} Our paper considers significant uncertainties in CAMENs, imposing fluctuations in resource demand/supply that can further impact the utilities of MUs, ESs and CSs. To this end, we adopt an interesting concept of overbooking, allowing ESs to reserve resources for more MUs than their practical resource supplies to cope with the dynamic resource demands, while improving resource utilization. Our analysis also includes an estimation of the risks that MUs, ESs and CSs may confront during practical transactions, with efforts made to keep these risks within an acceptable level. These two concepts are among the key and unique highlights of our paper, as they offer commendable advantages in supporting responsive and cost-effective services in dynamic CAMENs.

\vspace{-0.2cm}
\subsection{Novelty and Contribution}
This paper establishes a hybrid service trading market that integrates both futures and spot trading modes in dynamic CAMENs. Here, we investigate a set of cross-layer matching games to bridge diverse resource demands and supplies. For futures trading mode, to better assess CAMEN dynamics, we adopt economic concepts of overbooking and risk management. Specifically, an OA-CLM mechanism is deployed in advance to facilitate the contracts for future resource transactions. In spot trading, relying on current network/market conditions, we implement two plans: \textit{i)} selecting some MUs as volunteers based on pre-signed contracts; and \textit{ii)} developing the OS-CLM mechanism to modify/conduct the current resource trading during each practical transaction. To the best of our knowledge, this paper makes a pioneering effort in designing cross-layer and hybrid matching mechanisms in dynamic CAMENs. Key contributions are summarized below:

\noindent
$\bullet$ This paper introduces a hybrid market by integrating futures and spot trading modes over the CAMEN architecture, to build a bridge between the present and future resource provisioning. In particular, we involve a cross-layer matching game to establish stable connections among multiple MUs, ESs, and CSs, with an emphasis on optimizing the corresponding social welfare. Interestingly, overbooking is introduced to cope with market dynamics, e.g., uncertain MUs' resource demand, time-varying channel qualities, and ever-changing resource supply. Such a strategy can greatly support timely resource provisioning and commendable resource utilization.

\noindent
$\bullet$ For futures market, we develop an efficient OA-CLM mechanism that facilitates two types of forward contracts: contracts between MUs and ESs, as well as that between ESs and CSs, upon determining contract terms such as service prices and default clauses. Both contracts are pre-signed among participants under risk analysis, e.g., negative utility and unsatisfying trading experience, which will be fulfilled accordingly in each practical transaction. More importantly, our incorporation of overbooking allows the market to book more resources than the theoretical resource supply, thereby achieving good resource usage. We show that OA-CLM can ensure key properties, e.g., matching stability, individual rationality, and weak Pareto optimality.

\noindent
$\bullet$ For spot market, we engage two cases during each practical transaction: \textit{i)} when the overall resource demand of an ES surpasses its supply (including resources booked from CSs), the ES determines some contractual MUs as volunteers to give up the promised services; \textit{ii)} in cases where MUs without pre-signed forward contracts have unmet resource requirements, and there are surplus resources available from either ESs or CSs, OS-CLM mechanism is designed as a backup plan, helping to enhance the social welfare. We also show that OS-CLM can support crucial properties similar to those held by OA-CLM.

\noindent
$\bullet$ Comprehensive simulations are conducted that encompass both numerical settings and real-world datasets to validate the remarkable performance of our proposed cross-layer matching game within the hybrid service trading market, involving criteria such as social welfare, as well as time and energy efficiency.


\section{System Overview and Modeling}
\subsection{Overview}
\begin{figure*} \centering
	\vspace{-0.2cm}
	\subfigbottomskip=-1pt
	\subfigcapskip=-15cm
	\setlength{\abovecaptionskip}{-0cm}
	\setlength{\belowcaptionskip}{-0.5cm}
	\subfigure[] {
		\includegraphics[width=2.04\columnwidth]{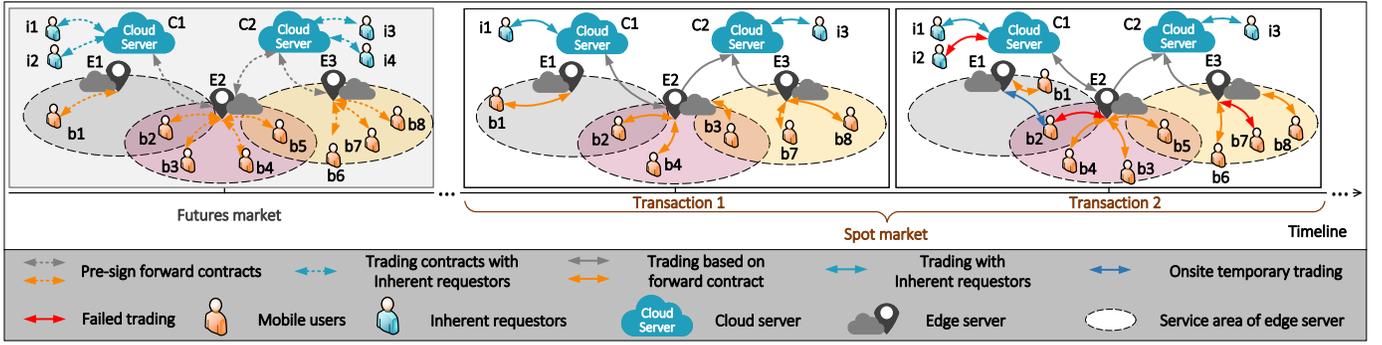} 
	} 
	
	\caption{Framework and procedure in terms of a timeline associated with our proposed cross-layer matching game in dynamic CAMENs. } 
	\label{fig} 
	\vspace{-0.3cm}
\end{figure*}
The considered hybrid service trading market over CAMENs involves three key parties (layers): \textit{i)} multiple MUs denoted by $ \bm{B}=\left\{b_1,...,b_i,...,b_{\bm{|B|}}\right\} $, where each of which carries a periodically generated computation-intensive task; \textit{ii)} multiple ESs represented by $ \bm{S^{E}}=\left\{s_1^E,...,s_j^E,...,s_{\bm{|\bm{S^{E}}|}}^E\right\} $, and \textit{iii)} multiple CSs modeled as $ \bm{S^{C}}=\left\{s_1^C,...,s_k^C,...,s_{\bm{|\bm{S^{C}}|}}^C\right\} $. Specifically, ESs can serve MUs, while borrowing resources from CSs when confronting resource shortages (e.g., when the resource supply of an ES fails to cover the overall demand from its MUs). In particular, non-negligible uncertain factors should be counted in a dynamic market: \textit{i)} uncertain MUs' resource demand, \textit{ii)} time-varying wireless channel qualities, and \textit{iii)} dynamic demand of \textit{inherent requestors (of CSs)}, since we consider a general and practical scenario where CSs can also serve other requestors (apart from the studied MUs)\cite{9763875}.
To ease the analysis, resources of ESs and CSs are quantized by virtual machines (VMs)\footnote{ In practice, resource owners can run independent tasks that are offloaded from resource requesters on a single physical server, and these tasks can access the underlying physical resources while being isolated from each other \cite{8289317}. For instance, the ES independently computes different tasks through different VMs\cite{8016573}. Such a setting represents the key reason of our assumption in which an VM in ES can only serve one MU's task \cite{8815852}.}, where Each VM is dedicated to serving a single task. Namely, each MU can be matched to a VM on an ES or a CS.

In this paper, we are interested in integrating: \textit{i)} futures trading, which facilitates participants to sign forward contracts under acceptable risks in advance to practical transactions, by implementing OA-CLM; and \textit{ii)} spot trading, where resources can be traded among participants by following the pre-signed forward contracts, or the well-designed OS-CLM. 
Specifically, a forward contract contains two key terms: service price and default clause. For the contract between a MU $ b_i $ and an ES $ s^E_j $, our model considers $ p_{i,j}^{U\rightarrow E} $ as the payment of ES $ s_j^E $ gets from MU $ b_i $, while $ q^{E\rightarrow U} $ and $ q^{U\rightarrow E} $ denote the penalty when MUs, or ESs break the contract, respectively. For the contract\footnote{Our market allows an ES to sign a contract with a CS for each task, which may further be offloaded to the CS.} between an ES $ s_j^E $ and a CS $ s_k^C $, let $ p_{i,j,k}^{E\rightarrow C} $ indicate the payment that CS $ s_k^C $ gets from ES $ s_j^E $ for further offloading $ b_i $'s task data to $ s_k^C $, and $ q^{E\rightarrow C} $ denote the penalty that an ES breaks the contract with a CS.
In general, lending resources to ES can incur resource shortages for CSs. In this case, a CS may have to pay refunds to their inherent requestors when their resources fail to afford the corresponding requirements, for which we use $ q_k^{C\rightarrow I} $ to describe the unit compensation from $ s_k^C $ to its inherent requestors.

Building a bridge to connect different parties calls for addressing a cross-layer matching problem. Our major goal in this paper is \textit{to obtain proper mappings among MUs, ESs, and CSs, to facilitate responsive and cost-effective resource provisioning in dynamic CAMENs.}

Fig. 1 depicts a schematic of our proposed cross-layer matching game\footnote{Note that our proposed cross-layer matching game is referred to as a distributed decision-making system. In this system, MUs, ESs, and CSs independently make decisions to determine their prices and select their interested traders.}. Regarding the procedure shown in Fig. 1 in terms of a timeline, participants first negotiate contracts according to the past data/information of historical transactions, before future practical transactions (futures market). Then, in each practical transaction (spot market), contractual participants will perform their pre-signed contracts, while others can implement onsite transactions if needed. For example, considering the coexistence of two CSs (i.e., C1 and C2), three ESs (i.e., E1-E3), and eight MUs (i.e., b1-b8) in Fig. 1, participants first sign forward contracts through our proposed OA-CLM game. Then, in each practical transaction (e.g., Transactions 1 and 2 in Fig. 1), ESs can select some contractual MUs as volunteers when the overall resource demand exceeds their supply (e.g., b2 and b7 are selected as volunteers by E2 and E3, respectively). Also, the uncertain resource demand of MUs associated with an ES may enable surplus resources and thus can serve other MUs through our designed OS-CLM game (e.g., E1 serves b2 in Transaction 2). Besides, a CS may have to compensate for some inherent requestors when its resources are insufficient to support the demand of its inherent requestors (e.g., C1 can not provide service to i2 in Transaction 2, since a certain amount of VMs has been booked by E2).

\subsection{Preliminary Modeling}
We next introduce key participated parties, i.e., MUs, ESs, and CSs, along with significant correlative uncertainties.

\textbf{Modeling of MUs.} Each MU $ b_i\in\bm{B} $ is modeled by an 8-tuple $ b_i=\left\{f_i^U,e_i^{t},e_i^U,r_i^U,d_i^U,\bm{C_i},\alpha_i,\gamma_{i,j}\right\} $, where $ f_i^U $ denotes its local computing capability (e.g., CPU cycles/s); $ e_i^{t} $ (Watt) and $ e_i^U $ (Watt) indicate the transmission power, and local computing power of $ b_i $, respectively; $ r_i^U $ refers to the required computing resources of MU $ b_i $ (CPU cycles), $ d_i^U $ denotes the corresponding data size (bit), and $ \bm{C_i} $ represents the set of ESs that can serve $ b_i $, i.e., ESs that are within the communication range of $ b_i $. To capture the dynamic nature of the network, we consider two uncertainties for MUs: \textit{i)} the uncertain participation of MUs is denoted by $ \alpha_i $, describing whether $ b_i $ attends a transaction and follows a Bernoulli distribution, $ \alpha_i \sim \mathbf{B}\left\{(1,0),(a_i,1-a_i)\right\} $. Specifically, $ \alpha_i=1 $ indicates that $ b_i $ takes part in the current transaction; while $ \alpha_i=0 $ otherwise; and \textit{ii)} the time-varying network condition among ESs and MUs is modeled as $ \gamma_{i,j} $, obeying an uniform distribution\footnote{This paper assumes that MUs are randomly moving within a region. In addition, $ \gamma_{i,j} $ denotes a function of multiple factors such as channel fading, path loss, and noise\cite{9763875,8854324}.} $ \gamma_{i,j} \sim \mathbf{U}(\mu_1,\mu_2) $, where fluctuations can be incurred due to multiple factors such as users' mobility and obstacles \cite{9763875,8854324}. 

\textbf{Modeling of ESs.}
Each ES $ s_j^E\in\bm{{S}^{E}} $ can be described by a 4-tuple $ s_j^E=\left\{f_j^E,e_j^E,G_j^E,K_j^E\right\} $, where $ f_j^E $ is the computing capability of $ s_j^E $ (e.g., CPU cycles/s), $ e_j^E $ (Watt) represents the local computing power. Each ES $ s_j^E $ owns $ G_j^E $ VMs, while a single VM can handle one task. In addition, each ES $ s_j^E $ has a set of orthogonal subcarriers\footnote{We assume that each ES connecting to an access point (AP, e.g., base station and roadside unit), has a certain number of orthogonal sub-carriers. This limits the number of MUs that can simultaneously access the ES, as each orthogonal sub-carrier can only serve one MU\cite{8815852}. This approach acknowledges the finite nature of bandwidth by limiting it through the number of orthogonal sub-carriers.} that can serve at most $ K_j^E $ MUs simultaneously\footnote{We make a general assumption that the number of VMs owned by ESs does not exceed the number of orthogonal subcarriers they have.} 

\textbf{Modeling of CSs.}
A CS $ s_k^C\in\bm{{S}^{C}} $ is denoted by a 3-tuple $ s_k^C=\left\{f_k^C,e_k^C,G_k^C\right\} $, where $ f_k^C $ denotes the computing capability of $ s_k^C $ (e.g., CPU cycles/s), $ e_k^C $ (Watt) indicates the computing power, and $ s_k^C $ owns $ G_k^C $ VMs. To better assess fluctuating resource demand, our model also involves inherent requestors of CSs, for which a random variable $ \varepsilon_k $ is considered that follows a Poisson distribution, denoted by $ \varepsilon_k \sim \mathbf{P}(\sigma_k) $, where $ \varepsilon_k\in \left\{0,1,..., G_k^C\right\} $. Apparently, $ \varepsilon_k $ reflects the resource demand of other requestors (without considering demands from the studied ESs and MUs).

\section{OA-CLM Game in Futures Market}
We next introduce OA-CLM game, which is a unique cross-layer matching tailored to the characteristics of the futures market. Our purpose is to achieve stable and reliable mappings among MU, ES, and CS, and reach two types of forward contracts (e.g., contract terms such as prices and default clauses) according to the mapping results.
We first define the following assignment $ \omega(.) $ between MUs and ESs, which represents a significant notation in the following model designs.

\noindent
$\bullet$ $ \omega\left(s_j^E\right) $: the set of MUs served by ES $ s_j^E $, where $ \omega\left(s_j^E\right)\subseteq\bm{B} $;

\noindent
$\bullet$ $ \omega\left(b_i\right) $: an ES in $ \bm{C_i} $ that serves MU $ b_i $, where $ \omega\left(b_i\right)\subseteq\bm{C_i} $, and $ \bm{C_i} \subseteq\bm{{S}^{E}} $;

\subsection{Utility, Expected Utility and Risk}
\subsubsection{Utility, expected utility and risk of MUs}
The local task completion time of $ b_i $ can be calculated by $ \frac{r_i^U}{f_i^U} $, while that of edge computing (e.g., $ b_i $ offloads its task to ES $ s_j^E $) is computed as $ \frac{r_i^U}{f_j^E}+\frac{d_i^U}{W\text{log}_2(1+e_i^{t}\gamma_{i,j})} $. Specifically, $ W $ is the bandwidth of MU-to-ES communication links\footnote{We primarily focuses on the allocation of computing resources (i.e., VMs) among MUs, ESs, and CSs in dynamic CANEMs. For analytical simplicity, we assume that the bandwidth of the communication link between each MU and an ES is consistent. Namely, we do not spend efforts on solving the bandwidth allocation problem among MUs, ESs, and CSs, aligning with existing related literature, such as \cite{9763875,8815852}.}, and $ e_i^{t}\gamma_{i,j} $ indicates the received signal noise ratio (SNR) of ES $ s^E_j $. Thus, for $b_i$, we calculate the amount of time that can be saved from enjoying edge service\footnote{The model of MUs only concerns delay incurred by data transmission to an ES, and its upload and data processing time, due to the following reasons: \textit{i)} the ES-CS transmission delay through wired connections can be much smaller than that encountered by MUs; then \textit{ii)} given data privacy concerns, MUs are typically unaware of whether their tasks are offloaded to CSs. Therefore, from the MUs' perspective, they can ignore the transmission delay between an ES and a CS.},
\begin{equation}\label{key}{\small
		\begin{aligned}
		t_{i,j}^{save}=\frac{r_i^U}{f_i^U}-\left(\frac{r_i^U}{f_j^E}+\frac{d_i^U}{W\text{log}_2(1+e_i^{t}\gamma_{i,j})}\right),
	\end{aligned}}
\end{equation}
and that of energy consumption, as given by
\begin{equation}\label{key}{\small
		\begin{aligned}
				c_{i,j}^{save}=\frac{r_i^Ue_i^U}{f_i^U}-\frac{d_i^Ue_i^{t}}{W\text{log}_2(1+e_i^{t}\gamma_{i,j})}.
		\end{aligned}}
\end{equation}

Accordingly, we define the valuation (e.g., the profit on saving time and energy) that ES $ s_j^E $ can bring to MU $ b_i $, as shown by
\begin{equation}\label{key}{\small
		\begin{aligned}
v_{i,j}=\mathbb{V}_1t_{i,j}^{save}+\mathbb{V}_2c_{i,j}^{save},
	\end{aligned}}
\end{equation}
where $ \mathbb{V}_1 $ and $ \mathbb{V}_2 $ are positive weighting coefficients.

\textbf{MU's utility and its expectation.} The utility that MU $ b_i $ can obtain from trading with ESs covers three key parts: \textit{i)} the valuation that ES $ s_j^E \in\omega(b_i) $ can bring to MU $ b_i $ minus its payment to $ s_j^E $, \textit{ii)} the penalty when $ b_i $ breaks the contract (e.g., $ \alpha_i=0 $ in a practical transaction), and \textit{iii)} the compensation if $ b_i $ is selected as a volunteer, expressed as
\begin{equation}\label{key}{\small
		\begin{aligned} 
	&u^U(b_i,\omega\left(b_i\right))=\left(1-\lambda_{i,j}\right)\alpha_i\left(v_{i,j}-p_{i,j}^{U\rightarrow E}\right)\\&-(1-\alpha_i)q^{U\rightarrow E}+{\alpha_i\lambda}_{i,j}q^{E\rightarrow U},
		\end{aligned} }
\end{equation}
where $ \lambda_{i,j}=1 $ denotes that $ b_i $ is selected by $ s_j^E $ as a volunteer in a practical transaction, and $ \lambda_{i,j}=0 $ otherwise. Since the uncertain factors impose great challenges to directly maximize the practical value of (4) in our designed futures market, we opt to involve its expected value as shown below
\begin{equation}\label{key}{\small
			\begin{aligned}
	&\overline{u^U}(b_i,\omega\left(b_i\right))=\left(1-\text{E}[\lambda_{i,j}])\text{E}[\alpha_{i}](\text{E}[v_{i,j}]-p_{i,j}^{U\rightarrow E}\right)\\&-(1-\text{E}[\alpha_{i}])q^{U\rightarrow E}+\text{E}[\alpha_{i}]\text{E}[\lambda_{i,j}]q^{E\rightarrow U},
			\end{aligned}}
\end{equation}
where we have $ \text{E}[\alpha_{i}]=\text{Pr}(\alpha_i=1)\times 1+\text{Pr}(\alpha_i=0)\times 0=a_i$, and the derivations of $ \text{E}[v_{i,j}]$ and $ \text{E}[\lambda_{i,j}] $ are detailed in Appendix B.1.

\textbf{Risk assessment of MUs.} A futures market is generally a coexistence of benefits and risks, as the uncertainties can bring losses to participants. Thus, we assess two risks for each MU $ b_i \in \bm{B} $.
First, a MU $ b_i $ is risking an unsatisfying utility (e.g., the value of $u^U(b_i,\omega\left(b_i\right))$ turns negative) during a practical transaction, which is defined by the probability that the utility of $ b_i $ falls below a tolerable value $ u_{min} $, as given by
\begin{equation}\label{key}{\small
		\begin{aligned}
			R^{U}_1\left( b_{i},\omega \left( b_{i} \right) \right)=\text{Pr}\left(u^U(b_i,\omega\left(b_i\right))<u_{min}\right) ,	
	\end{aligned}}
\end{equation}
where $ u_{min} $ is a positive value approaching to 0.
Then, as resource overbooking is allowed in our designed market, a contractual MU (e.g., a MU who has signed a forward contract with an ES) may be selected as a volunteer and thus forced to process its task locally. Correspondingly, such a risk can be concerned as the probability when $ \lambda_{i,j}=1 $. 
\begin{equation}\label{key}{\small
		\begin{aligned}
		R^{U}_2\left( b_{i},\omega \left( b_{i} \right) \right)=\text{Pr}(\lambda_{i,j}=1) .	
	\end{aligned}}
\end{equation}

Both of the above-discussed risks should be acceptable when designing forward contracts; otherwise, each $b_i$ will opt for spot trading instead, meaning that they will not enter into forward contracts with any ES in the futures market.

\subsubsection{Utility, expected utility and risk of ESs}
 We next define the assignment $ \varphi(.) $ between ESs and CSs:
 
 \noindent
 $\bullet$ $ \varphi\left(s_j^E\right) $: the set of CSs that provide resources to ES $ s_j^E $, where $ \varphi\left(s_j^E\right)\subseteq\bm{{S}^{C}} $;
 
 \noindent
 $\bullet$ $ \varphi\left(s_k^C\right) $: the set of ESs that purchase resources from CS $ s_k^C $, where $ \varphi\left(s_k^C\right)\subseteq\bm{{S}^{E}} $;
 
 Processing MUs' tasks can consume certain costs to ESs, where the monetary cost incurred by energy consumption for $ s_j^E $ contributing services to MU $ b_i $ can be calculated as
 \begin{equation}\label{key}{\small
 		\begin{aligned}
 		c_{i,j}^E=\mathbb{V}_3\frac{r_i^Ue_j^E}{f_j^E}+c^{H}_j,		
 	\end{aligned}}
 \end{equation}
 where $ \mathbb{V}_3 $ represents the cost coefficient to enable a unified unit for energy, e.g., dollars; while $ c^{H}_j $ describes the hardware cost, which is assumed to be a constant and the same across all the VMs owned by $ s_j^E $ \cite{8815852}.
 
	 As the resources of an ES is generally limited, it can borrow resources from CSs to attract more MUs, which mainly depends on whether the overall gathered resource demand (e.g., $ |\omega\left(s_j^E\right)| $) exceeds its supply (e.g., $ (1+\tau)G_j^E $, where $ \tau $ denotes an overbooking rate\footnote{It is critical to apply overbooking strategies with carefulness to avoid buyer dissatisfaction and potential risks to both parties (e.g., resource sellers and buyers). Therefore, our overbooking rate setting is meticulously considered, detailed in Sec. 5.3.3.}). Thus, we use $ x_j=|\omega\left(s_j^E\right)|-(1+\tau)G_j^E $ to describe the necessity for purchasing resources from CSs. Apparently, $ x_j>0 $ means that $ s_j^E $ is confronting local resource shortage and needs a certain amount of cloud resources; while $ x_j\le0 $ otherwise. To answer the question of which tasks (namely, MUs) should be migrated to CSs, we use $ \mathbb{B}_{j,k} $ to describe the set of MUs, while their tasks will be offloaded to CS $ s^C_k $ for further processing.
 Let $ \beta_{i,j,k} $ indicate whether $ s_j^E $ fulfills a contract with CS $ s_k^C \in \varphi\left(s_j^E\right) $ about a task of $b_i$ in $ \mathbb{B}_{j,k} $ during a practical transaction, as defined by,
\begin{equation}\label{key}	{\small
	\beta_{i,j,k} = \left\{ \begin{matrix}
		{1,~~~ s_j^E \text{ fulfills a contract with } s_k^C \text{ for }b_i} \\
		{0,~~~ s_j^E \text{ breaks a contract with } s_k^C \text{ for }b_i} \\
	\end{matrix} \right .}
\end{equation}
where $ \beta_{i,j,k} =1 $ denotes that $ s_j^E $ will purchase cloud resources from $ s_k^C $ for $b_i$'s task, and $ \beta_{i,j,k} =0 $, otherwise. 

\textbf{ES's utility and its expectation.} The utility that ES $ s_j^E $ obtains from MUs involves three key parts: \textit{i)} the overall payment from MUs (e.g., $ b_i\in \omega\left(s_j^E\right) $) minus the cost for offering services to them, \textit{ii)} the penalty from MUs who break contracts (e.g., $ \alpha_i=0 $ in a practical transaction), and \textit{iii)} the compensation to volunteers, which is calculated by
\begin{equation}\label{key}{\small
\begin{aligned}
	&u^{U\leftrightarrow E}\left(s_j^E,\omega\left(s_j^E\right)\right)=\sum_{b_i\in\omega\left(s_j^E\right)}\alpha_i\left(1-\lambda_{i,j}\right)\left(p_{i,j}^{U\rightarrow E}-c_{i,j}^E\right)\\&+\sum_{b_i\in\omega\left(s_j^E\right)}\left(\left(1-\alpha_i\right)q^{U\rightarrow E}-{\alpha_i\lambda}_{i,j}q^{E\rightarrow U}\right).
\end{aligned}}
\end{equation}
As ESs implement two-way trading in such a three-tier network, an ES can obtain profit from CSs, where the corresponding utility includes three key aspects: \textit{i)} the payment to CSs for purchasing their resources, \textit{ii)} the saved service cost of ES through offloading MUs' tasks to CSs, and \textit{iii)} the compensation for breaking contracts (e.g., $ \beta_{i,j,k}=0 $), as given by
\begin{equation}\label{key}{\small
\begin{aligned}
	&u^{E\leftrightarrow C}\left(s_j^E,\varphi\left(s_k^E\right)\right)=-\sum_{s_k^C\in\varphi\left(s_j^E\right)}\sum_{b_i\in \mathbb{B}_{j,k}} \beta_{i,j,k}\left(p_{i,j,k}^{E\rightarrow C}-c_{i,j}^E\right)\\&-\sum_{s_k^C\in\varphi\left(s_j^E\right)}\sum_{b_i\in \mathbb{B}_{j,k}}\left(1-\ \beta_{i,j,k}\right)q^{E\rightarrow C}.
\end{aligned}}
\end{equation}


 Accordingly, we can compute the overall utility of $ s_j^E $ as,
\begin{equation}\label{key}{\small
\begin{aligned}
	&u^E\left(s_j^E,\omega\left(s_j^E\right),\varphi\left(s_j^E\right)\right)\\&=u^{U\leftrightarrow E}\left(s_j^E,\omega\left(s_j^E\right)\right)+u^{E\leftrightarrow C}\left (s_j^E,\varphi\left(s_j^E\right)\right ).
\end{aligned}}
\end{equation}

As the uncertainties stop us to obtain the practical value of (12) directly in futures market, the corresponding expectation of $ u^E\left(s_j^E,\omega\left(s_j^E\right),\varphi\left(s_j^E\right)\right) $ can be considered below
\begin{equation}\label{key}\small{
\begin{aligned}
	&\overline{u^E}\left(s_j^E,\omega\left(s_j^E\right),\varphi\left(s_j^E\right)\right)=
	\\&\sum_{b_i\in\omega\left(s_j^E\right)}a_i\left(1-\text{E}[\lambda_{i,j}]\right)\left(p_{i,j}^{U\rightarrow E}-c_{i,j}^E\right)\\&+\sum_{b_i\in\omega\left(s_j^E\right)}\left(\left(1-a_i\right)q^{U\rightarrow E}-{a_i\text{E}[\lambda}_{i,j}]q^{E\rightarrow U}\right)\\&-\sum_{s_k^C\in\varphi\left(s_j^E\right)}\sum_{b_i\in \mathbb{B}_{j,k}} \text{E}[\beta_{i,j,k}]\left(p_{i,j,k}^{E\rightarrow C}-c_{i,j}^E\right)\\&-\sum_{s_k^C\in\varphi\left(s_j^E\right)}\sum_{b_i\in \mathbb{B}_{j,k}}\left(1-\text{E} [\beta_{i,j,k}]\right)q^{E\rightarrow C},
\end{aligned}}
\end{equation}
where the derivation of $ \text{E}[\beta_{i,j,k}] $ is given in Appendix B.2.

\textbf{Risk assessment of ESs.} As the middle layer of CAMENs, ESs have to gather two-way information/data from MUs and CSs, which are both uncertain, e.g., dynamic MUs' resource demand and CSs' resource supply. To this end, an ES should estimate its risks involving: \textit{i)} the risk of breaking a contracts with a CS, as given by 
\begin{equation}\label{key}
\small{
\begin{aligned} 
	R_1^E\left(s_j^E,s_k^C\right)=\text{Pr}\left(\beta_{i,j,k}=0\right),~ \forall s_k^C\in\varphi\left(s_j^E\right), \forall b_i \in \mathbb{B}_{j,k}
\end{aligned}}
\end{equation}
and \textit{ii)} the risk of breaking the contract with MUs due to overbooking (i.e., the demand of ES exceeds its resource supply), is expressed as
\begin{equation}\label{key}\small{
\begin{aligned} 
	&R_2^E\left(s_j^E,\omega\left(s_j^E\right),\varphi\left(s_j^E\right)\right)=\\& \text{Pr}\left(\sum_{b_i\in\omega\left(s_j^E\right)}\alpha_i>\sum_{s_k^C\in\varphi\left(s_j^E\right)} |\mathbb{B}_{j,k}|+G_j^E\right),
\end{aligned}}
\end{equation} 
where $ \sum_{s_k^C\in\varphi\left(s_j^E\right)}|\mathbb{B}_{j,k}| $ denotes the number of tasks offloaded to CSs in $ \varphi\left(s_j^E\right) $.

For an ES, it is essential to effectively manage the risks mentioned above; otherwise, they may prefer participating in the spot trading market rather than entering into forward contracts in the futures market.

\vspace{-0.2cm}
\subsubsection{Utility, expected utility and risk of CSs}
The overhead of a CS $ s_k^C $ for serving the task of MU $ b_i $ offloaded by $ s_j^E $ can be computed as
\begin{equation}\label{key}{\small
		\begin{aligned}
		c_{i,j,k}^C=\mathbb{V}_3\frac{r_{i}^{U}e_j^{C}}{f_k^C}+c^{H}_k,	
	\end{aligned}}
\end{equation}
where $ c^{H}_k $ describes the hardware cost, which is assumed to be a constant and the same across all the VMs in $ s_k^C $ \cite{8815852}.
In addition, inherent requestors of $ s^C_k $ can lead to an imbalance between its demand and supply. Let $ \vartheta_k $ describe whether the resource demand of $ s_k^C $ exceeds its supply during practical transactions, as given by
\begin{equation}\label{key}	{\small
	\vartheta_k = \left\{ \begin{matrix}
		{0~~~, \sum_{s_j^E\in\varphi\left(s_k^C\right)}\sum_{b_i\in\mathbb{B}_{j,k}}\beta_{i,j,k}+\varepsilon_k \leq G_k^C} \\
		{1~~~, \sum_{s_j^E\in\varphi\left(s_k^C\right)}\sum_{b_i\in\mathbb{B}_{j,k}}\beta_{i,j,k}+\varepsilon_k > G_k^C} , \\
	\end{matrix} \right.}
\end{equation}
where $ \sum_{s_j^E\in\varphi\left(s_k^C\right)}\sum_{b_i\in\mathbb{B}_{j,k}}\beta_{i,j,k} $ denotes the number of VMs offering to ESs in set $ \varphi\left(s_k^C\right) $. When $ \vartheta_k=1 $, we use $ N_k $ to denote the number of inherent requestors who fail to enjoy cloud services, computed as
\begin{equation}\label{key}{\small
		\begin{aligned}
			N_k=\sum_{s_j^E\in\varphi\left(s_k^C\right)}\sum_{b_i\in\mathbb{B}_{j,k}}\beta_{i,j,k}+\varepsilon_k - G_k^C.
	\end{aligned}}
\end{equation}

\textbf{CS's utility and its expectation.} The utility of CS $ s_k^C $ trading with ESs consists of four key parts: \textit{i)} the overall payment from ESs (e.g., $ s_j^E\in\varphi\left(s_k^C\right) $) minus the cost for for offering services to them; \textit{ii)} the penalty from ES for breaking contracts (e.g., $ \beta_{i,j,k}=0 $ in a practical transaction); \textit{iii)} the payment from inherent requestors; and \textit{iv)} possible compensation to inherent requestors. Accordingly, we have
\begin{equation}\label{key}\small{
\begin{aligned}
	&u^C\left(s_k^C,\varphi\left(s_k^C\right)\right)=\\&\sum_{s_j^E\in\varphi\left(s_k^C\right)}\sum_{b_i\in\mathbb{B}_{j,k}}\left(\beta_{i,j,k}\left(p_{i,j,k}^{E\rightarrow C}-c_{i,j,k}^C\right)+(1-\beta_{i,j,k})q^{E\rightarrow C}\right)\\&+\vartheta_k[(\varepsilon_k-N_k)p_k^{I\rightarrow C}+N_kq_k^{C\rightarrow I}]+(1-\vartheta_k)\varepsilon_kp_k^{I\rightarrow C}.
\end{aligned}}
\end{equation}

Since the uncertainties bring challenges to obtain the practical value of (19), we consider its corresponding expectation instead, as given by
\begin{equation}\label{key}{\small
	\begin{aligned}
	&\overline{u^C}\left(s_k^C,\varphi\left(s_k^C\right)\right)=\\&\sum_{s_j^E\in\varphi\left(s_k^C\right)}\sum_{b_i\in\mathbb{B}_{j,k}}\left(\text{E}[\beta_{i,j,k}]\left(p_{i,j,k}^{E\rightarrow C}-c_{i,j,k}^C\right)+(1-\text{E}[\beta_{i,j,k}])q^{E\rightarrow C}\right)\\&+ \text{E}[\vartheta_k]\left\{( \text{E}[\varepsilon_k]- \text{E}[N_k])p_k^{I\rightarrow C}+ \text{E}[N_k]q_k^{C\rightarrow I}\right\}\\&+\left(1- \text{E}[\vartheta_k]\right) \text{E}[\varepsilon_k]p_k^{I\rightarrow C}
	\end{aligned}}
\end{equation}
where the derivations of $ \text{E}[N_k] $, $  \text{E}[\varepsilon_k] $, and $  \text{E}[\vartheta_k] $ are given in Appendix B.3.

\textbf{Risk assessment of CSs.} Due to the dynamic resource requirements from ESs and its inherent requestors, $ s_k^C $ also faces a risk reflected by the probability that $ s^C_k $'s resource demand exceeds its supply $ G_k^C $, is given by
\begin{equation}\label{key}{\small
		\begin{aligned}
			R^C\left(s_k^C,\varphi\left(s_k^C\right)\right)=\text{Pr}(\vartheta_k=1)
	\end{aligned}}
\end{equation}

Apparently, for any CS, $ R^C $ should be well estimated and controlled; otherwise, it will attend the spot market instead. 

\subsection{Key Definitions in OA-CLM Game}
Service trading in our designed futures market over CAMENs can be described as an OA-CLM game, comprising of two matching types: \textit{i)} MU-ES many-to-one (M2O) matching, and \textit{ii)} ES-CS many-to-many (M2M) matching, helping with guiding participants to sign forward contracts. Such matching mechanisms work good for handling the uncertainties (and thus control potential risks) in CAMENs, which differentiates them from conventional matching mechanisms (e.g., matching decisions are made based on the current network/market conditions). Our OA-CLM game can be formalized based on a series of key definitions given below.
\begin{Defn}(Many-to-one matching in the futures market)
	A M2O matching $ \omega $ in the futures market constitutes a mapping between MU set $ \bm{B} $ and ES set $\bm{{S}^{{E}}} $, while satisfying the following properties:
	
	\noindent
	$\bullet$ for each MU $ b_i\in\bm{B} $, $ \omega\left(b_i\right){\subseteq\bm{C_i}} $, $\bm{C_i}\subseteq\bm{{S}^{{E}}} $, $ \left|\omega\left(b_i\right)\right|=1 $;
	
	\noindent
	$\bullet$ for each ES $ s_j^E\in\bm{C_i} $, $ \omega\left(s_j^E\right)\subseteq\bm{B} $;
	
	\noindent
	$\bullet$ for MU $ b_i $ and ES $ s_j^E $, $ b_i\in\omega\left(s_j^E\right) $ if and only if $ \left\{s_j^E\right\}=\omega\left(b_i\right) $.
\end{Defn}

We next define the concept of \textit{blocking pair}\footnote{Blocking pairs refer to any two entities (such as individuals, schools, job positions, etc.) that are not matched to each other but would prefer to be paired together rather than remaining in their current matches. This concept is a critical factor in assessing the stability of a given set of matching solutions\cite{9154594}.}, representing a significant factor that may lead to instability of a M2O matching. 

\begin{Defn}(Blocking pair of MU-ES matching)
	Under a given M2O matching $ \omega $, MU $ b_i $ and ES $ s_j^E $ form a blocking pair $ \left(b_i;s_j^E\right) $, for which we consider two types.
	
	\textbf{Type 1 blocking pair}: Type 1 blocking pair satisfies the following condition:
	\begin{equation}\label{key}{\small
			\begin{aligned}
					\overline{u^{U\leftrightarrow E}}\left(s_j^E,\omega\left(s_j^E\right)\backslash \widetilde{\omega}\left(s_j^E\right)\cup \left\{b_i\right\}\right)>\overline{u^{U\leftrightarrow E}}\left(s_j^E,\omega\left(s_j^E\right)\right),
		\end{aligned}}
	\end{equation}
	which indicates that $ s_j^E $ can increase its expected utility by giving up some MUs, e.g., $ \widetilde{\omega}\left(s_j^E\right) $, while serving $ b_i $ instead.
	
	\textbf{Type 2 blocking pair}: Type 2 blocking pair satisfies the following condition:
	\begin{equation}\label{key}{\small
			\begin{aligned}
			\overline{u^{U\leftrightarrow E}}\left(s_j^E,\omega\left(s_j^E\right)\cup\left\{b_i\right\}\right)>\overline{u^{U\leftrightarrow E}}\left(s_j^E,\omega\left(s_j^E\right)\right),	
		\end{aligned}}
	\end{equation}
	which makes a matching unstable since $ s_j^E $ can serve more MUs under its resource constraint, to improve its expected utility.
\end{Defn}

The matching between ESs and CSs can be formalized relying on the following definitions.
\begin{Defn}(Many-to-many matching in the futures market)
	A M2M matching $ \varphi $ in the futures market denotes a mapping between ES set $ \bm{S^E} $ and CS set $ \bm{{S}^{ {C}}} $, while satisfying the following conditions:
	
	\noindent
	$\bullet$ for each CS $ s_k^C\in\bm{{S}^{ {C}}} $, $ \varphi\left(s_k^C\right){\subseteq\bm{S^E}} $;
	
	\noindent
	$\bullet$ for each ES $ s_j^E\in\bm{S^E} $, $ \varphi\left(s_j^E\right)\subseteq\bm{{S}^{ {C}}} $;
	
	\noindent
	$\bullet$ for CS $ s_k^C $ and ES $ s_j^E $, $ s_k^C\in\varphi\left(s_j^E\right) $ if and only if $ s_j^E\in\varphi\left(s_k^C\right) $.
\end{Defn}

We next introduce the concept of \textit{blocking coalition}, which represents a crucial factor that can make a M2M matching unstable.

\begin{Defn}(Blocking coalition of ES-CS matching)
	Given a M2M matching $ \varphi $, CS $ s_k^C $ and ES set $ \mathbb{S}\subseteq\bm{S^{E}}$ form a blocking pair $ \left(s_k^C;\mathbb{S}\right) $, for which we consider two types:
	
	\noindent	
	\textbf{Type 1 blocking coalition}: Type 1 blocking coalition can be incurred when the following conditions are met:
		\begin{equation}\label{key}{\small
			\begin{aligned}
				\overline{u^C}\left(s_k^C,\mathbb{S}\right)>\overline{u^C}\left(s_k^C,\varphi\left(s_k^C\right)\right),
		\end{aligned}}
	\end{equation}
	\begin{equation}\label{key}\small{
		\begin{aligned}
		\overline{u^{E\leftrightarrow C}}\left(s_j^E,\varphi\left(s_j^E\right)\backslash\widetilde{\varphi}\left(s_j^E\right)\cup\left\{s_k^C\right\}\right)>\overline{u^{E\leftrightarrow C}}\left(s_j^E,\varphi\left(s_j^E\right)\right).
		\end{aligned}}
	\end{equation}
	
	\noindent
	\textbf{Type 2 blocking coalition}: Type 2 blocking coalition can be incurred when the following conditions are met:
	\begin{equation}\label{key}{\small
			\begin{aligned}
				\overline{u^C}\left(s_k^C,\mathbb{S}\right)>\overline{u^C}\left(s_k^C,\varphi\left(s_k^C\right)\right),
		\end{aligned}}
	\end{equation}
		\begin{equation}\label{key}{\small
		\begin{aligned}
				\overline{u^{E\leftrightarrow C}}\left(s_j^E,\varphi\left(s_j^E\right)\cup\left\{s_k^C\right\}\right)>\overline{u^{E\leftrightarrow C}}\left(s_j^E,\varphi\left(s_j^E\right)\right).
	\end{aligned}}
\end{equation}
\end{Defn}
It can be inferred that a Type 1 blocking coalition can result in the instability of a matching, since an ES is motivated to select a different set of CSs to achieve a higher expected utility. Similarly, a Type 2 blocking coalition can also introduce instability to a matching since an ES can engage with more CSs to enhance its expected utility.

\subsection{Problem Formulation}
The key purpose of the three parties in our designed futures market is \textit{to maximize their individual expected utilities under acceptable risks}. Accordingly, each MU $ b_i\in\bm{B} $ aims to maximize its overall expected utility, given as
\begin{subequations}{\small
	\begin{align}
	\bm{\mathcal{F}^U}:~&\underset{{\omega\left(b_i\right)}}{\max}~\overline{u^U}\left(b_i,\omega\left(b_i\right)\right)\tag{28}\\
\text{s.t.}~~~
&\omega\left(b_i\right)=\left\{\left\{s_j^E\right\},\varnothing \right\},~\forall s_j^E\subseteq \bm{C_i}\tag{28a}\\
&p_{i,j}^{U\rightarrow E}\le \text{E}[v_{i,j}],~\text{if}~\omega\left(b_i\right)=\left\{s_j^E\right\}\tag{28b}\\
&R^{U}_1\left( b_{i},\omega \left( b_{i} \right) \right)\leq \rho_1\tag{28c}\\
&R^{U}_2\left( b_{i},\omega \left( b_{i} \right) \right)\leq \rho_2\tag{28d}
\end{align}}
\end{subequations}
where $\rho_1$ and $ \rho_2 $ are risk thresholds, and $\rho_1,~\rho_2\in (0, 1] $. In the optimization problem $ \bm{\mathcal{F}^U} $, constraint (28a) guarantees that a MU can only be matched to one ES, 
constraint (28b) ensures that the obtained expected valuation of $ b_i $ benefit from $ s^E_j $ can cover its corresponding payment; 
constraints (28c) and (28d) evaluate and control the risks of $ b_i $, and the derivations of the two risks constraints are conducted in Appendix B.1. Moreover, in the considered futures market, each ES $ s_j^E\in\bm{{S}^{E}} $ prefers to maximize its overall expected utility, as modeled by the following optimization problem $ \bm{\mathcal{F}^E} $:
\begin{subequations}\label{key}{\small
\begin{align}
	\bm{\mathcal{F}^E}:~~&\underset{{\omega\left(s_j^E\right),\varphi\left(s_j^E\right)}}{\max} \overline{u^E}\left(s_j^E,\omega\left(s_j^E\right),\varphi\left(s_j^E\right)\right)\tag{29}\\
	\text{s.t.}~~~~&
	\omega\left(s_j^E\right)\subseteq \bm{B}\tag{29a}\\
	&\varphi\left(s_j^E\right)\subseteq \bm{S^{C}}, ~\text{if} ~x_j>0\tag{29b}\\
	&p_{i,j}^{U\rightarrow E}\geq c_{i,j}^E, ~\forall b_i \in \omega\left(s_j^E\right)\tag{29c}\\
	&p_{i,j,k}^{E\rightarrow C}\le p_{j}^{max},~\forall s_k^C \in \varphi\left(s_j^E\right),~\text{if}~x_j>0\tag{29d}\\
	&|\omega\left(s_j^E\right)|\le\nonumber\\&(1+\tau)\left(G_j^E+\sum_{s_k^C\in\varphi\left(s_j^E\right)} |\mathbb{B}_{j,k}|\right)\le (1+\tau) K_j^E \tag{29e}\\
	&R_1^E\left(s_j^E,\omega\left(s_j^E\right)\right) \leq\rho_3\tag{29f}\\
	&R_2^E\left(s_j^E,\omega\left(s_j^E\right),\varphi\left(s_j^E\right)\right) \leq \rho_4\tag{29g}
\end{align}}
\end{subequations}
Since we are unaware of the specific tasks that can further be offloaded to CSs for process at this time, we use $ p_j^{max}=\text{min} \left\{p_{i,j}^{U\rightarrow E}~|~\forall b_i \in \omega\left(s_j^E\right)\right\} $ to estimate the maximum payment (per offloaded task) from $ s_j^E $ to a CS, which conforms the rule that the expense of $ s^E_j $ offered to $ s_k^C $ can fall below the payment from $b_i$.
In problem $ \bm{\mathcal{F}^E} $, $\rho_3$ and $ \rho_4 $ are risk thresholds within interval $ (0, 1] $. Constraints (29a) and (29b) show the affiliation of $ \omega\left(s_j^E\right) $ and $ \varphi\left(s_j^E\right) $, respectively; constraints (29c) and (29d) ensure that the payment from each MU to $ s_j^E $, and the expense of $ s^E_j $ offered to $ s_k^C $ can cover the corresponding service costs; constraint (29e) considers the resource supply of $ s^E_j $ for serving MUs in $ \omega\left(s_j^E\right) $ should not exceed its overbooked resources, while below the number of MUs that can access to it simultaneously (e.g., $  K_j^E$); constraints (29f) and (29g) management the risks that an ES may confront, and the derivations of such risks are conducted in Appendix B.2.
Then, similar to MUs and ESs, the goal of each CS $ s_k^C\in\bm{{S}^{{C}} }$ is to maximize its expected utility, as described by $ \bm{\mathcal{F}^C} $: 
\begin{subequations}\label{key}{\small
\begin{align}
	\bm{\mathcal{F}^C}:~~&\underset{{\varphi\left(s_k^C\right)}}{\max}~ \overline{u^C}\left(s_k^C,\varphi\left(s_k^C\right)\right)\tag{30}\\
\text{s.t.}~~~~&
	\varphi\left(s_k^C\right){\subseteq\bm{{S}^{E}}}\tag{30a}\\
	&p_{i,j,k}^{E\rightarrow C}\geq c_{i,j,k}^C,~\forall s_j^E \in \varphi\left(s_k^C\right)\tag{30b}\\
&R^C\left(s_k^C,\varphi\left(s_k^C\right)\right) \leq \rho_5 \tag{30c}
\end{align}}
\end{subequations}
where $\rho_5$ denotes a risk threshold falls in interval $ (0, 1] $. In $ \bm{\mathcal{F}^C} $, constraint (30a) guarantees $ \varphi\left(s_k^C\right) $ belongs to set $ \bm{{S}^{{E}}} $; constraint (30b) helps CS $ s_k^C $ to enjoy a non-negative expected utility; and constraint (30c) constrains the risk that $ s_k^C $'s resource demand may exceeds its resource supply, and the derivation of (30c) is studied in Appendix B.3.

\subsection{Algorithm Design}
\begin{algorithm*}[t!] \small{ 
		\setstretch{0.1} 
	\caption{Proposed OA-CLM in the futures market}
	\LinesNumbered 
	\textbf{Initialization:} $ m\leftarrow1 $, $ n\leftarrow1 $, $ p_{i,j}^{U\rightarrow E}\left\langle 1 \right\rangle\leftarrow p_{i,j}^{min}$, $ p_{i,j,k}^{E\rightarrow C}\left\langle 1 \right\rangle\leftarrow p_{i,j,k}^{min} $, for $\forall i,j,k$, $ flag_i\leftarrow 1 $, $ flag_j\leftarrow 1 $ \ 
	
	\textit{\textbf{\% Phase 1. MU-ES M2O matching game}}
	
	\While{$ {flag}_{i} $}{
		\textbf{Calculate} $ L^U_{i}$, where $R_1^U<\rho_1,~R_2^U<\rho_2 $
		
		\textbf{Determine} $ \mathbb{W}(b_i)\leftarrow $ choose the top ES in $ L_i^U $
		
		\textbf{$ {flag}_{i} \leftarrow 0 $}
		
		\If{$ \forall\mathbb{W}\left( b_{i} \right) \neq \varnothing $}{
			\For{$ 
				s^E_{j} \leftarrow \mathbb{W}\left( b_{i} \right) $}{$ b_{i} $ sends a proposal including $ p_{i,j}^{U\rightarrow E}\left\langle m \right\rangle $ and $ r^U_{i} $ to $ s^E_j $}
			\While{$ \Sigma_{b_{i}\in\bm{B}}{flag}_{i} > 0 $}{
				Collect proposals from the MUs in $ \bm B $, e.g., using $ \widetilde{\mathbb{W}}\left(s_j^E\right) $ to include the MUs that send proposals to $ s_j^E $
				
				$ \mathbb{W}\left(s_j^E\right)\leftarrow $ choose MUs from $ \widetilde{\mathbb{W}}\left(s_j^E\right) $ to maximize the ES's expected utility under limited $ (1+\tau)K_j^E $ VMs constraint and risks constraints $ R_2^E $
				
				$ s_j^E $ temporally accept the MUs in $ \mathbb{W}\left(s_j^E\right) $, and rejects the others, and notify each MU that whether be selected and the probability of being selected as a volunteer in this round
			}
			\For{
				$ b_i \in \mathbb{W}\left(s_j^E\right) $
			}{
				\If{$ b_i $ is rejected by $ s_j^E $, $ p_{i,j}^{U\rightarrow E}\left\langle m \right\rangle<v_{i,j} $, and $ R_1^U<\rho_1 $ }{
					$ p_{i,j}^{U\rightarrow E}\left\langle m+1 \right\rangle\leftarrow \text{min}\left\{p_{i,j}^{U\rightarrow E}\left\langle m \right\rangle+\mathrm{\Delta}p_i,\text{E}[v_{i,j}]\right\} $}
				\Else{$ p_{i,j}^{U\rightarrow E}\left\langle m+1 \right\rangle\leftarrow p_{i,j}^{U\rightarrow E}\left\langle m \right\rangle $}
			}
			
			\If{there exists $
				p_{i,j}^{U\rightarrow E}\left\langle m+1 \right\rangle\neq p_{i,j}^{U\rightarrow E}\left\langle m \right\rangle $, $\forall b_{i} \in \mathbb{W}\left(s_j^E\right) $}{
				$ {flag}_i\leftarrow 1 $,
				$ m\leftarrow m+1 $
			}
		}
	}
	\textit{\textbf{\% Phase 2. ES-CS M2M matching game}}
	
	\While{$ flag_j $}{
		\textbf{Calculate} $ L^E_{i,j}$\\
		\textbf{Determine} $ \mathbb{Y}\left(s_j^E\right)\leftarrow$ choose each task's interested CSs from $ L^E_{i,j} $; \\
		\textbf{Calculate} the probability of ES $ s_j^E $ fulfills the contract $  \text{E}[\beta_{i,j,k}] $ to $ s_k^C $
		
		\textbf{$ {flag}_{j} \leftarrow 0 $}\
		
		\If{$s_j^E\in \left\{s_j^E~|~x_j>0,~\forall s_j^E\in \bm{{S}^{{E}}}\right\} $}{
			\For{$ s_k^C\in \mathbb{Y}\left(s_j^E\right) $}{$ s_j^E $ sends its information (e.g., payment for each task $ p_{i,j,k}^{E\rightarrow C}\left\langle n \right\rangle $, probability of fulfilling the contract $  \text{E}[\beta_{j,k}] $, and required amount of resource $ r_j^{max} $) to $ s_k^C \in \mathbb{Y} \left(s_j^E\right) $}

		\While{$ \Sigma_{s^E_{j}\in\bm{S^{E}}}{flag}_{j} > 0 $}{
			Collect proposals from the ESs in $ \left\{s_j^E~|~x_j>0,~\forall s_j^E\in \bm{{S}^{{E}}}\right\} $, e.g., using $ \widetilde{\mathbb{Y}}\left(s_k^C\right) $ to include the ESs' tasks that send proposals to $ s_k^C $\\
			$ \mathbb{Y} \left(s_k^C\right)\leftarrow $ choose ESs' tasks from $ \widetilde{\mathbb{Y}} \left(s_k^C\right) $ to maximize the CS's expected utility under limited $ G_k^C $ VMs constraint and risk constraint $ R^C $\\
			$ s_k^C $ temporally accept the ESs in $ \mathbb{Y} \left(s_k^C\right) $, and rejects the others
		}
		\For{
			$ s_j^E \in \mathbb{Y} \left(s_k^C\right) $
		}{
			\If{a task from $ s_j^E $ is rejected by $ s_k^C $ and $ p_{i,j,k}^{E\rightarrow C}\left\langle n \right\rangle<p_{j}^{max} $}{
				$ p_{i,j,k}^{E\rightarrow C}\left\langle n+1 \right\rangle\leftarrow \text{min}\left\{p_{i,j,k}^{E\rightarrow C}\left\langle n \right\rangle+\mathrm{\Delta}p_j,p_{j}^{max}\right\} $}
			\Else{$ p_{i,j,k}^{E\rightarrow C}\left\langle n+1 \right\rangle\leftarrow p_{i,j,k}^{E\rightarrow C}\left\langle n \right\rangle $}
		}
		
		\If{there exists $
			p_{i,j,k}^{E\rightarrow C}\left\langle n+1 \right\rangle\neq p_{i,j,k}^{E\rightarrow C}\left\langle n \right\rangle $, $\forall s_k^C \in \mathbb{Y} \left(s_j^E\right) $}{
			$ {flag}_j\leftarrow 1 $,
			$ n\leftarrow n+1 $
		}
	}
	}
\% \textit{\textbf{Phase 3. Cross-layer interaction}}

\If{$ |\mathbb{W}\left(s_j^E\right)|>(1+\tau)\left(G_j^E+\sum_{s_k^C\in\varphi\left(s_j^E\right)} |\mathbb{B}_{j,k}|\right),~\forall s_j^E\in \bm{{S}^{{E}}} $}{$ \mathbb{W} \left(s_j^E\right)\leftarrow $ $ s_j^E $ selects some MUs from $ \mathbb{W}\left(s_j^E\right) $ to maximize its expected utility based on ES's overbooked resource supply $ (1+\tau)\left(G_j^E+\sum_{s_k^C\in\varphi\left(s_j^E\right)} |\mathbb{B}_{j,k}|\right) $} 
\If{The risks of participants(i.e., MU, ES, and CS) are unacceptable }{The corresponding participant will give up the futures trading }
$\omega\left(s_j^E\right)\leftarrow\mathbb{W}\left(s_j^E\right)$, $\omega(b_i)\leftarrow\mathbb{W}(b_i)$, $\varphi\left(s_j^E\right)\leftarrow\mathbb{Y} \left(s_j^E\right)$, $\varphi\left(s_k^C\right)\leftarrow\mathbb{Y} \left(s_k^C\right)$\\
\textbf{Return:} $\omega\left(s_j^E\right)$, $\omega(b_i)$, $\varphi\left(s_j^E\right)$, $\varphi\left(s_k^C\right)$
}
\end{algorithm*}

Our proposed OA-CLM facilitates two-way negotiations involving MUs, ESs, and CSs, where the trading between MUs and ESs, as well as between ESs and CSs, can significantly influence each other, resulting in an iterative, cross-layer, and intricate trading process. 
To provide a clearer illustration of how the OA-CLM game operates, the entire matching procedure involves three key phases with multiple rounds. It specific details can be found in Algorithm 1.

\textit{i) MU-ES matching game (Phase 1):} We first introduce the M2O matching between MUs and ESs and adopt the Gale-Shapley algorithm\cite{9667258,9687261} to construct a stable mapping in the considered futures market.

\noindent
\textbf{Step 1. Initialization:} At the beginning of each round, each MU $ b_i $ sets its payment by $ p_{i,j}^{U\rightarrow E}\left\langle 1 \right\rangle=p_{i,j}^{min}$ (line 1), where $ p_{i,j}^{min} $ denotes the initial payment from $ b_i $ to $ s_j^E $. In addition, each MU announces its requests to ESs according to its \textit{preference list} (Definition 5).
\begin{Defn}(Preference list of MU) The preference list $ L_{i}^U $ of a MU $ b_i $ regarding ESs represents a vector of $ s^E_j \in \bm{C_i} $, sorted by the expected value of $ {u^U}\left(b_i,s_j^E\right) $ under a non-ascending order: 
\begin{equation}\label{key}{\small
	\begin{aligned}
				L_{i}^U= \left\{s_j^E~|~\text{{\rm non-ascending on }} \overline{u^U}\left(b_i,s_j^E\right),\forall s_j^E\in \bm{C_i}\right\},	
	\end{aligned}}
\end{equation}
\end{Defn}
\noindent
where we use $\mathbb{W}(b_i) \in L_i^U$ to represent the favorite ES of $ b_i $ (e.g., $|\mathbb{W}(b_i)|=1$), and $\mathbb{W}\left(s_j^E\right)$ to indicate the set of MUs that are temporarily accepted by $ s_j^E $ during the matching procedure.

\noindent
\textbf{Step 2. Proposal of MUs:} At round $m$, each MU $ b_i $ chooses the top ES in $ L_i^U $ under acceptable risks, and records it in $\mathbb{W}(b_i)$ (line 5). Then, $ b_i $ sends a proposal to the ES in $\mathbb{W}(b_i)$ (denoted by $ s^E_j $ for analytical simplicity), including its payment $ p_{i,j}^{U\rightarrow E}\left\langle m \right\rangle $ and the required amount of resource $ r_i^U $ (line 9).

\noindent
\textbf{Step 3. MU selection on ESs' side:} We use set $\widetilde{\mathbb{W}}\left(s_j^E\right)$ to collect the information from MUs, each ES $ s_j^E $ then determines a collection of temporarily MUs, as recorded in set $\mathbb{W}\left(s_j^E\right)$, where $\mathbb{W}\left(s_j^E\right)\subseteq \widetilde{\mathbb{W}}\left(s_j^E\right)$ that enable the maximum expected utility under the overbooked maximum access $(1+\tau) K_j^E $ (e.g., constraint (29e)). Then, each $ s_j^E $ informs MUs in set $\widetilde{\mathbb{W}}\left(s_j^E\right)$ about its determinations, as well as the probability on being selected as a volunteer, in the current round (lines 10-13).

\noindent
\textbf{Step 4. Decision-making on MUs' side:} After obtaining decisions from ES $ s_j^E \in \mathbb{W}(b_i) $, $ b_i $ considers the following conditions:

\textbf{Condition 1.} The payment from $ b_i $ remains unchanged, when one of the following conditions is met (line 18): \textit{i)} $ b_i $ is accepted by an ES $ s_j^E $; \textit{ii)} $ b_i $'s current payment $ p_{i,j}^{U\rightarrow E}\left\langle m \right\rangle $ equals to its expected valuation $ \text{E}[v_{i,j} ]$; \textit{iii)} risks are turning intolerable control, e.g., raising the payment can lead to an unacceptable risk (i.e., fails to meet constraint (28c));

\textbf{Condition 2.} If $ b_i $ is rejected by an ES $ s_j^E $, its current payment $ p_{i,j}^{U\rightarrow E}\left\langle m \right\rangle $ stays below its expected valuation $ \text{E}[v_{i,j}] $, and the corresponding risk $ R_1^U $ is acceptable, $ b_i $ will put up its payment to $ s_j^E $ in the next round (line 16).

\noindent
\textbf{Step 5. Repeat:} If payments of all the MUs stay unchanged from the $ (m-1)^\text{th} $ round to the $ m^\text{th} $ round, the matching will be terminated at round $ m $ (e.g., $ \Sigma_{b_{i}\in\bm{B}}{flag}_{i} = 0 $, line 6). Otherwise, OA-CLM repeats the above steps (e.g., lines 3-20) in the next round.

\textit{ii) ES-CS matching game (Phase 2):}
We next introduce the M2M matching between ESs and CSs and adopt the Gale-Shapley algorithm\cite{9667258,9687261}, contributing to construct a stable mapping for purchasing cloud resources in our futures market, as follows:

\noindent
\textbf{Step 1. Initialization:} At the beginning of each round, each ES $ s_j^E $ sets its payment for a task as $ p_{i,j,k}^{E\rightarrow C}\left\langle 1 \right\rangle=p_{i,j,k}^{min}$, where $p_{i,j,k}^{min}$ denotes the initial payment from ES $ s_j^E $. We apply $\mathbb{Y}\left(s_k^C\right)$ to indicate the tasks that are temporarily accepted by CS $ s_k^C $, and $\mathbb{Y}\left(s_j^E\right)$ to represent CSs that are interested in the tasks of ES $s_j^E$, based on each task's (namely, MU's) preference list, according to the following definition.
\begin{Defn}(Preference list of ES) The preference list $ L_{i,j}^E $ of an ES $s_j^E$ in handling the task of $b_i$ regarding CSs (e.g., which CS can the task of $b_i$ be offloaded to for further processing, from $s_j^E$) is a vector of $ s^C_k \in \bm{S^C} $, sorted by the expected value of $ {u^{E\leftrightarrow C}}\left (s_j^E,\varphi\left(s_j^E\right)\right ) $ and following a non-ascending order: 	
	\begin{equation}\label{key}{\small
		\begin{aligned}
	&L_{i,j}^E=
	\\ & \left\{s_k^C~|~\text{{\rm non-ascending on }} \overline{u^{E\leftrightarrow C}}\left (s_k^C,\varphi\left(s_j^E\right)\right ),\forall s^C_k \in \bm{S^C}\right\}.
\end{aligned}}
\end{equation}
\end{Defn}

\noindent
\textbf{Step 2. Proposal of ESs:} At round $ n $, each ES $ s_j^E \in \left\{s_j^E~|~x_j>0,~\forall s_j^E\in \bm{{S}^{{E}}}\right\}$ reports its information, including each task's payment $ p_{i,j,k}^{E\rightarrow C}\left\langle n \right\rangle $, the corresponding required amount of resources $r_i^U $, and the probability on breaking the contract with $ s_k^C $ (i.e., $\text{Pr}\left(\beta_{i,j,k}=0\right)$). Since we are unaware of the specific tasks that can be offloaded to CSs for processing during this time, let $ r^{max}_j=\max \left\{r_i^U~|~\forall b_i \in \omega \left(s_j^E\right)\right\} $ to estimate the maximum amount of required resources (per offloaded task) from $ s_j^E $ to a CS (line 29). 

\noindent
\textbf{Step 3. ES selection on CSs' side:} After collecting the information from ESs in set $\widetilde{\mathbb{Y}}\left(s_k^C\right)$, each CS $ s_k^C $ determines a collection of temporarily ESs' tasks, as recorded by set $\mathbb{Y}\left(s_k^C\right)$, where $\mathbb{Y}\left(s_k^C\right)\subseteq \widetilde{\mathbb{Y}}\left(s_k^C\right)$, that enables the maximum expected utility under limited resource supply (e.g., $ G_k^C $) while supporting an acceptable risk $ R^C $. Then, each $ s_k^C $ informs ESs about its determinations in the current round (lines 30-33).

\noindent
\textbf{Step 4. Decision-making on ESs' side:} After obtaining decisions from each CS $ s_k^C \in \mathbb{Y}\left(s_j^E\right) $, ES $ s_j^E $ considers the following conditions:

\textbf{Condition 1.} If the task of a MU associated with ES $ s_j^E $ is accepted by $ s_k^C $, or its current payment $ p_{i,j,k}^{E\rightarrow C}\left\langle n \right\rangle $ equals to its maximum payment $ p_{j}^{max} $ (e.g., $ p_{j}^{max}=\text{min} \left\{p_{i,j}^{U\rightarrow E}\left\langle m \right\rangle | b_i\in \mathbb{W}\left(s_j^E\right)\right\} $), the payment from $ s_j^E $ remains unchanged, i.e., $ p_{i,j,k}^{E\rightarrow C}\left\langle n+1 \right\rangle=p_{i,j,k}^{E\rightarrow C}\left\langle n \right\rangle $
(line 38);

\textbf{Condition 2.} If the task of a MU associated with $ s_j^E $ is rejected by $ s_k^C $ and its current payment $ p_{i,j,k}^{E\rightarrow C}\left\langle n \right\rangle $ is still below its maximum payment $ p_{j}^{max} $, $ s_j^E $ increases its payment to $ s_k^C $ in the next round (line 36).

\noindent
\textbf{Step 6. Repeat:} If payments of all the ESs stay unchanged from the $ (n-1)^\text{th} $ round to the $ n^\text{th} $ round, the matching will be terminated at round $ n $ (e.g., $ \Sigma_{s^E_{j}\in\bm{S^{E}}}{flag}_{j} = 0 $, line 26). Otherwise, OA-CLM repeats the above steps (e.g., lines 22-40) in the next round.

\textit{iii) Cross-layer interaction (Phase 3):}
This step describes the mutual impacts between the above-discussed two phases.
When ES $ s_j^E $ did not purchase sufficient resources from CSs in Phase 2, i.e., $ |\mathbb{W}\left(s_j^E\right)|>(1+\tau)\left(G_j^E+\sum_{s_k^C\in\varphi\left(s_j^E\right)} |\mathbb{B}_{j,k}|\right)$, $ s_j^E $ will further choose some MUs from $ \mathbb{W}\left(s_j^E\right) $ to maximize its expected utility based on its overbooked resource supply, e.g., $ (1+\tau)\left(G_j^E+\sum_{s_k^C\in\varphi\left(s_j^E\right)} |\mathbb{B}_{j,k}|\right) $ (lines 42-43), which makes the previous two phases mutually impact each other (i.e., the quantity of matched MUs for each ES during Phase 1 determines the resource demand for each ES in Phase 2, whereas the quantity of cloud resources acquired by an ES in Phase 2 directly influences the resource supply of that ES in Phase 1). Our OA-CLM facilitates a risk-aware futures market, where participants will later engage in spot market when incurring unacceptable risks (lines 44-45). 

\subsection{Design Targets and Property Analysis}
As our OA-CLM game is pre-deployed for future practical resource transactions, our focus lies in crafting distinctive objectives that diverge from conventional matching games, considering the expected utility for participants as well as their associated risks.

\begin{Defn}(Individual rationality of MU-ES matching in futures market) For both MUs and ESs, a matching $ \omega $ is individual rational when the following conditions are satisfied:
	
	\noindent
	$\bullet$ for MUs: the risks of MU $ b_i\in\bm{B} $ on obtaining an unsatisfying utility, and being selected as a volunteer, are controlled within certain ranges, i.e., constraints (28c) and (28d) are satisfied.
	
	\noindent
	$\bullet$ for ESs: each $ s_j^E $ matched to a MU set $ \omega\left(s_j^E\right) $ can achieve a positive expected utility, i.e.,
			\begin{equation}\label{key}{\small
			\begin{aligned}
				\overline{u^{U\leftrightarrow E}}	\left(s_j^E,\omega\left(s_j^E\right)\right) > 0.
		\end{aligned}}
	\end{equation}
	
	The risks of each ES on breaking the pre-signed contracts are controlled within certain ranges, i.e., constraints (29f) and (29g) are satisfied.
	Moreover, the overall resource demand of $ s_j^E $ will not exceed its supply, as given by
\begin{equation}\label{key}{\small
	\begin{aligned}
		|\omega\left(s_j^E\right)| \leq (1+\tau)\left(G_j^E+\sum_{s_k^C\in\varphi\left(s_j^E\right)} |\mathbb{B}_{j,k}|\right).		
	\end{aligned}}
\end{equation}
\end{Defn}

\begin{Defn}(Individual rationality of ES-CS matching) For both ESs and CSs, a matching $ \varphi $ is individual rational when the following conditions are satisfied:
	
	\noindent
	$\bullet$ for CSs: each CS $ s_k^C $ matched to an ES set $ \varphi\left(s_k^C\right) $ can achieve a positive expected utility, i.e.,
	\begin{equation}\label{key}{\small
		\begin{aligned}
		\overline{u^C}\left(s_k^C,\varphi\left(s_k^C\right)\right)>0		
		\end{aligned}}
	\end{equation}
	
	Then, the risk of each CS that its resource demand exceeds supply can be controlled within a certain range, i.e., constraints (30c) is satisfied.
	
	\noindent
	$\bullet$ for ESs: the risk of each ES on breaking the contracts with CSs is controlled within a certain range, i.e., constraint (29f). Meanwhile, each ES can achieve a positive expected utility, i.e.,
	\begin{equation}\label{key}{\small
		\begin{aligned}
				\overline{u^E}	\left(s_j^E,\omega\left(s_j^E\right),\varphi\left(s_j^E\right)\right)>0
		\end{aligned}}
	\end{equation}
\end{Defn}

\begin{Defn}(Strong stability of OA-CLM) The proposed OA-CLM is strongly stable if MU-ES matching and ES-CS matching are individually rational and have no blocking pair or coalition.
\end{Defn}

Note that competitive equilibrium represents a conventional concept in economic behaviors, which plays an important role in analyzing the performance of commodity markets upon having flexible prices and multiple players\cite{8815852}. When the considered market arrives at the competitive equilibrium, there exists a price at which the number of MUs that will pay is equal to the number of ESs that will sell\cite{8815852}. Correspondingly, the competitive equilibrium in futures market is defined below.

\begin{Defn}(Competitive equilibrium associated with trading between MUs and ESs in OA-CLM) The trading between MUs and ESs reaches a competitive equilibrium if the following conditions are satisfied:
	
	\noindent
	$\bullet$ For each ES $ s_j^E \in \bm S^{E} $, if $ s_j^E $ is associated with a MU $ b_i\in \bm B $, then $ c_{i,j}^E\leq p_{i,j}^{U\rightarrow E} $,
	
	\noindent
	$\bullet$ For each MU $ b_i\in \bm B $, $ b_i $ is willing to trade with the ES that can bring it with the maximum expected utility,
	
	\noindent
	$\bullet$ For each MU $ b_i\in \bm B $, if $ b_i $ is not associated with any ES, then the payment paid by $ b_i $ is equal to its expected valuation $ \text{E}[v_{i,j}] $ or the maximum payment $ p_{i,j}^{max} $ under accepted risk $ R_1^U $ through trading with $ s_j^E $.
\end{Defn}

\begin{Defn}(Competitive equilibrium associated with trading between ESs and CSs in OA-CLM) The trading between ESs and CSs reaches a competitive equilibrium if the following conditions are satisfied:
	
	\noindent
	$\bullet$ For each CS $ s_k^C \in \bm S^{C} $, if $ s_k^C $ is associated with an ES $ s_j^E\in \bm{S^E} $, then $ c_{i,j,k}^C\leq p_{i,j,k}^{E\rightarrow C} $,
	
	\noindent
	$\bullet$ For each ES $ s_j^E\in \bm{S^E} $, $ s_j^E $ is willing to trade with the CS that can bring it with the maximum expected utility,
	
	\noindent
	$\bullet$ For each ES $ s_j^E\in \bm{S^E} $, if a task of $ s_j^E $ is not associated with any CS, then $ p_{i,j,k}^{E\rightarrow C} = p_{j}^{max} $.
\end{Defn}

For a multi-objective optimization problem (e.g., optimization problems $ \bm{\mathcal{F}^U} $, $ \bm{\mathcal{F}^E} $, and $ \bm{\mathcal{F}^C} $), a Pareto improvement occurs when the \textit{expected social welfare} can be increased with another feasible matching result. Specifically, the expected social welfare refers to a summation of expected utilities of MUs, ESs, and CSs in our designed CAMENs. Thus, a matching is weak Pareto optimal when there is no Pareto improvement \cite{8815852}.

\begin{Defn}(Weak Pareto optimality of OA-CLM) The proposed OA-CLM is weak Pareto optimal if there is no Pareto improvement.
\end{Defn}

We show that our proposed OA-CLM can support the above-discussed properties, while the corresponding analysis and proofs are given by Lemmas 1-6 in Appendix C, due to space limitation.

\section{OS-CLM Game in Spot Market}
To cope with the dynamics of our considered market (e.g., fluctuant resource demand/supply and varying communication qualities), while ensuring efficient and dependable resource provisioning, we then engage two cases during each practical resource transaction: \textit{i)} When the total resource demand of an ES surpasses its supply, including resources borrowed from CSs as specified in forward contracts, the ES identifies certain contractual MUs as volunteers to relinquish the committed services, to meet the resource constraint; \textit{ii)} In cases where certain MUs without pre-signed forward contracts have resource requirements, while some ESs or CSs have surplus resources to offer, the OS-CLM mechanism is designed as a complementary plan to assist participants in achieving improved utilities.

To distinguish the MUs, ESs, and CSs that practically participate in spot trading market, we first introduce the following notations. 

\noindent
$\bullet$ $ \bm{B^{\prime}} $: the set of MUs that need to trade with ESs to get resources in spot market, including MUs those without forward contracts and the volunteers, i.e., $\bm{B^{\prime}}\subseteq\bm{B} $;

\noindent
$\bullet$ $ \bm{{S}^{{E}\prime}} $ and $ \bm{{S}^{ {C\prime}}} $: the set of ESs, and the set of CSs with idle resources in spot market, i.e., $ \bm{{S}^{{E}\prime}}\subseteq\bm{{S}^{{E}}} $ and $ \bm{{S}^{{C}\prime}}\subseteq\bm{{S}^{{C}}} $;
%

\noindent
$\bullet$ $ \omega^{\prime}\left(s_j^E\right) $: the set of MUs served by ES $ s_j^E $ in spot market, where $ \omega^{\prime}\left(s_j^E\right)\subseteq\bm{B^{\prime}} $;

\noindent
$\bullet$ $ \omega^{\prime}\left(b_i\right) $: an ES that serves MU $ b_i $ in spot market, where $ \omega^{\prime}\left(b_i\right)\subseteq\bm{ C_i^{\prime}} $, and $ \bm{ C_i^{\prime}} \subseteq\bm{{S}^{{E}\prime}} $, $ \bm{C_i^{\prime}} $ represents the set of candidate ESs in $ \bm{{S}^{{E}\prime}} $ that $ b_i $ may get services;

\noindent
$\bullet$ $ \varphi^{\prime}\left(s_j^E\right) $: the set of CSs that lend resources to ES $ s_j^E $ in spot market, where $ \varphi^{\prime}\left(s_j^E\right)\subseteq\bm{{S}^{{C}\prime}} $;

\noindent
$\bullet$ $ \varphi^{\prime}\left(s_k^C\right) $: the set of ESs that borrow resources from CS $ s_k^C $ in spot market, where $ \varphi^{\prime}\left(s_k^C\right)\subseteq\bm{{S}^{{E}\prime}} $.

	In the spot market, we then define the practical utility of MU $ b_i\in \bm{B^\prime} $ as the difference between its practically obtained valuation and its payment to an ES, as given by
\begin{equation}\label{key}{\small
	\begin{aligned} 
		&u^{U\prime}(b_i,\omega^{\prime}\left(b_i\right))=v_{i,j}-p_{i,j}^{U\prime\rightarrow E\prime}.
	\end{aligned} }
\end{equation}

Similarly, the utility of ES $ s_j^E \in \bm{{S}^{{E}\prime}} $ can be calculated as its total received payments minus service costs and its payment to CSs, as shown by
\begin{equation}\label{key}{\small
	\begin{aligned}
		&u^{E\prime}\left(s_j^E,\omega^{\prime}\left(s_j^E\right),\varphi^{\prime}\left(s_j^E\right)\right)
		=\sum_{b_i\in\omega^{\prime}\left(s_j^E\right)}\left(p_{i,j}^{U\prime\rightarrow E\prime}-c_{i,j}^E\right)\\&-\sum_{s_k^C\in\varphi^{\prime}\left(s_j^E\right)}\sum_{b_i\in\mathbb{B}^\prime_{j,k}}\left(p_{i,j,k}^{E\prime\rightarrow C\prime}-c_{i,j}^E\right), 
	\end{aligned}}
\end{equation}
where $ \mathbb{B}^\prime_{j,k} $ refers to the set of MUs that can enjoy computing services offered by $ s_j^E $, for further purchasing resources from $ s_k^C $ in spot market.
In addition, the utility of CS is determined by the overall payment obtained from ESs minus its service cost:
\begin{equation}\label{key}{\small
	\begin{aligned}
		&u^{C\prime}\left(s_k^C,\varphi^{\prime}\left(s_k^C\right)\right)=\sum_{s_j^E\in\varphi^{\prime}\left(s_k^C\right)}\sum_{b_i\in\mathbb{B}^\prime_{j,k}}\left(p_{i,j,k}^{E\prime\rightarrow C\prime}-c_{i,j,k}^C\right).
	\end{aligned}}
\end{equation}

 Note that each MU, ES, and CS aims to \textit{maximize its overall practical utility in the designed spot market}. Since the OS-CLM game is similar to OA-CLM in some aspects, the optimization problems, the relevant definitions of matching $ \omega^{\prime}(.) $ and $ \varphi^{\prime}(.) $ as well as the algorithm design, the design targets and their corresponding analysis as well as algorithm design of OS-CLM are detailed by Appendix D.

\section{Evaluation}
\subsection{Simulation Settings}
To verify the superior performance of our proposed hybrid cross-layer matching game in dynamic CAMENs, we conduct comprehensive evaluations, involving both experiments relying on \textit{i)} numerical simulations with various parameter settings (see Sec. 5.3); as well as \textit{ii)} real-world data-based simulations (see Sec. 5.4), upon considering real-world EUA Dataset\cite{lai2018optimal}, which offers the information of base stations (as ESs in our paper) and MUs (e.g., location information) within Metropolitan Melbourne, Australia (an overall region of over 9,000 $ \text{km}^2 $). Simulations are carried out via MATLAB R2019a on the desktop with 12th Gen Intel Core i5-12400 2.5 GHz and 16 GB RAM.
For notational simplicity, our proposed cross-layer matching game in the hybrid service trading market is abbreviated as "Hybrid\_F\_S". Without loss of generality, we adopt the Monte Carlo method, where each value associated with simulation figures represents the average value over 1000 independent simulations \cite{10321730}.

The primary parameters for the simulation have been established in accordance with the supportive existing literature\cite{8815852,9763875,10321730,9682584}, as detailed below: $f_i^U\in[1,1.5]\times10^9$ CPU cycles/s, $ f_j^E\in[1,3]\times10^{12}$ CPU cycles/s, $f_k^C\in[1,3]\times10^{12}$ CPU cycles/s, $d_i^U\in[1,1.5]$ Mb, $r_i^U=600$ cycles/bit $\times~d_i^U$, $e_i^{t}\in[500,550]$ mW, $e_i^U=e_j^E=e_k^C\in[450,500] $ mW, $\mu_1=100$, $\mu_2=400$ (i.e., the received SNR thus falls within $ [17 , 23]\text{dB}$ roughly), $a_i\in[0.64,0.96]$, $G_j^E\in[4,5]$, $K_j^E\in[6,8]$, $G_k^C\in[8,12]$, $\varepsilon_k\in[0,G_k^C]$, $ \sigma_k\in[2,4] $, $W=6$ MHz, $ q^{U\rightarrow E}=3 $, $ q^{E\rightarrow U}=3 $, $ q^{E\rightarrow C}=2 $, $ q^{C\rightarrow I}=1.5 $, $c^{H}=0.05$, 
$p_{i,j}^{min}=p_{i,j,k}^{min}=1.5$, $ \tau=0.1 $, thresholds $\rho_1-\rho_5$ are set by 0.3, the number of CSs is 12.

\subsection{Benchmark Methods and Evaluation Matrices}
To better evaluate the performance of Hybrid\_F\_S, we involve a set of comparable methods as benchmarks.

\noindent
$\bullet$ \textbf{Cross-layer matching game conventional spot trading (Conventional\_S) \cite{8815852}:} Conventional\_S is implemented in a spot trading market over CAMENs, where resources are shared between MUs and ESs, and between ESs and CSs, through M2O and M2M matching.

\noindent
$\bullet$ \textbf{Cross-layer matching game in hybrid market without risk analysis (Hybrid\_F\_S(NoRiskA)):} Hybrid\_F\_S(NoRiskA) is similar to our proposed Hybrid\_F\_S, where the key difference refers to that Hybrid\_F\_S(NoRiskA) omits risk constraints(i.e., constraints (28b), (28c), (29f), (29g), and (30c)) to highlight the importance of incorporating risk constraints in trading decisions.

\noindent
$\bullet$ \textbf{MUs' utility-prioritized cross-layer matching (MU\_Prioritized) \cite{9127160}:} MU\_Prioritized focuses on maximizing MUs' utility by allowing them to select the ES that offers the greatest utility under a spot trading mode, as inspired by greedy algorithms.

\noindent
$\bullet$ \textbf{ESs' utility-prioritized cross-layer matching (ES\_Prioritized) \cite{9127160}:} ES\_Prioritized puts the ESs' utility at the first place, under spot trading mode, as inspired by greedy algorithms.

\noindent
$\bullet$ \textbf{Random matching (Random\_M) \cite{9127160}:} Random\_M offers a fast approach wherein ESs randomly choose MUs and CSs, serving as a baseline to demonstrate the trade-off between time efficiency and resource allocation performance.

To conduct an exhaustive analysis of the factors encompassed by our considered baseline methods, we involve key factors into distinct categories, as summarized and contrasted in Table 1. Apparently, our considered baseline methods have covered most of the existing and popular algorithms with different features.
\begin{table}[htb] 
	\vspace{-0.2cm}
	{\footnotesize
		\caption{\footnotesize{A summary of baseline methods (Sta.: Stable, Gre.: Greedy, Ran.: Random, Ove.: Overbooking, Ris.: Risk evaluation)}} \vspace{-0.6cm} 
		\begin{center}
			\setlength{\tabcolsep}{0.5mm}{
				\begin{tabular}{|c|c|c|c|c|c|c|c|}
					\hline
					\multirow{2}{*}{\textbf{Algorithm}} & \multicolumn{2}{c|}{\textbf{Trading mode}}&\multicolumn{3}{c|}{\textbf{Matching property}} & \multicolumn{2}{c|}{\makecell[c]{\textbf{Innovative}\\ \textbf{attributes}}}\\  \cline{2-8} 
					&\makecell[c]{Spot}&\makecell[c]{Futures}&Sta.&\makecell[c]{Gre.}&Ran.&Ove.&\makecell[c]{Ris.}\\ \hline
					\makecell[c]{Conventional\_S\cite{8815852}} &$\surd$& &$\surd$& & & &\\ \hline
					\makecell[c]{Hybrid\_F\_S(NoRiskA)} &$\surd$&$\surd$&$\surd$& & &$\surd$&\\ \hline
					\makecell[c]{MU\_Prioritized\cite{9127160}} & $\surd$& & &$\surd$& & &\\ \hline
					\makecell[c]{ES\_Prioritized\cite{9127160}} &$\surd$& & &$\surd$& & &\\ \hline
					\makecell[c]{Random\_M\cite{9127160}} &$\surd$& & & &$\surd$ & & \\ \hline
					our work (Hybrid\_F\_S)&$\surd$&$\surd$&$\surd$& & &$\surd$&$\surd$\\ \hline
			\end{tabular}}
	\end{center}}
\end{table}

To conduct quantitative performance evaluations, we focus on the following performance indicators:

\noindent
$\bullet$ \textbf{Social welfare:} The summation of utilities of MUs, ESs, and CSs.

\noindent
$\bullet$ \textbf{Running time (RT, millisecond):} RT can be obtained by MATLAB 2019a, which reflects the delay cost consumed by looking for matching results, and thereby the time efficiency.

\noindent
$\bullet$ \textbf{Number of interactions (NI):} NI estimates the total number of interactions between MUs and ESs consumed by obtaining matching results\footnote{This paper assumes the communications between can be supported by ESs and CSs wired links, we thus omit the corresponding communication overhead.}, which further captures the overhead of decision-making.

\noindent
$\bullet$ \textbf{Practical task completion time (PTCT, millisecond):} The PTCT of a task should consider the latency caused by matching decision-making, which is estimated by the end-to-end (E2E) delay of communication links between MUs and ESs falls in $ [1, 15] $ ms\cite{10321730} in our simulation. Accordingly, we calculate the PTCT of each task as the summation of data transmission and execution time, as well as the decision-making latency.

\subsection{Numerical Simulations}
Having comprehensive numerical simulations helps us to better illustrate the superior performance of our proposed Hybrid\_F\_S, as comparing to benchmark methods upon considering diverse problem scales.
\subsubsection{Performance evaluation on social welfare, RT, NI, and PTCT}
\begin{figure*} \centering 
	\vspace{-2.3cm}
	\subfigtopskip=2pt
	\subfigbottomskip=10pt
	\subfigcapskip=-2cm
	\setlength{\abovecaptionskip}{-1.65cm}
	\subfigure[] {
		\label{fig:a}   
		\includegraphics[width=0.52\columnwidth]{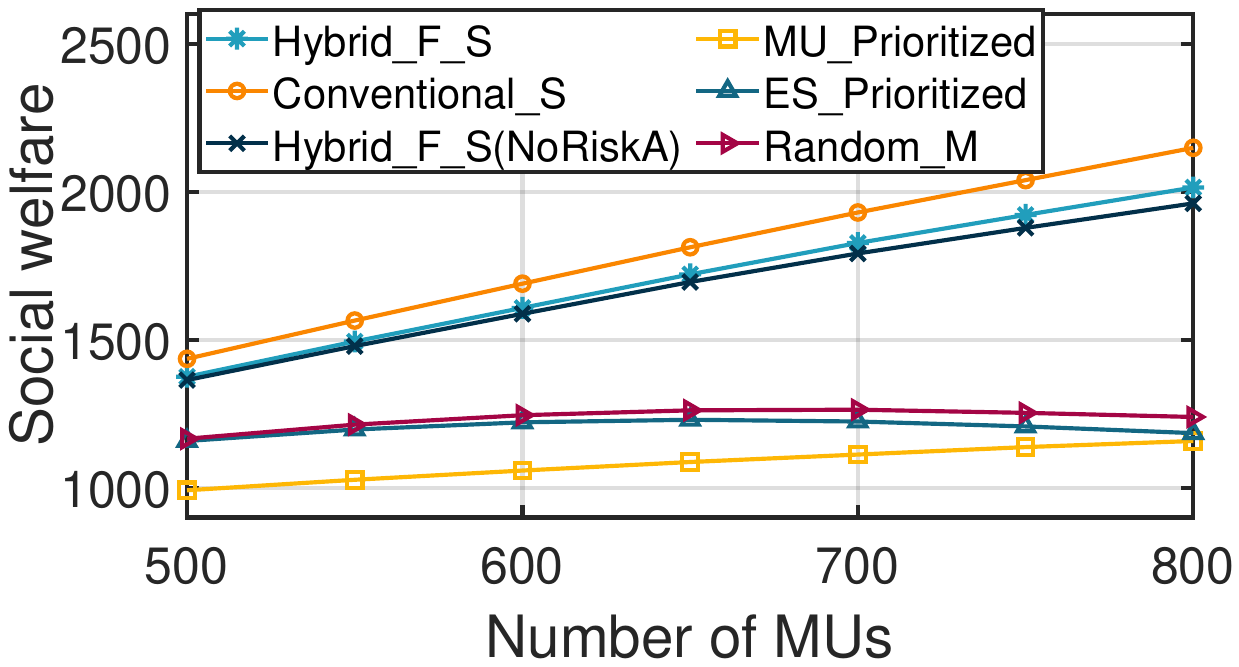} 
	}  \hspace{-5.3mm} 
	\subfigure[] { 
		\label{fig:b}   
		\includegraphics[width=0.52\columnwidth]{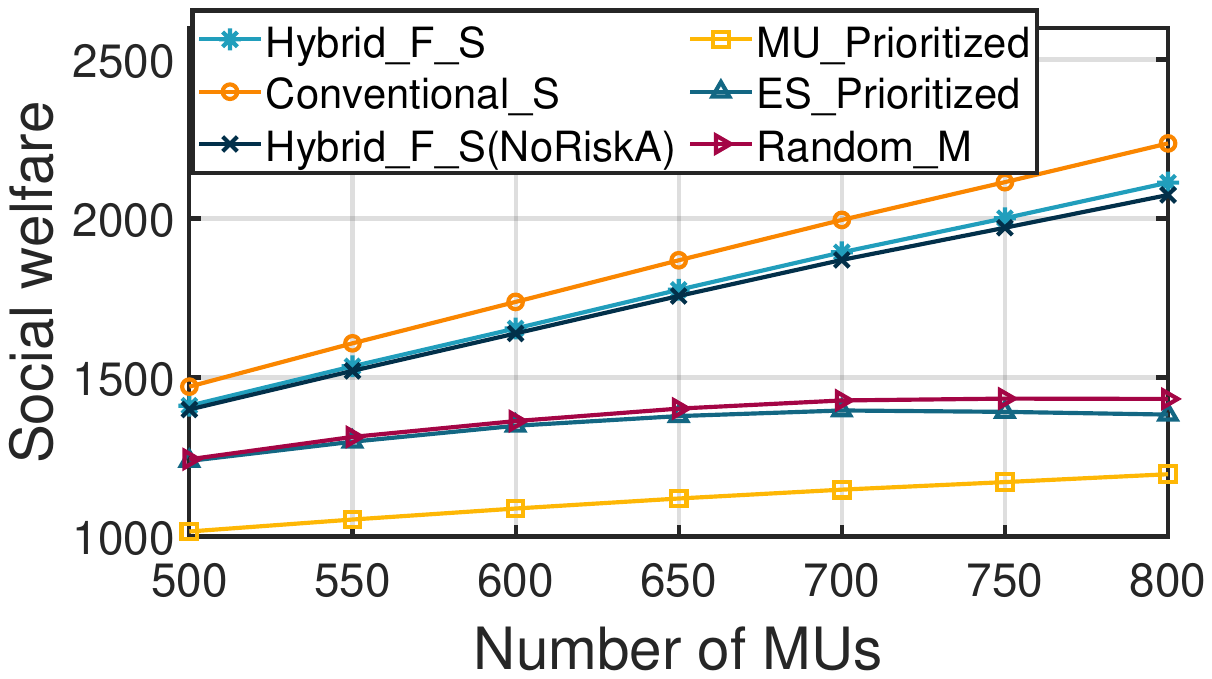}   
	}   \hspace{-5.7mm}
	\subfigure[] {
		\label{fig:c}   
		\includegraphics[width=0.52\columnwidth]{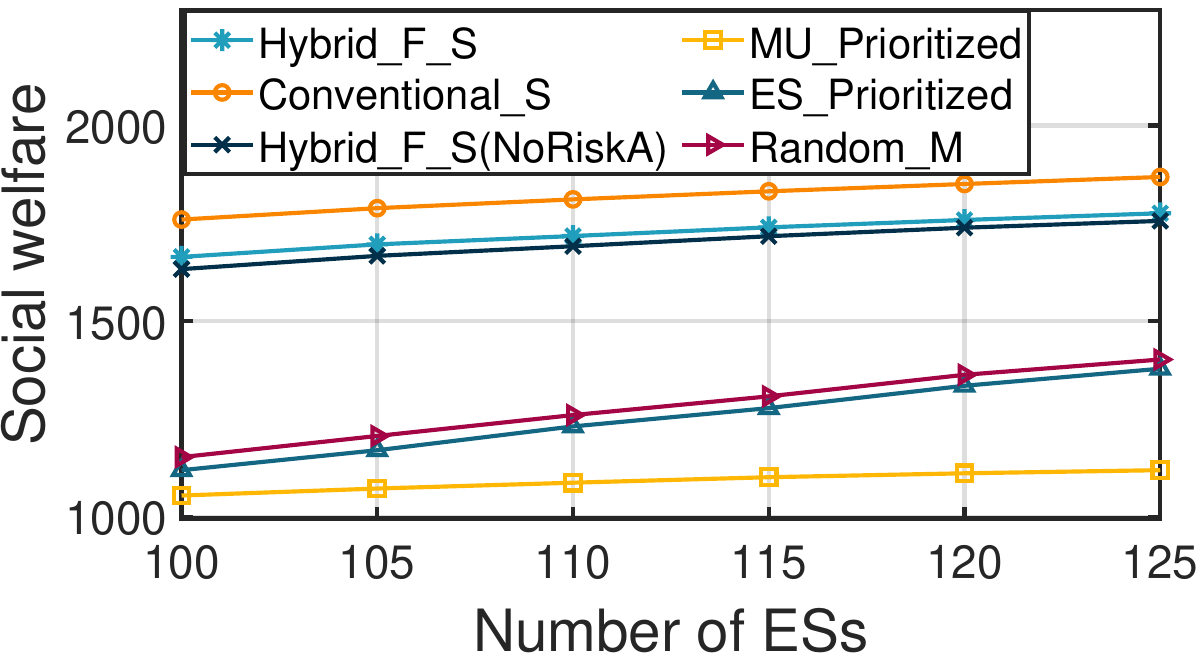} 
	}  \hspace{-6mm} 
	\subfigure[] { 
		\label{fig:d}   
		\includegraphics[width=0.52\columnwidth]{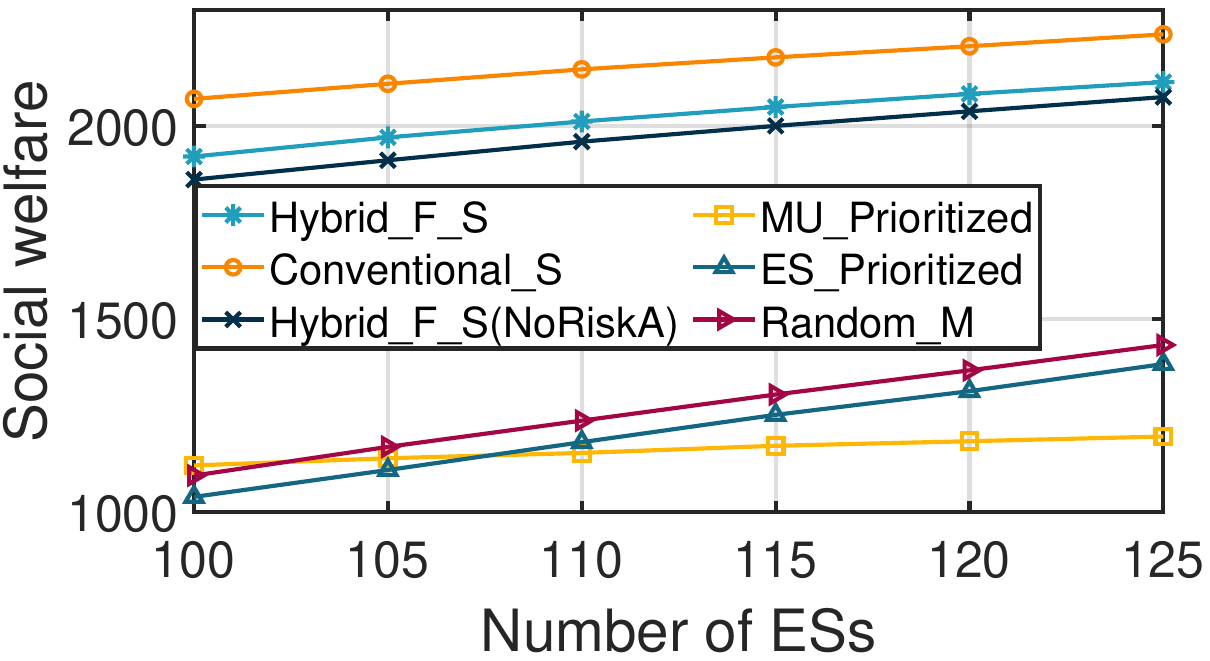}   
	}   
	\caption{Performance comparisons in terms of social welfare, where (a) 110 ESs, (b) 125 ESs, (c) 650 MUs, and (d) 800 MUs.}   
	\label{fig}  
	\vspace{-0.2 cm} 
\end{figure*}
\begin{figure*}[t!] \centering 
	\subfigtopskip=2pt
	\subfigbottomskip=10pt
	\subfigcapskip=-0.2cm
	\setlength{\abovecaptionskip}{-0.3cm}
	\subfigure[] {
		\label{fig:a}   
		\includegraphics[width=0.56\columnwidth]{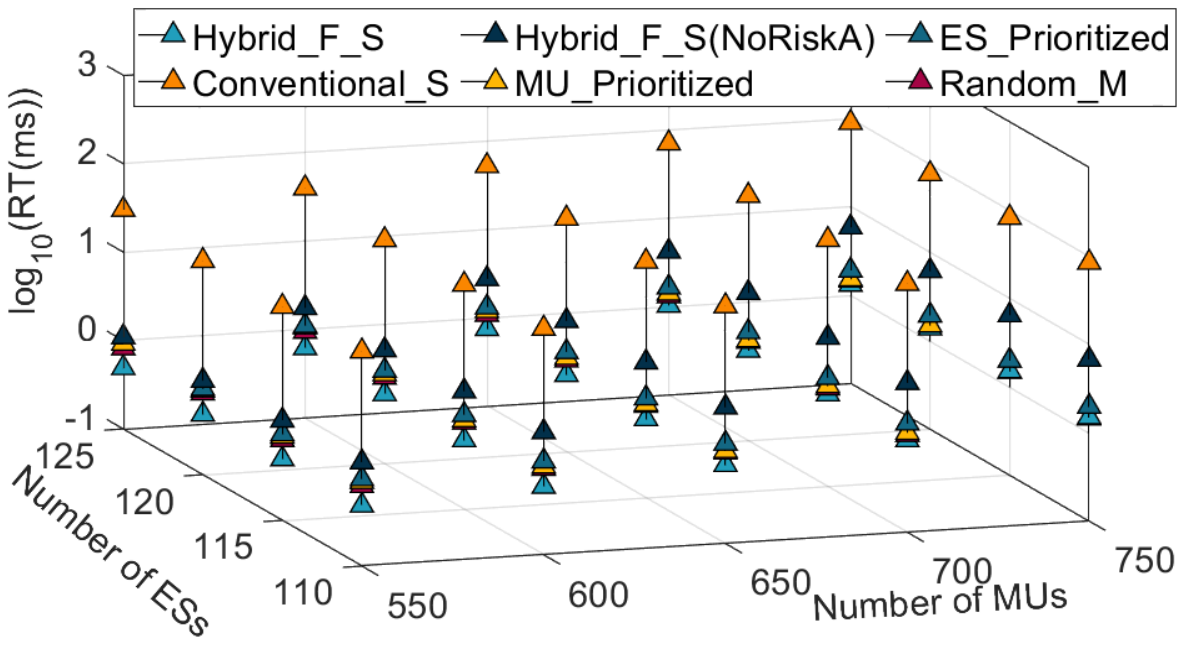} 
	}  \hspace{-4mm} 
	\subfigure[] { 
		\label{fig:b}   
		\includegraphics[width=0.5\columnwidth]{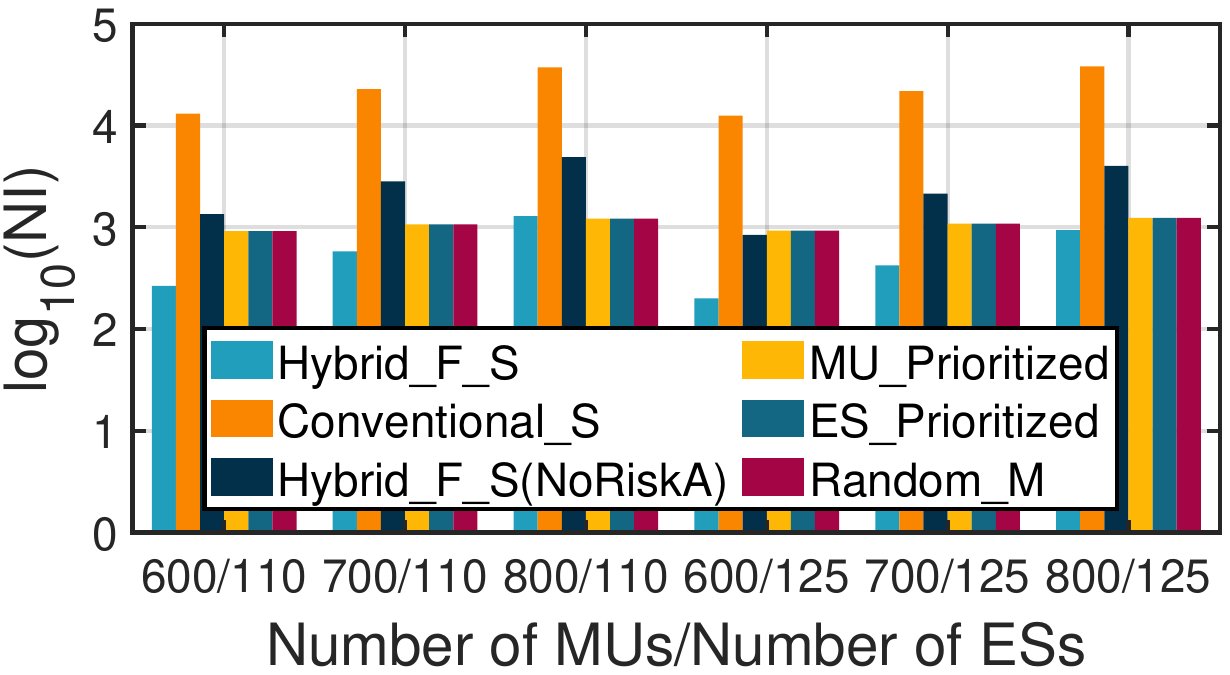}   
	}  \hspace{-4mm} 
	\subfigure[] {
		\label{fig:c}   
		\includegraphics[width=0.48\columnwidth]{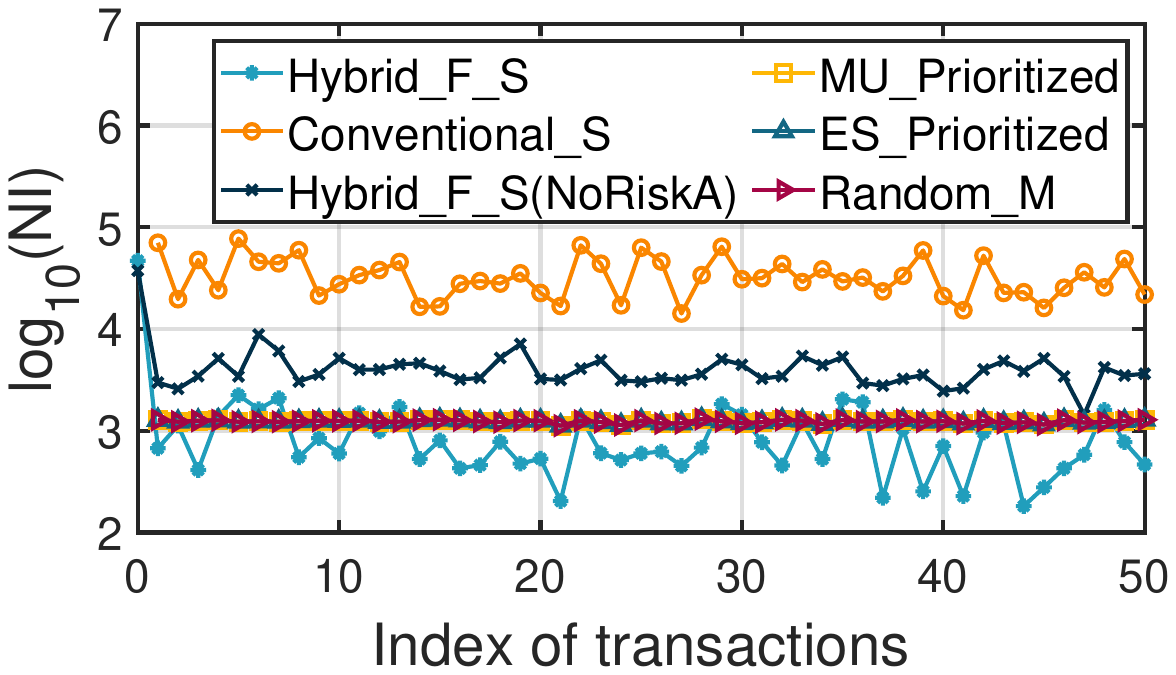} 
	}  \hspace{-4mm} 
	\subfigure[] { 
		\label{fig:d}   
		\includegraphics[width=0.5\columnwidth]{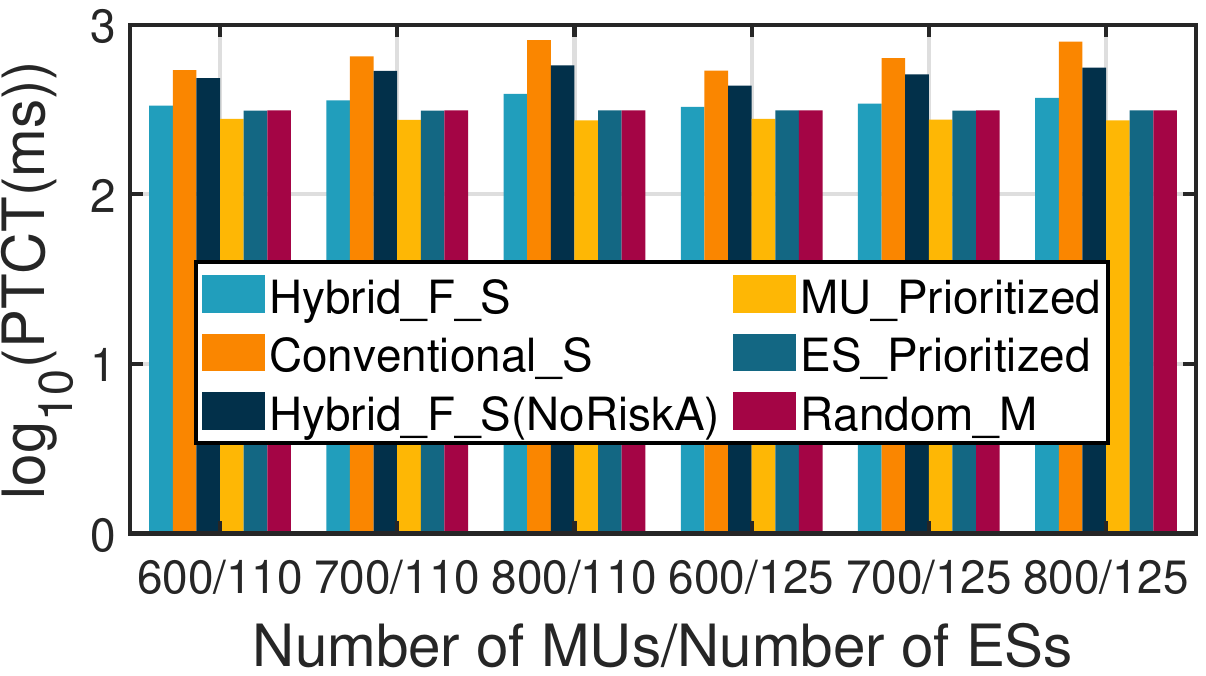}   
	}   
	\caption{Performance comparisons in terms of running time, the number of interactions and practical task completion time.}   
	\label{fig}  
	\vspace{-0.2 cm}
\end{figure*} 

As social welfare represents a crucial factor in evaluating the mutually beneficial performance of our proposed cross-layer matching game to all the three parties, we conduct Fig. 2 upon having various market scales.

Fig. 2(a) and Fig. 2(b) involve 110 and 125 ESs, respectively, to estimate the resource supply in the considered market. In Fig. 2(a), the curves of our proposed Hybrid\_F\_S, as well as benchmark methods Conventional\_S and Hybrid\_F\_S(NoRiskA) show a rising trend since the existence of more MUs implies an increase in resource demand and thus the utility of both MUs and ESs. The proposed Hybrid\_F\_S outperforms Hybrid\_F\_S(NoRiskA) thanks to the well-designed careful risk analysis, thus supporting the individual rationality of all the three parties. Conventional\_S achieves the best performance in terms of social welfare due to its analysis of the current network/market conditions during each practical transaction, which, in turn, suffers from excessive overhead (see Fig. 3). Such a drawback makes it inapplicable in real-world dynamic networks particularly when handling moving and battery-constrained users as well as delay-sensitive mobile tasks. Moreover, our proposed Hybrid\_F\_S greatly outperforms the other three benchmark methods, since their basic matching logic can bring disadvantages. For example, MU\_Prioritized and ES\_Prioritized methods only care about the utility of one of the three parties, which thus sacrificing the profits of others. Also, Random\_M method is risking the uncertainty brought by its randomness, and thus causes unsatisfying social welfare. Note that in Fig. 2(a), the curves of ES\_Prioritized and Random\_M methods slightly go down after 700 MUs, the reason behind which is that each ES is given more options on MU selection, but constrained by limited resources. For example, when the number of MUs reaches a certain level, ESs may have to pay more penalties to volunteers. Fig. 2(b) illustrates a similar performance to Fig. 2(a), while thanks to more sufficient resource supply, the curves rise with an increasing number of MUs, although ES\_Prioritized still takes the risk that the social welfare may fall down. 

 Fig. 2(c) and Fig. 2(d) consider 650 and 800 MUs to reflect the growing resource demand in our designed hybrid market, upon having various numbers of ESs (namely, resource supply). The social welfare raises with a growing number of ESs owing to the existence of a bigger resource pool. Also, as can be seen from Figs. 2(c)-(d), our proposed Hybrid\_F\_S reaches far better performance than MU\_Prioritized, ES\_Prioritized, Random\_M, and Hybrid\_F\_S(NoRiskA), mainly because of the well-designed cross-layer matching mechanism, achieving risk-aware and mutually beneficial utilities for all the parties. Although in these two figures, the value of social welfare of Hybrid\_F\_S stays below that of Conventional\_S, Conventional\_S is undergoing unacceptable decision-making overhead. Additionally, the detailed analysis of individual utilities of MUs, ESs, and CSs are provided in Appendix F, due to space limitation.

Time/energy efficiency plays one significant role in evaluating the performance of a resource trading market over CAMENs. To achieve better assessments, we consider three indicators in Fig. 3, reflecting the overhead incurred by obtaining matching decisions, which are RT (see Fig. 3(a)), NI (see Figs. 3(b)-(c)), and PTCT (see Fig. 3(d)). Note that the logarithmic representation is utilized for the y-axis of each figure in Fig. 3 to visually enlarge the gap among different methods. 

Fig. 3(a) illustrates the running time performance under different market scales. It is obvious that the value of RT of Conventional\_S stays dramatically high in comparison with other methods since a large amount of time should be consumed for looking for matching results during each transaction, this case appears to be more severe when confronting increasing resource demands, e.g., a raising number of ESs and MUs. Our proposed Hybrid\_F\_S achieves far lower RT than Conventional\_S due to that many of the ESs and MUs will not engage in spot trading thanks to resource overbooking and pre-signed forward contracts, which thus no longer have to spend time/energy in negotiating matching results in practical transactions. Moreover, the well-designed risk analysis mechanism makes Hybrid\_F\_S reach a better performance on RT than Hybrid\_F\_S(NoRiskA), by ensuring the fulfillment of most contracts. Namely, the number of participants of Hybrid\_F\_S(NoRiskA) that join in spot trading market can generally be larger than that of Hybrid\_F\_S. Although the other three methods (MU\_Prioritized, ES\_Prioritized, and Random\_M) achieve similar RT with Hybrid\_F\_S owing to no bargain of service price among different parties, they suffer from unsatisfying performance on social welfare (see Fig. 2). 

We are then interested in capturing the overhead (e.g., time and energy cost) caused by matching decision-making, for which Figs. 3(b)-(c) are conducted. Apparently, a large value of NI can reveal a heavy overhead for obtaining matching decisions, e.g., excessive time and energy can be consumed during the interaction/communication process among ESs and MUs. Fig. 3(b) shows an overall performance on NI upon considering diverse resource demand/supply settings (e.g., various numbers of ESs and MUs), where our proposed Hybrid\_F\_S shows better values than other methods. First, the NI of Hybrid\_F\_S falls far below that of Conventional\_S since its participants should bargain for the amount of trading resources and service prices in each transaction, which definitely imposes time and energy overhead. Fortunately, Hybrid\_F\_S encourages many participants to take part in futures trading first, and thus greatly improves time/energy efficiency. Then, as risk analysis presents an important part in futures trading, our Hybrid\_F\_S works better than Hybrid\_F\_S(NoRiskA) thanks to our design on risk management. In addition, Hybrid\_F\_S slightly outperforms MU\_Prioritized, ES\_Prioritized, and Random\_M on NI although they do not care about bargains among participants, simple interactions such as ESs should inform each MU about its willingness to offer services, as well as MUs should report their intention, which also consumes certain overhead, especially when the market scale becomes large. To detail the gaps of NI in every transaction, Fig. 3(c) depicts the performance on NI in 50 specific transactions, upon having 800 MUs and 125 ESs as a general example. As can be seen from Fig. 3(c), our Hybrid\_F\_S always outperforms Conventional\_S and Hybrid\_F\_S(NoRiskA), while achieving lower NI than MU\_Prioritized, ES\_Prioritized, and Random\_M in most transactions.

Generally, the processing of a task can only start after the final decision of resource trading, e.g., a MU can offload its task data only after receiving the notification of a certain ES. 
To this end, our simulation takes a fresh look at the time efficiency performance of the market, by involving the decision-making overhead in to task completion time, which distinguishes it from existing literature. Specifically, the value of PTCT is a summation of \textit{i)} the delay on bargain, which can be estimated by the E2E delay between MUs and ESs, as well as their interactions \cite{9763875}, \textit{ii)} the theoretical task completion time according to Sec. 3.1.1. We can see from Fig. 3(d), our Hybrid\_F\_S obtains superior performance on PTCT than other methods, benefit from overbooking, risk analysis, and the hybrid mechanism of trading. In summary, Hybrid\_F\_S achieves commendable social welfare performance, while outperforming other methods on significant evaluation indicators such as RT, NI, and PTCT, offering a good reference for improving time/energy efficiency.

\subsubsection{Property Analysis}
\begin{figure} \centering 
	\vspace{-1.9999cm}
	\subfigtopskip=2pt
	\subfigbottomskip=10pt
	\subfigcapskip=-2.0cm
	\setlength{\abovecaptionskip}{-1.6cm}
	\subfigure[] {
		\label{fig:a}   
		\includegraphics[width=0.50\columnwidth]{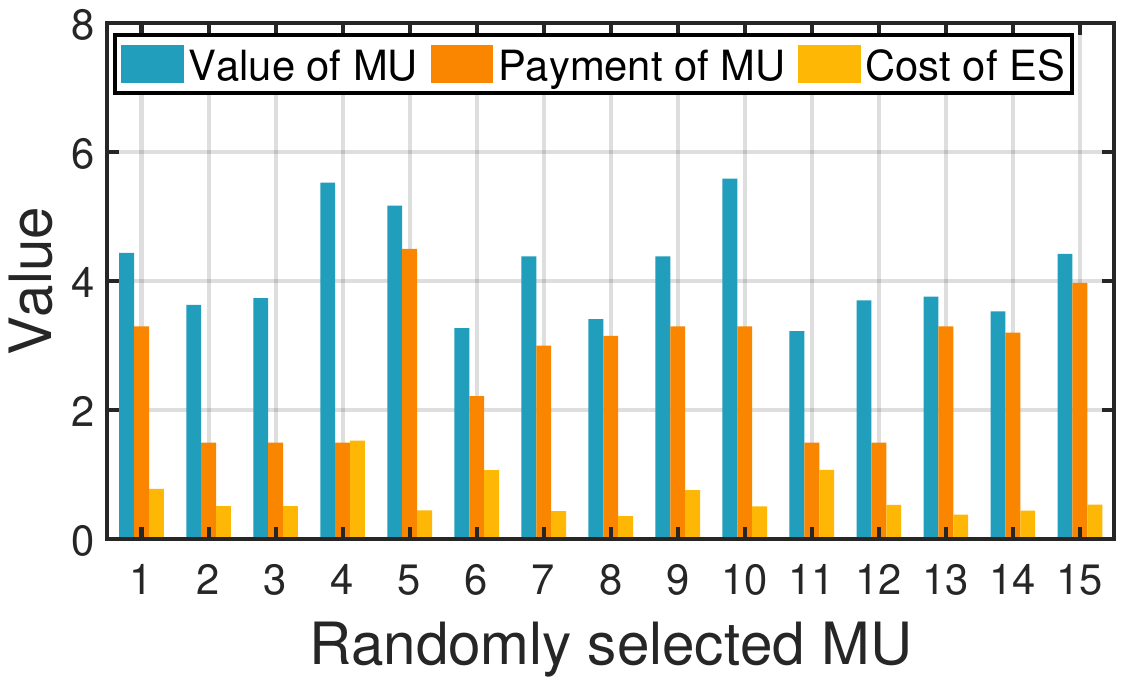} 
	}  \hspace{-6mm}
	\subfigure[] { 
		\label{fig:b}   
		\includegraphics[width=0.50\columnwidth]{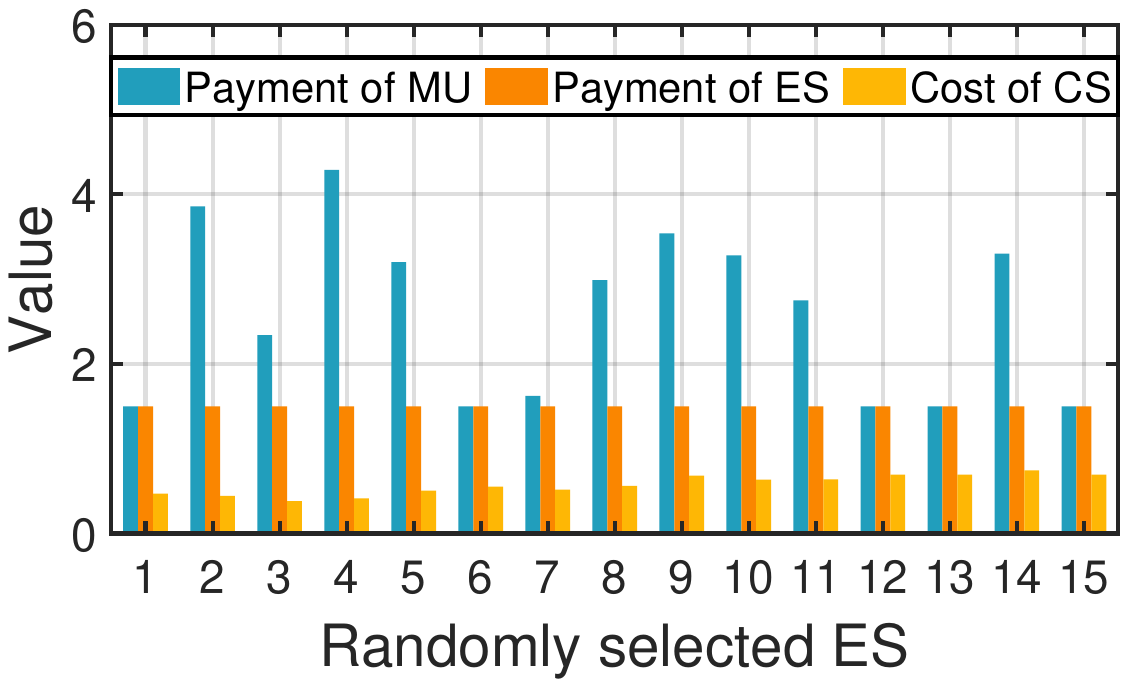}   
	}     
	\caption{Individual rationality in terms of utilities.}   
	\label{fig}  
	\vspace{-0.2 cm}
\end{figure} 
\begin{table}[]
	{\small
		\caption{Individual rationality in terms of risks, in which the first, fourth, and seventh columns denote the index of the selected MUs, ES, and CSs, respectively.}\vspace{-0.5cm}
		\begin{center}
			\setlength{\tabcolsep}{1.2mm}{
				\begin{tabular}{|cccccccc|}
					\hline
					MU & $ R_1^U $& $ R_2^U $& ES &$ R_1^E $ &$ R_2^E $ & CS & $ R^C $\\ \hline
					65&0.0541&$ 1.53\times10^{-5} $&12&0.1432&0.1830&1&0.1428 \\ \hline
					122 &0.0628&$ 3.02\times 10^{-3} $&79&0.2523&0.1789&5&0.0527\\ \hline
					403 &0.1232&$ 3.43\times 10^{-4} $&99&0&0.2104&9&0.1847\\ \hline
					634 &0.2464&$ 2.53\times 10^{-4} $&116&0.02520&0.2586&11&0.2067\\ \hline
			\end{tabular}}
	\end{center}}
\vspace{-0.5cm}
\end{table}
Since our proposed cross-layer matching game designs the property of individual rationality from a rather different view as comparing conventional one, this section conducts Fig. 4 and Table 2 to verify individual rationality from the following aspects: \textit{i)} valuations, payments, and costs; and \textit{ii)} risk analysis. Specifically, we randomly select 15 MUs (among 800 ones) and show their valuations (calculated by equation (3) in Sec. 3.1), payments, as well as the service costs of the corresponding ESs in Fig. 4(a). Apparently, Fig. 4(a) shows that the payments of MUs never exceed their valuations, and while the payment received by ESs will definitely cover their costs, verifying that our proposed Hybrid\_F\_S guarantees the individual rationality of MUs and ESs. Fig. 4(b) considers 15 randomly selected ESs out of 125 ones, and shows their payment from MUs, paid payments to CSs, and service costs of the corresponding CSs. As can be seen from this figure, the payment obtained by ESs from MUs will not exceed their payments to CSs. Meanwhile, the service cost of CSs can be covered by their asked payments, which further verifies that ESs and CSs in our proposed Hybrid\_F\_S are individually rational. 

To describe the risks that the three parties may confront, we randomly consider 4 MUs (among 800 ones), 4 ESs (among 125 ones), and 4 CSs (among 12 ones) in Table 2. This table shows that the risks of different participants are always controlled within their acceptable ranges (e.g., $ R_1^U $, $ R_2^U $, $ R_1^E $, $ R_2^E $, and $ R^C$ will stay below $ 0.3$), proving that our Hybrid\_F\_S greatly supports individual rationality property from the perspective of risk management.
\subsubsection{Impact brought by overbooking}
We next illustrate the advantages brought by overbooking in coping with uncertainties and dynamics in CAMENs. We evaluate the performances in terms of RT and NI under different overbooking rates (i.e., $\tau$) in Fig. 5, involving 800 MUs, 125 ESs, and 12 CSs. Notably, since Conventional\_S, MU\_Prioritized, ES\_Prioritized, and Random\_M are implemented according to the spot trading mode, different overbooking rates leave no impact on these methods (thus explaining why they have rather stable curves).

\begin{figure} \centering  
	\vspace{-1.9999cm}
	\subfigtopskip=2pt
	\subfigbottomskip=10pt
	\subfigcapskip=-2.0cm
	\setlength{\abovecaptionskip}{-1.6cm}
	\subfigure[] {
		\label{fig:a}   
		\includegraphics[width=0.50\columnwidth]{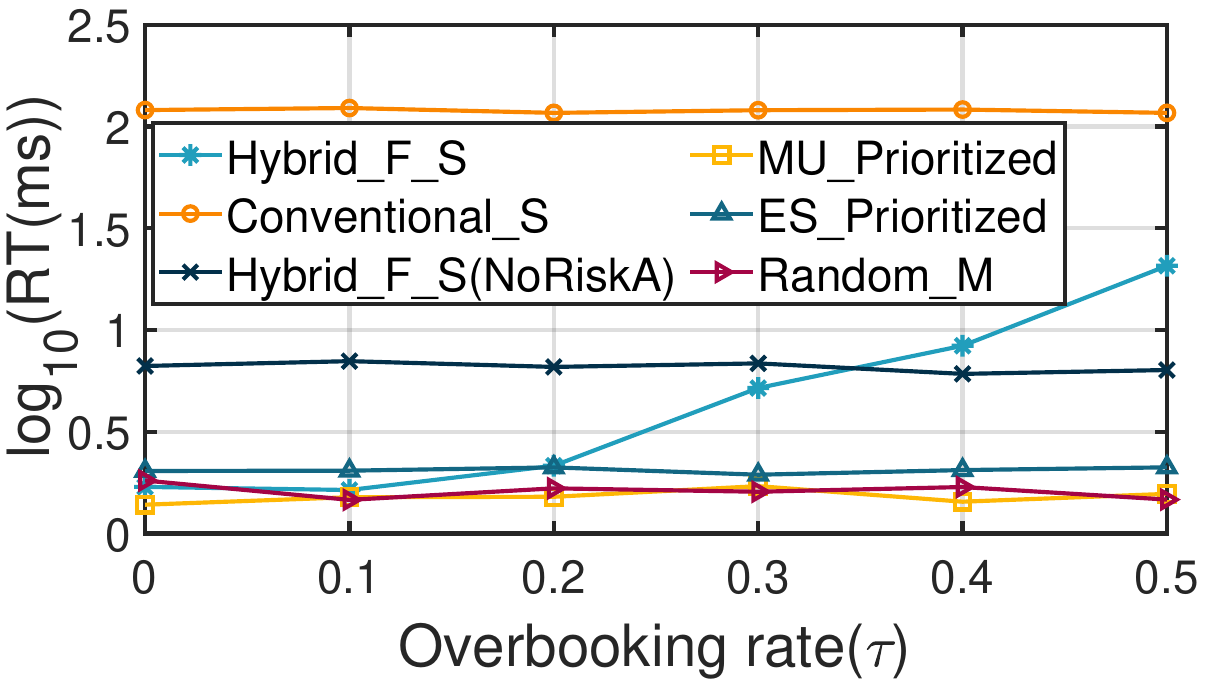} 
	}  \hspace{-5.8mm}
	\subfigure[] { 
		\label{fig:b}   
		\includegraphics[width=0.50\columnwidth]{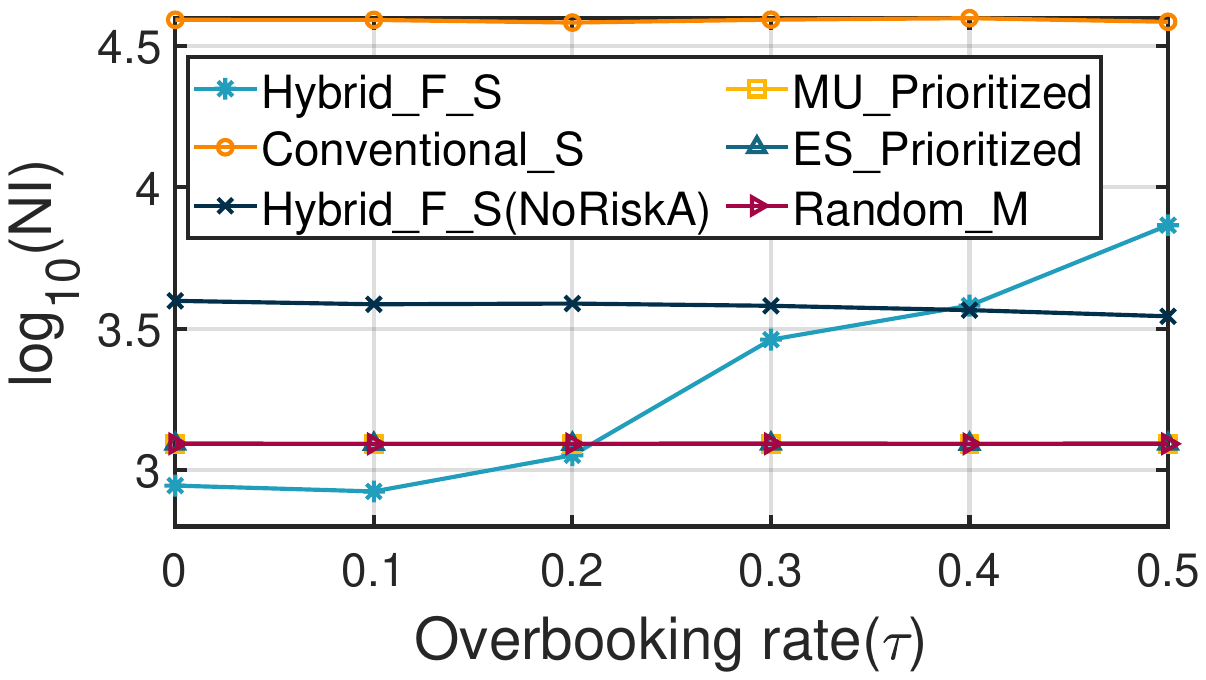}   
	}     
	\caption{Performance comparisons in terms of running time and number of interaction under different overbooking rates ($\tau$).}   
	\label{fig}  
	\vspace{-0.2 cm}
\end{figure} 

Fig. 5(a) evaluates the performance on RT upon raising the value of overbooking rate from 0 (e.g., the market does not allow resource overbooking) to 0.5. Our Hybrid \_F\_S greatly outperforms Conventional\_S on RT (similar results can be found in previous Fig. 3); while although MU\_Prioritized, ES\_Prioritized, and Random\_M can reach a low value of RT, they are undergoing poor social welfare, as depicted in Fig. 2. Note that RT of Hybrid\_F\_S decreases from $\tau=0$ to $\tau=0.1$, since MUs are given more chances to trade with ESs in the designed futures market, which further reducing the number of participants in spot trading market during practical transactions, and thus the time spent on matching decision-making. 
However, it is obvious that RT of Hybrid\_F\_S raises with an increasing overbooking rate after $\tau=0.1$, where the key reasoning behind which is that a larger amount of overbooked resources enables each ES to sign contracts with more MUs in the futures market, further leading to a large burden on risk analysis and volunteer selection during practical transactions, as well as an increases RT. Although Hybrid\_F\_S(NoRiskA) facilitates a lower RT after $\tau=0.35$ (roughly) rather than that of our approach, it sustains from the case where a large number of participants will break the pre-signed contracts due to the lack of risk management. For example, excessive penalties can be incurred among different parties. 

Fig. 5(b) illustrates the performance on NI by applying different overbooking rates, where the value of NI of Hybrid\_F\_S slightly falls down from $\tau=0$ to $\tau=0.1$, since more MUs can take part in the futures market and no longer have to be engaged in onsite decision-making. Similar to Fig. 5(a), the value of NI of Hybrid\_F\_S begins to rise after $\tau=0.1$ in Fig. 5(b) due to that more workload on risk analysis should be considered during practical transactions. Even MU\_Prioritized, ES\_Prioritized, and Random\_M methods can get better performance on NI than Hybrid\_F\_S when increasing $\tau$, they suffer from unexpected social welfare. Moreover, the curve of Hybrid\_F\_S(NoRiskA) decreases as $\tau$ grows since without risk analysis, and more MUs will join in the futures market rather than spot market. However, Hybrid\_F\_S(NoRiskA) confronts undesired performance on PTCT and social welfare, as illustrate in Fig. 3. 

All in all, our proposed Hybrid cross-layer matching game can achieve commendable performance in terms of social welfare, while outperforming benchmark methods on decision-making overhead.

\subsection{Simulations on Real-World Datasets}
To facilitate better evaluations, we consider real-world EUA Dataset \cite{lai2018optimal}, regarding the Melbourne central business district area as our simulated region (shown in Fig. 6, Appendix E). Table 3 also verifies that our proposed Hybrid\_F\_S can achieve commendable performance in terms of social welfare, and outperforms baseline methods in terms of RT, NI, and PTCT.
\begin{table}[h!]
	{\small \vspace{-0.4cm}
		\caption{Performance Evaluations in EUA Datasets (Alg 1: Hybrid\_F\_S, Alg 2: Conventional\_S, Alg 3: Hybrid\_F\_S(NoRiskA), Alg 4: MU\_Prioritized, Alg 5: ES\_Prioritized, Alg 6: Random\_M)}
		\begin{center}
			\setlength{\tabcolsep}{0.5mm}{\vspace{-0.2cm}
				\begin{tabular}{|ccccccc|}
					\hline
					\textbf{Performance} & \textbf{Alg 1}& \textbf{Alg 2}&\textbf{Alg 3}&\textbf{Alg 4} & \textbf{Alg 5} & \textbf{Alg 6}\\ \hline
					Social welfare &2037.84&2172.81&2004.93&1188.37&1344.06&1392.38
					\\ \hline
					RT (ms) &1.8&90.1&7.0 &1.6&2.1&1.5\\ \hline
					NI &1031.36&27143.51&3823.26&1239.08&1239.08&1239.08
					\\ \hline
					PTCT (ms) &379.26&662.41&543.84&271.67&312.55&312.64	\\ \hline
			\end{tabular}}
	\end{center}}
\end{table}

\vspace{-0.3cm}
\section{Conclusion}
\noindent
This paper delves into an interesting cross-layer matching game within dynamic CAMENs, focusing on a hybrid computing resource trading market that encompasses futures and spot trading. Initially, we developed an OA-CLM game tailored for futures trading, fostering pre-signed forward contracts among multiple MUs, ESs, and CSs by meticulously analyzing historical statistics linked to uncertainties.
Given concerns regarding fluctuations in dynamic resource demand and supply, we introduced an OS-CLM mechanism for spot trading as a robust contingency plan. Theoretical exploration demonstrates that these matching mechanisms can uphold essential properties including individual rationality, strong stability, competitive equilibrium, and weak Pareto optimality.
Extensive simulations underscore the commendable performance of our methodology, showcasing superior results in running time, minimizing interactions among participants, and enhancing social welfare compared to benchmark methods, under both numerical and real-world datasets. We are also interested in exploring potential future directions, such as the design of smart contracts for resource trading and potential collaborations among ESs.

\begin{spacing}{0.96}
\bibliographystyle{IEEEtran}
\bibliography{Reference}

\begin{thebibliography}{10}
\providecommand{\url}[1]{#1}
\csname url@samestyle\endcsname
\providecommand{\newblock}{\relax}
\providecommand{\bibinfo}[2]{#2}
\providecommand{\BIBentrySTDinterwordspacing}{\spaceskip=0pt\relax}
\providecommand{\BIBentryALTinterwordstretchfactor}{4}
\providecommand{\BIBentryALTinterwordspacing}{\spaceskip=\fontdimen2\font plus
\BIBentryALTinterwordstretchfactor\fontdimen3\font minus
  \fontdimen4\font\relax}
\providecommand{\BIBforeignlanguage}[2]{{%
\expandafter\ifx\csname l@#1\endcsname\relax
\typeout{** WARNING: IEEEtran.bst: No hyphenation pattern has been}%
\typeout{** loaded for the language `#1'. Using the pattern for}%
\typeout{** the default language instead.}%
\else
\language=\csname l@#1\endcsname
\fi
#2}}
\providecommand{\BIBdecl}{\relax}
\BIBdecl

\bibitem{survey1}
S.~Duan, D.~Wang, J.~Ren, F.~Lyu, Y.~Zhang, H.~Wu, and X.~Shen, ``Distributed
  artificial intelligence empowered by end-edge-cloud computing: A survey,''
  \emph{IEEE Commun. Surveys Tut.}, vol.~25, no.~1, pp. 591--624, 2023.

\bibitem{CC1}
X.~Li, Z.~Xu, F.~Fang, Q.~Fan, X.~Wang, and V.~C.~M. Leung, ``Task offloading
  for deep learning empowered automatic speech analysis in mobile edge-cloud
  computing networks,'' \emph{IEEE Trans. Cloud Comput.}, vol.~11, no.~2, pp.
  1985--1998, 2023.

\bibitem{CC2}
A.~E. Eshratifar, M.~S. Abrishami, and M.~Pedram, ``Jointdnn: An efficient
  training and inference engine for intelligent mobile cloud computing
  services,'' \emph{IEEE Trans. Mobile Comput.}, vol.~20, no.~2, pp. 565--576,
  2021.

\bibitem{CC3}
P.~Bellavista, A.~Corradi, A.~Edmonds, L.~Foschini, A.~Zanni, and T.~M.
  Bohnert, ``Elastic provisioning of stateful telco services in mobile cloud
  networking,'' \emph{IEEE Trans. Serv. Comput.}, vol.~14, no.~3, pp. 710--723,
  2021.

\bibitem{R2_2}
S.~Chen, L.~Jiao, F.~Liu, and L.~Wang, ``Edgedr: An online mechanism design for
  demand response in edge clouds,'' \emph{IEEE Trans. Parallel Distrib. Syst.},
  vol.~33, no.~2, pp. 343--358, 2022.

\bibitem{CCbad}
F.~Shirin~Abkenar, P.~Ramezani, S.~Iranmanesh, S.~Murali, D.~Chulerttiyawong,
  X.~Wan, A.~Jamalipour, and R.~Raad, ``A survey on mobility of edge computing
  networks in iot: State-of-the-art, architectures, and challenges,''
  \emph{IEEE Commun. Surveys Tut.}, vol.~24, no.~4, pp. 2329--2365, 2022.

\bibitem{MEC1}
G.~Zhu, D.~Liu, Y.~Du, C.~You, J.~Zhang, and K.~Huang, ``Toward an intelligent
  edge: Wireless communication meets machine learning,'' \emph{IEEE Commun.
  Mag.}, vol.~58, no.~1, pp. 19--25, 2020.

\bibitem{MEC2}
T.~Qiu, J.~Chi, X.~Zhou, Z.~Ning, M.~Atiquzzaman, and D.~O. Wu, ``Edge
  computing in industrial internet of things: Architecture, advances and
  challenges,'' \emph{IEEE Commun. Surveys Tut.}, vol.~22, no.~4, pp.
  2462--2488, 2020.

\bibitem{R2_3}
B.~Gao, Z.~Zhou, F.~Liu, F.~Xu, and B.~Li, ``An online framework for joint
  network selection and service placement in mobile edge computing,''
  \emph{IEEE Trans. Mobile Comput.}, vol.~21, no.~11, pp. 3836--3851, 2022.

\bibitem{R2_4}
Q.~Chen, Z.~Zheng, C.~Hu, D.~Wang, and F.~Liu, ``On-edge multi-task transfer
  learning: Model and practice with data-driven task allocation,'' \emph{IEEE
  Trans. Parallel Distrib. Syst.}, vol.~31, no.~6, pp. 1357--1371, 2020.

\bibitem{R2_5}
L.~Pan, L.~Wang, S.~Chen, and F.~Liu, ``Retention-aware container caching for
  serverless edge computing,'' in \emph{IEEE Conf. Comput. Commun.}, 2022, pp.
  1069--1078.

\bibitem{R1}
J.~Wu, L.~Wang, Q.~Pei, X.~Cui, F.~Liu, and T.~Yang, ``Hitdl: High-throughput
  deep learning inference at the hybrid mobile edge,'' \emph{IEEE Trans.
  Parallel Distrib. Syst.}, vol.~33, no.~12, pp. 4499--4514, 2022.

\bibitem{8626532}
Z.~Xu, L.~Zhou, S.~C.-K. Chau, W.~Liang, H.~Dai, L.~Chen, W.~Xu, Q.~Xia, and
  P.~Zhou, ``Near-optimal and collaborative service caching in mobile edge
  clouds,'' \emph{IEEE Trans. Mobile Comput.}, vol.~22, no.~7, pp. 4070--4085,
  2023.

\bibitem{9763875}
M.~Liwang and X.~Wang, ``Overbooking-empowered computing resource provisioning
  in cloud-aided mobile edge networks,'' \emph{IEEE/ACM Trans. Netw.}, vol.~30,
  no.~5, pp. 2289--2303, 2022.

\bibitem{9154594}
X.~Wang, J.~Wang, X.~Zhang, X.~Chen, and P.~Zhou, ``Joint task offloading and
  payment determination for mobile edge computing: A stable matching based
  approach,'' \emph{IEEE Trans. Veh. Technol.}, vol.~69, no.~10, pp.
  12\,148--12\,161, 2020.

\bibitem{R2_6}
S.~Chen, L.~Wang, and F.~Liu, ``Optimal admission control mechanism design for
  time-sensitive services in edge computing,'' in \emph{IEEE Conf. Comput.
  Commun.}, 2022, pp. 1169--1178.

\bibitem{9771321}
M.~Liwang, X.~Wang, and R.~Chen, ``Computing resource provisioning at the edge:
  An overbooking-enabled trading paradigm,'' \emph{IEEE Wireless Commun.},
  vol.~29, no.~5, pp. 68--76, 2022.

\bibitem{9453820}
M.~Liwang, Z.~Gao, and X.~Wang, ``Let's trade in the future! a futures-enabled
  fast resource trading mechanism in edge computing-assisted uav networks,''
  \emph{IEEE J. Select. Areas Commun.}, vol.~39, no.~11, pp. 3252--3270, 2021.

\bibitem{1990Two}
A.~E. Roth and M.~A.~O. Sotomayor, ``Two-sided matching: A study in
  game-theoretic modeling and analysis,'' no.~3, pp. 54--77, 1990.

\bibitem{9416305}
J.~Ren, F.~Xia, X.~Chen, J.~Liu, M.~Hou, A.~Shehzad, N.~Sultanova, and X.~Kong,
  ``Matching algorithms: Fundamentals, applications and challenges,''
  \emph{IEEE Trans. Emerg. Topics Comput.}, vol.~5, no.~3, pp. 332--350, 2021.

\bibitem{8854324}
S.~Sheng, R.~Chen, P.~Chen, X.~Wang, and L.~Wu, ``Futures-based resource
  trading and fair pricing in real-time iot networks,'' \emph{IEEE Wireless
  Commun. Lett.}, vol.~9, no.~1, pp. 125--128, 2020.

\bibitem{7476840}
Q.~Li, B.~Niu, and L.-K. Chu, ``Forward sourcing or spot trading? optimal
  commodity procurement policy with demand uncertainty risk and forecast
  update,'' \emph{IEEE Syst. J.}, vol.~11, no.~3, pp. 1526--1536, 2017.

\bibitem{9357934}
R.~Chen, X.~Wang, and X.~Liu, ``Smart futures based resource trading and
  coalition formation for real-time mobile data processing,'' \emph{IEEE Trans.
  Serv. Comput.}, vol.~15, no.~5, pp. 3047--3060, 2022.

\bibitem{MA2019192}
``Examining customer perception and behaviour through social media research –
  an empirical study of the united airlines overbooking crisis,'' \emph{Trans.
  Res. Part E: Logistics Transp. Rev.}, vol. 127, pp. 192--205, 2019.

\bibitem{haynes2020perceptions}
N.~Haynes and D.~Egan, ``The perceptions of frontline employees towards hotel
  overbooking practices: exploring ethical challenges,'' \emph{J. Revenue
  Pricing Manage.}, vol.~19, pp. 119--128, 2020.

\bibitem{9149184}
A.~Adebayo, D.~B. Rawat, and M.~Song, ``Prediction based adaptive rf spectrum
  reservation in wireless virtualization,'' in \emph{IEEE Int. Conf. Commun.
  (ICC)}, 2020, pp. 1--6.

\bibitem{6175013}
K.~Chard and K.~Bubendorfer, ``High performance resource allocation strategies
  for computational economies,'' \emph{IEEE Trans. Parallel Distrib. Syst.},
  vol.~24, no.~1, pp. 72--84, 2013.

\bibitem{9127810}
N.~Sharghivand, F.~Derakhshan, L.~Mashayekhy, and L.~Mohammadkhanli, ``An edge
  computing matching framework with guaranteed quality of service,'' \emph{IEEE
  Trans. Cloud Comput.}, vol.~10, no.~3, pp. 1557--1570, 2022.

\bibitem{9851813}
H.~Fang, Y.~Jia, Y.~Wang, Y.~Zhao, Y.~Gao, and X.~Yang, ``Matching game based
  task offloading and resource allocation algorithm for satellite edge
  computing networks,'' in \emph{Int. Symp. Netw., Comput. and Commun.
  (ISNCC)}, 2022, pp. 1--5.

\bibitem{9108577}
R.~Fantacci and B.~Picano, ``A matching game with discard policy for virtual
  machines placement in hybrid cloud-edge architecture for industrial iot
  systems,'' \emph{IEEE Trans. Ind. Informat.}, vol.~16, no.~11, pp.
  7046--7055, 2020.

\bibitem{9200548}
N.~Raveendran, H.~Zhang, L.~Song, L.-C. Wang, C.~S. Hong, and Z.~Han, ``Pricing
  and resource allocation optimization for iot fog computing and nfv: An epec
  and matching based perspective,'' \emph{IEEE Trans. Mobile Comput.}, vol.~21,
  no.~4, pp. 1349--1361, 2022.

\bibitem{8815852}
Y.~Du, J.~Li, L.~Shi, T.~Liu, F.~Shu, and Z.~Han, ``Two-tier matching game in
  small cell networks for mobile edge computing,'' \emph{IEEE Trans. Serv.
  Comput.}, vol.~15, no.~1, pp. 254--265, 2022.

\bibitem{9344666}
Q.~Tang, Z.~Fei, B.~Li, and Z.~Han, ``Computation offloading in leo satellite
  networks with hybrid cloud and edge computing,'' \emph{IEEE Internet Things
  J.}, vol.~8, no.~11, pp. 9164--9176, 2021.

\bibitem{9214500}
M.~Aazam, S.~u. Islam, S.~T. Lone, and A.~Abbas, ``Cloud of things (cot):
  Cloud-fog-iot task offloading for sustainable internet of things,''
  \emph{IEEE Trans. Sustain. Comput.}, vol.~7, no.~1, pp. 87--98, 2022.

\bibitem{9195499}
L.~Li and H.~Zhang, ``Delay optimization strategy for service cache and task
  offloading in three-tier architecture mobile edge computing system,''
  \emph{IEEE Access}, vol.~8, pp. 170\,211--170\,224, 2020.

\bibitem{9838921}
R.~Zhang and C.~Zhou, ``A computation task offloading scheme based on
  mobile-cloud and edge computing for wbans,'' in \emph{IEEE Int. Conf.
  Commun.(ICC)}, 2022, pp. 4504--4509.

\bibitem{9616429}
S.~Tian, C.~Chang, S.~Long, S.~Oh, Z.~Li, and J.~Long, ``User preference-based
  hierarchical offloading for collaborative cloud-edge computing,'' \emph{IEEE
  Trans. Serv. Comput.}, vol.~16, no.~1, pp. 684--697, 2023.

\bibitem{8289317}
G.~Premsankar, M.~Di~Francesco, and T.~Taleb, ``Edge computing for the internet
  of things: A case study,'' \emph{IEEE Internet Things J.}, vol.~5, no.~2, pp.
  1275--1284, 2018.

\bibitem{8016573}
Y.~Mao, C.~You, J.~Zhang, K.~Huang, and K.~B. Letaief, ``A survey on mobile
  edge computing: The communication perspective,'' \emph{IEEE Commun. Surveys
  Tut.}, vol.~19, no.~4, pp. 2322--2358, 2017.

\bibitem{9667258}
Z.~Gao, W.~Hao, and S.~Yang, ``Joint offloading and resource allocation for
  multi-user multi-edge collaborative computing system,'' \emph{IEEE Trans.
  Veh. Technol.}, vol.~71, no.~3, pp. 3383--3388, 2022.

\bibitem{9687261}
H.~Lee, H.~Lee, S.~Jung, and J.~Kim, ``Stable marriage matching for
  traffic-aware space-air-ground integrated networks: A gale-shapley
  algorithmic approach,'' in \emph{Int. Conf. Inform. Netw. (ICOIN)}, 2022, pp.
  474--477.

\bibitem{lai2018optimal}
P.~Lai, Q.~He, M.~Abdelrazek, F.~Chen, J.~Hosking, J.~Grundy, and Y.~Yang,
  ``Optimal edge user allocation in edge computing with variable sized vector
  bin packing,'' in \emph{Int. Conf. Service-Oriented Comput. (ICSOC)}.\hskip
  1em plus 0.5em minus 0.4em\relax Springer, 2018, pp. 230--245.

\bibitem{10321730}
H.~Qi, M.~Liwang, S.~Hosseinalipour, X.~Xia, Z.~Cheng, X.~Wang, and Z.~Jiao,
  ``Matching-based hybrid service trading for task assignment over dynamic
  mobile crowdsensing networks,'' \emph{IEEE Trans. Serv. Comput.}, pp. 1--14,
  2023.

\bibitem{9682584}
M.~Liwang, R.~Chen, X.~Wang, and X.~Shen, ``Unifying futures and spot market:
  Overbooking-enabled resource trading in mobile edge networks,'' \emph{IEEE
  Trans. Wireless Commun.}, vol.~21, no.~7, pp. 5467--5485, 2022.

\bibitem{9127160}
S.~Luo, X.~Chen, Q.~Wu, Z.~Zhou, and S.~Yu, ``Hfel: Joint edge association and
  resource allocation for cost-efficient hierarchical federated edge
  learning,'' \emph{IEEE Trans. Wireless Commun.}, vol.~19, no.~10, pp.
  6535--6548, 2020.

\bibitem{9301243}
S.~Hosseinalipour, A.~Rahmati, D.~Y. Eun, and H.~Dai, ``Energy-aware stochastic
  uav-assisted surveillance,'' \emph{IEEE Trans. Wireless Commun.}, vol.~20,
  no.~5, pp. 2820--2837, 2021.

\end{thebibliography}
\end{spacing}

	\vspace{-1.0cm}
\begin{IEEEbiography}[{\includegraphics[width=1in,height=1.25in,clip,keepaspectratio]{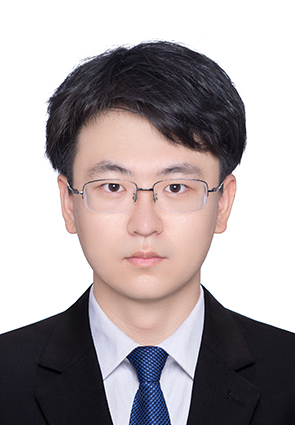}}]{Houyi Qi} received his B.S. degree in electronic information engineering from Zhengzhou University, China, in 2021. He is currently working toward the M.S. degree in School of Informatics, Xiamen University, China. His research interests include mobile crowdsensing networks, matching theory and cloud/edge/service computing.
\end{IEEEbiography}

\vspace{-0.79cm}
\begin{IEEEbiography}[{\includegraphics[width=1in,height=1.25in,clip,keepaspectratio]{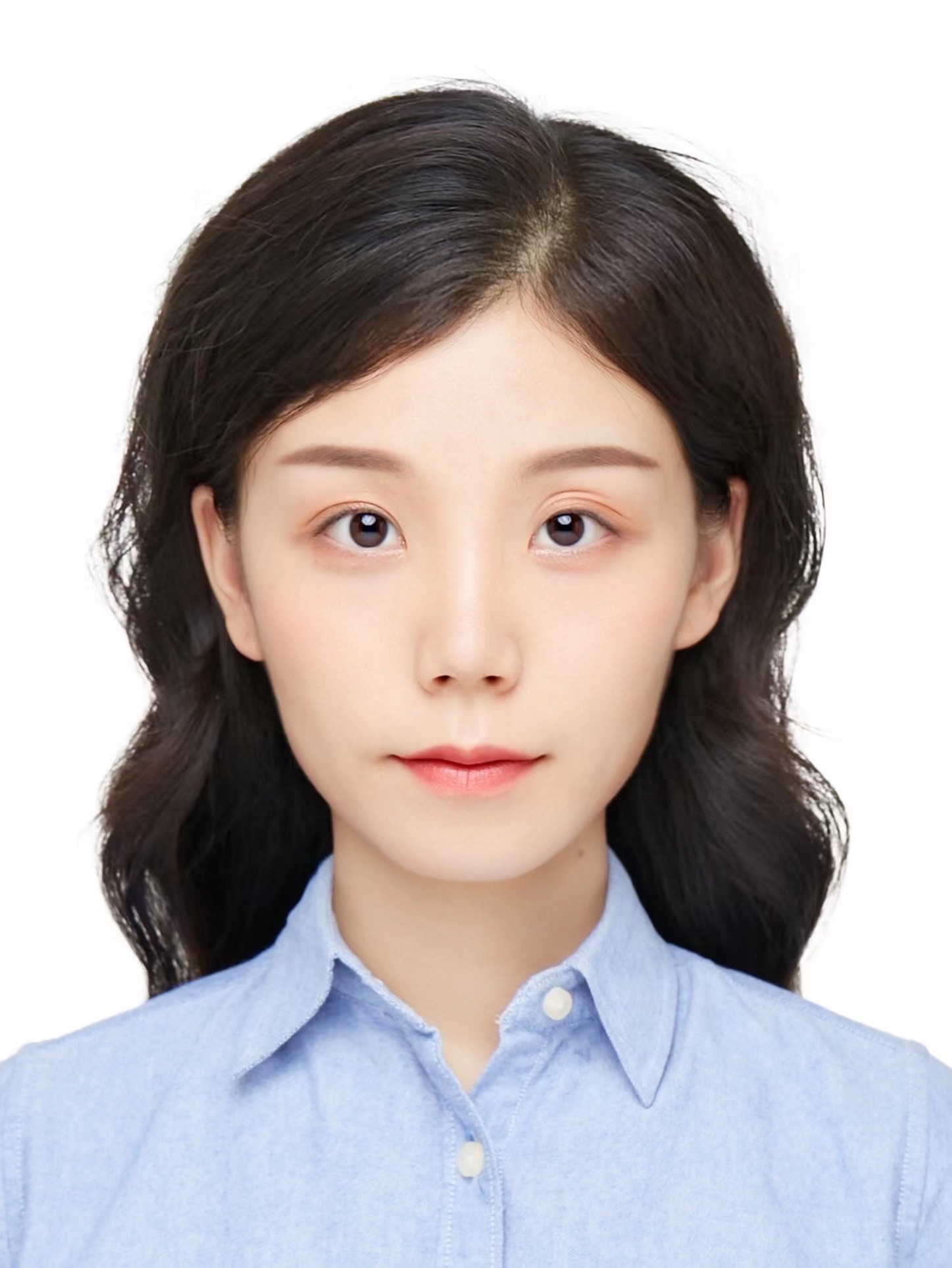}}]{Minghui Liwang (Member, IEEE)} received her Ph. D. degree in School of Informatics, Xiamen University, China, in 2019. She is currently an assistant professor in School of Informatics, Xiamen University, China. Her research interests include wireless communication systems, Internet of Things, cloud/edge/service computing, federated learning as well as economic models and applications in wireless communication networks.
\end{IEEEbiography}

\vspace{-0.79cm}
\begin{IEEEbiography}[{\includegraphics[width=1in,height=1.25in,clip,keepaspectratio]{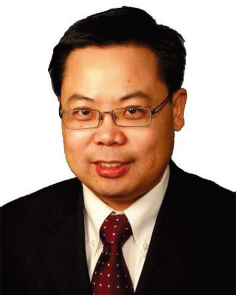}}]{Xianbin Wang (Fellow, IEEE)} is a professor and a Tier-1 Canada Research Chair in 5G and Wireless IoT Communications with Western University, Canada. His current research interests include 5G/6G technologies, Internet of Things, communications security, machine learning, and intelligent communications. He is a Fellow of the Canadian Academy of Engineering and a Fellow of the Engineering Institute of Canada. He has received many prestigious awards and recognitions, including the IEEE Canada R.A. Fessenden Award, Canada Research Chair, Engineering Research Excellence Award with Western University, Canadian Federal Government Public Service Award, Ontario Early Researcher Award, and nine Best Paper Awards. 
\end{IEEEbiography}

\vspace{-0.79cm}
\begin{IEEEbiography}[{\includegraphics[width=1in,height=1.25in,clip,keepaspectratio]{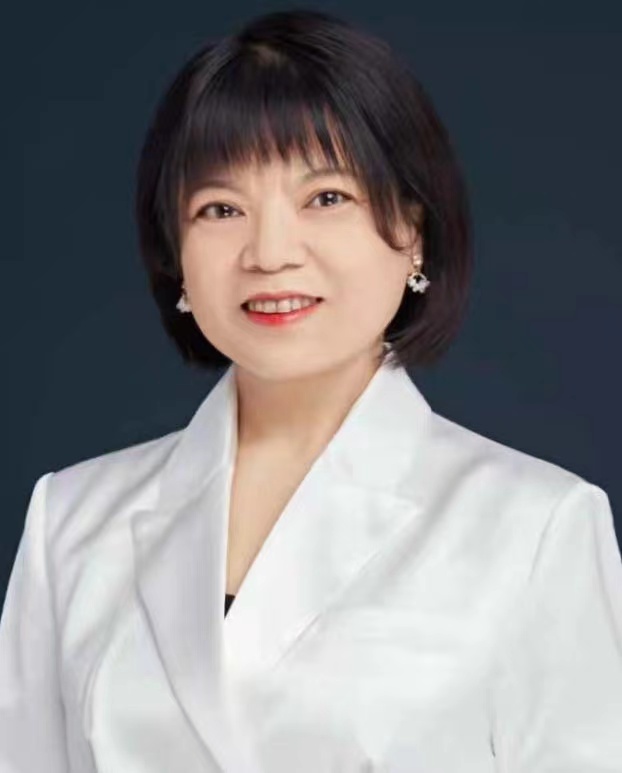}}]{Li Li (Member, IEEE)} received the B.S. and M.S. degrees in electrical automation from Shengyang Agriculture University, Shengyang, China, in 1996 and 1999, respectively, and the Ph.D. degree in mechatronics engineering from the Shenyang Institute of Automation, Chinese Academy of Science, Shenyang, in 2003. She then joined at Tongji University, Shanghai, China, where she is currently a Professor of control science and engineering. She has over 50 publications, including ﬁve books, over 30 journal articles, and two book chapters. Her current research interests include data-based modeling and optimization, computational intelligence, and machine learning.
\end{IEEEbiography}

\vspace{-0.79cm}
\begin{IEEEbiography}[{\includegraphics[width=1in,height=1.25in,clip,keepaspectratio]{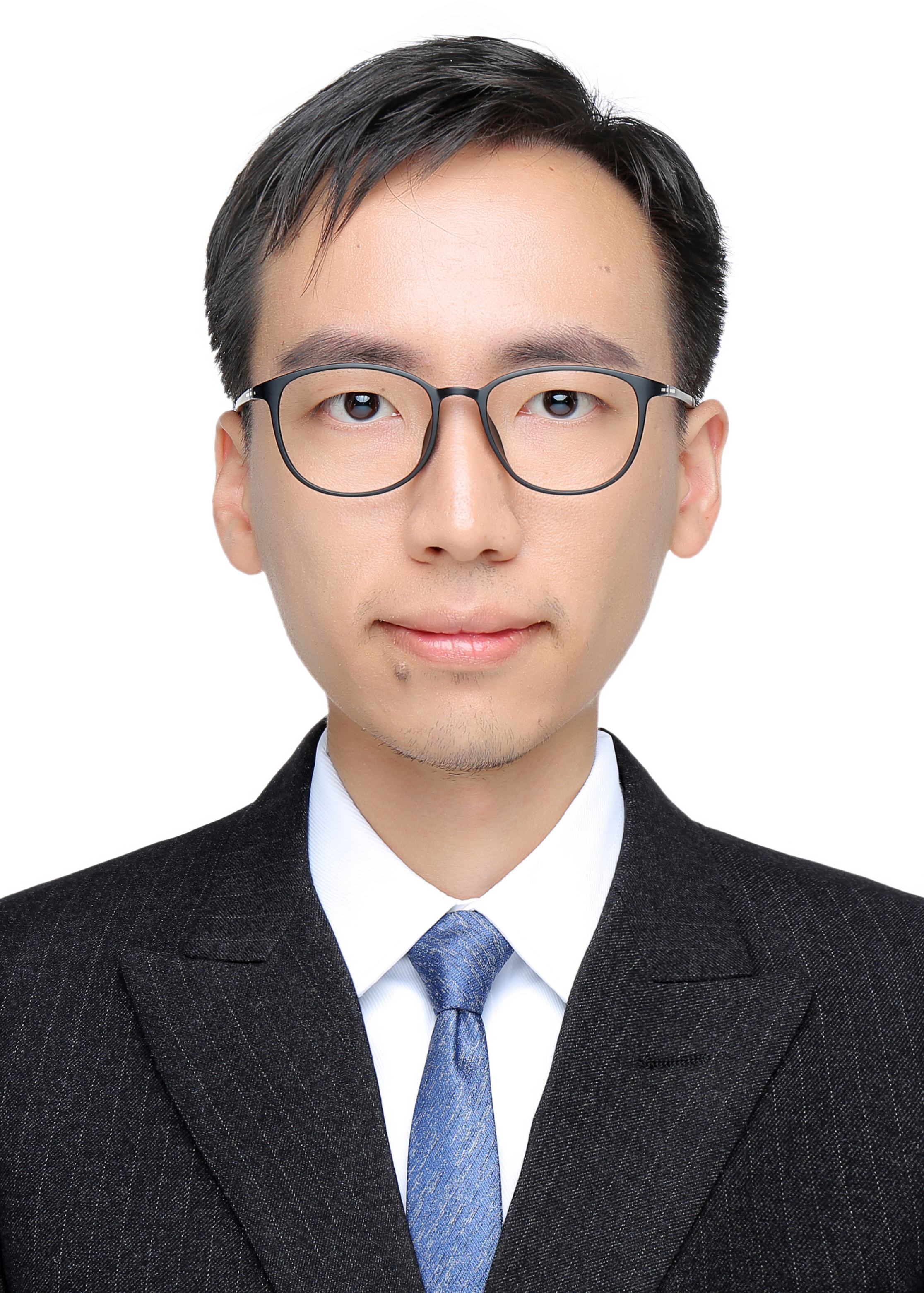}}]{Wei Gong (Member, IEEE)} is an assistant professor with the Department of Control Science and Engineering, Tongji University. He received his Ph.D. degree in computer science at the University of Chinese Academy of Sciences. He was a postdoctoral fellow with the Department of Electrical and Computer Engineering, Western University, Canada. His research interests include wireless networking, edge intelligence, and distributed learning.
\end{IEEEbiography}

\vspace{-0.79cm}
\begin{IEEEbiography}[{\includegraphics[width=1in,height=1.25in,clip,keepaspectratio]{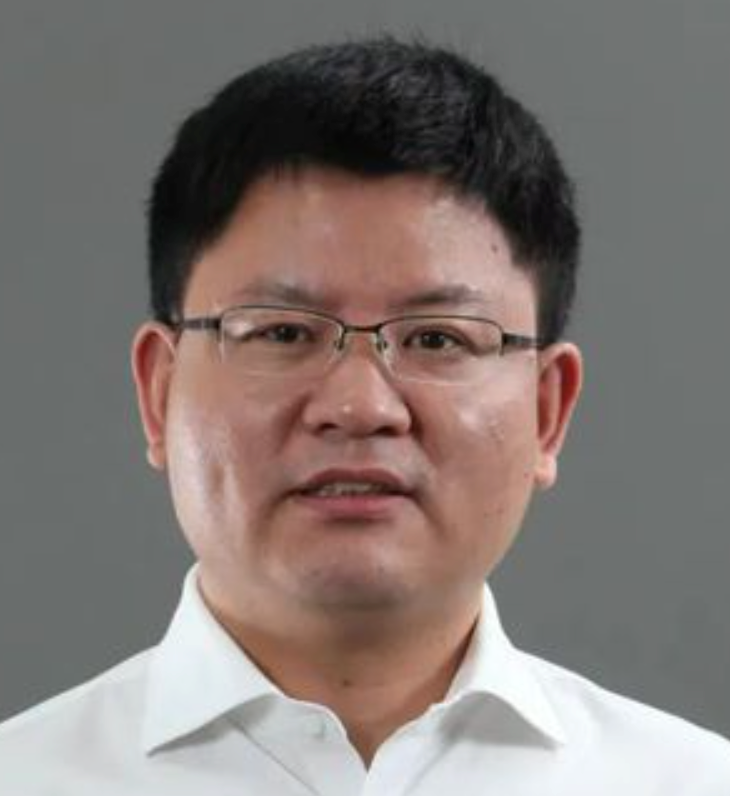}}]{Jian Jin} is working with Industrial Internet and IOT Research Institute, China Academy of Information and Communications Technology, Beijing, China. His research interests are network identity and protocol, blockchain systems, and etc.
\end{IEEEbiography}

\vspace{-0.79cm}
\begin{IEEEbiography}[{\includegraphics[width=1in,height=1.25in,clip,keepaspectratio]{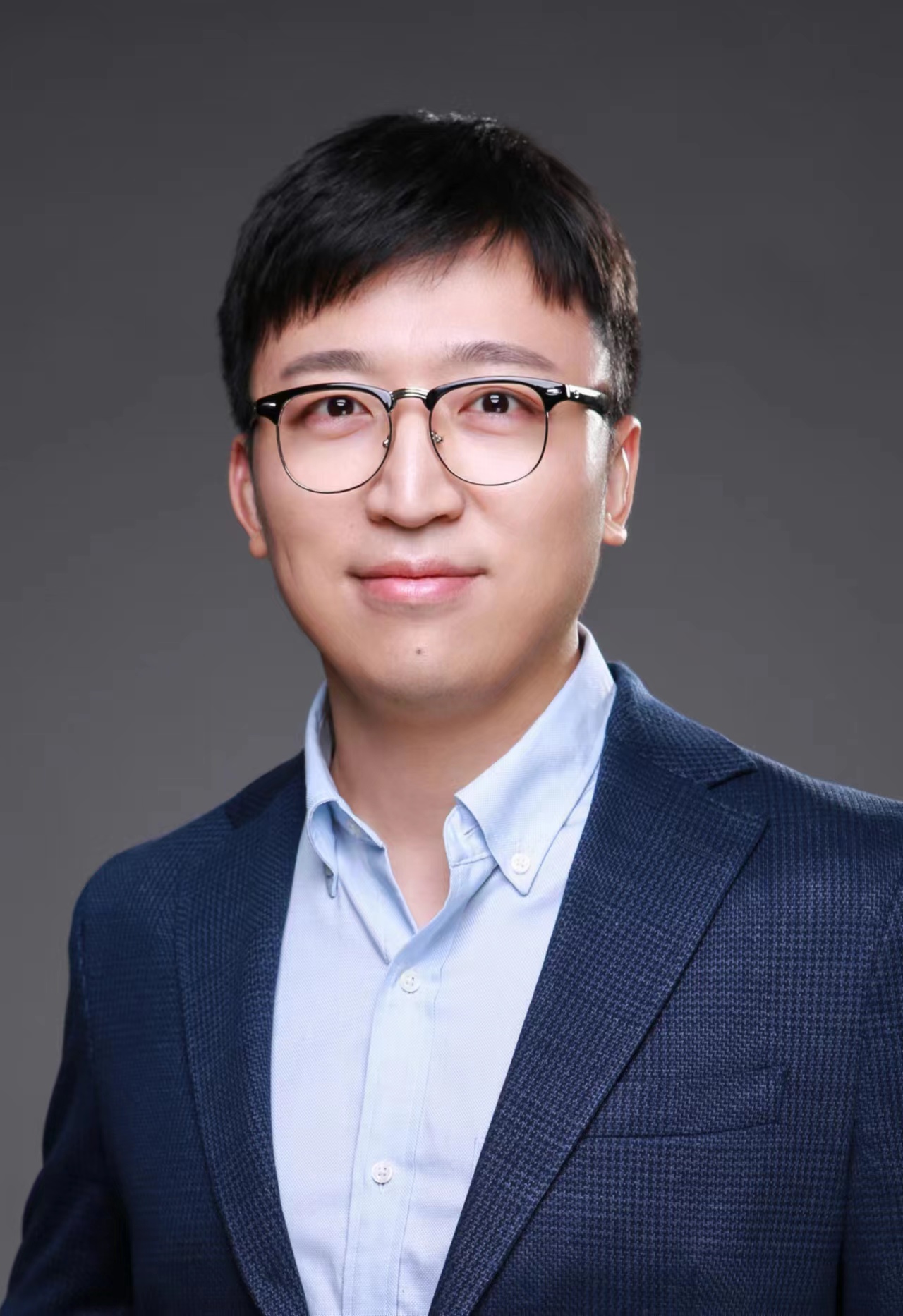}}]{Zhenzhen Jiao} received the Ph.D. degree from University of Chinese Academy of Sciences in computer science. He was an associate professor at the Institute of Computing Technology, Chinese Academy of Sciences and the director of a blockchain research laboratory. He is currently the director of the Teleinfo iF-Labs in China Academy of Information and Communications Technology. His research interests include blockchain, self-organized systems and networks, and computing power network.
\end{IEEEbiography}

\vfill

\newpage
\clearpage
\appendices
\section{Key Notations}
Key notations in this paper are summarized in Table 4.
\begin{table*}[b!]{\footnotesize\vspace{-0.3cm}
	\caption{{\footnotesize KEY NOTATIONS}}\vspace{-0.4cm}
	\begin{center}
	\begin{tabular}{ll}
		\hline
		\multicolumn{1}{l}{Notation} & \multicolumn{1}{c}{Explanation} \\ \hline
$ \bm{B} $, $ \bm{S^{E}} $, $ \bm{S^{C}} $ & The set of MUs, ESs, and CSs in OA-CLM \\
$ \bm{B^\prime} $, $ \bm{S^{E\prime}} $, $ \bm{S^{C\prime}} $ & The set of MUs, ESs, and CSs in OS-CLM \\
$ b_i $, $ s_j^E $, $ s_k^C $  &The $  {i}^\text{th} $ MU in $ \bm{B} $, $ \bm{B^\prime} $, the $ {j}^\text{th} $ ES in $ \bm{{S}^{ {E}}} $, $ \bm{S^{E\prime}} $, the $ {k}^\text{th} $ CS in $ \bm{{S}^{ {C}}} $, $ \bm{S^{C\prime}} $\\
$ \alpha_{i} $& Random variable that describes the participation of MU $ b_i $ \\
$ r_i^U $	&Required computing resources of $ b_i $ (CPU cycles) \\
$ f_i^U $, $ f_j^E $, $ f_k^C $& Computing capability of $ b_i $, $ s_j^E $, $ s_k^C $ (CPU cycles/s) \\ 
$ e_i^{t} $ & Transmission power of $ b_i $ (Watt)\\
$ e_i^U $,$ e_j^E $, $ e_k^C $ & Local power consumption of $ b_i $, $ s_j^E $, $ s_k^C $ (Watt)\\
$ \varepsilon_k $ & Random variable that describes the participation of inherent requestors of $s_k^C$\\
$ d_i^U $ & Data size of the task of $b_i$ (bit) \\
$ \gamma_{i,j} $ & Random variable that describes the time-varying channel quality between $ b_i $ and $ s_j^E $ \\
$ \lambda_{i,j} $& Indicator of whether $ b_i $ is selected as a volunteer by $ s_j^E $ \\
$G_j^E, G_k^C $& Number of resources (e.g., VMs) owned by $ s_j^E $, and $ s_k^C $ \\
$ K_j^E $ & Number of orthogonal subcarriers of $ s_j^E $ \\
$ c_{i,j}^E,c_{i,j,k}^C $ & Service cost of $ s_j^E $ for serving $b_i$, service cost of $ s_k^C $ for processing $ b_i $'s task offloaded by $ s_j^E $\\
$c_j^{H}$, $c_k^{H}$ & Unit hardware cost (per VM) of $ s_j^E $, $ s_k^C $\\
$x_j$& Necessity of $ s_j^E $ for purchasing cloud resources \\
$ \beta_{i,j,k} $& Indicator of whether $ s_j^E $ will break the contract with CS $ s_k^C $ about $b_i$ in the spot market \\
$ \varepsilon_k $ & Random variable that describes the inherent requestors' resource demand of $ s_k^C $'s \\
$\vartheta_k $ & Indicator of whether $ s_k^C $ has sufficient VMs to fulfill forward contracts in a practical transaction\\
$ p_{i,j}^{U\rightarrow E},p_{i,j,k}^{E\rightarrow C},p_{k}^{I\rightarrow C} $& Payment from $ b_i $ to $ s_j^E $, from $ s_j^E $ to CS $ s_k^C $ about the task of $b_i$, from an inherent requestor to CS $ s_k^C $ in OA-CLM\\
$ p_{i,j}^{U\prime\rightarrow E\prime},p_{i,j,k}^{E\prime\rightarrow C\prime} $& Payment from $ b_i $ to $ s_j^E $, from $ s_j^E $ to CS $ s_k^C $ about the task of $b_i$ in OS-CLM \\
$ q^{U\rightarrow E} ,q^{E\rightarrow U}$& Penalty from a MU to an ES, from an ES to a MU\\
$ q^{E\rightarrow C},q^{I\rightarrow C}$& Penalty from an ES to a CS (per task), from inherent requestor to a CS\\
$ {R_{1}^{U}},{R_{2}^{U}},{R_{1}^{E}},{R_{2}^{E}},{R^{C}} $& Risks associated with MUs, ESs, CSs \\ 
$ \mathrm{\Delta}p_i,\mathrm{\Delta}p_j$ & Step factor on payment rising regarding $ b_i $, $ s_j^E $ \\
\textbf{\text{B}},\textbf{\text{U}},\textbf{\text{P}}& Bernoulli distribution, Uniform distribution, Poisson distribution
\\
 \hline
	\end{tabular}
\end{center}}
\end{table*}

\section{Derivations associated with OA-CLM}
\subsection{Derivations associated with MUs}
\textbf{Mathematical expectation of $ v_{i,j} $.} 
 $ \gamma_{i,j} $ refers to the time-varying wireless channel quality between $ b_i $ and ES $s_j^E$, following an uniform distribution denoted by $ \gamma_{i,j} \sim \mathbf{U}(\mu_1,\mu_2) $. The expectation of $\gamma_{i,j}$ can simply computed as $  \text{E}[\gamma_{i,j}]=\frac{\mu_1+\mu_2}{2} $. According to (1), (2), and (3), we can obtain the expectation of $v_{i,j}$ (denoted by $\text{E}[v_{i,j}]$) as
\begin{equation}\label{key}{\small
	\begin{aligned}
	& \text{E}[v_{i,j}]= \mathbb{V}_1\left\{\frac{r_i^U}{f_i^U}-\left(\frac{r_i^U}{f_j^E}+\frac{d_i^U}{W\text{log}_2\left(1+e_i^{t}\frac{(\mu_1+\mu_2)}{2}\right)}\right)\right\}\\&+\mathbb{V}_2\left\{\frac{r_i^Ue_i^U}{f_i^U}-\frac{d_i^Ue_i^{t}}{W\text{log}_2\left(1+e_i^{t}\frac{(\mu_1+\mu_2)}{2}\right)}\right\}
\end{aligned}}
\end{equation}

\textbf{Mathematical expectation of $\lambda_{i,j}$.} 
As resource overbooking can lead to the case where a contractual MU (e.g., a MU who has signed a forward contract with an ES) can be selected as a volunteer, we use $\lambda_{i,j}$ to indicate whether MU $ b_i $ is determined as a volunteer by ES $ s_j^E $ in each practical transaction. For analytical simplicity, let $X_1= G_j^E+ \sum_{s_k^C\in\varphi\left(s_j^E\right)} |\mathbb{B}_{j,k}| $, which is the maximum number of MUs that $s_j^E$ can serve in spot market.

Due to the uncertain resource demand of MUs in the considered market, assessing the probability of a MU being selected as a volunteer needs a large amount of calculations. We first use notation $\bm{\mathcal{M}_j}=\left\{\bm{M_1},..,\bm{M_n},...,\bm{M_{|\mathcal{M}_j|}}\right\}$ to collect all the possible cases of the participation of MUs' in set $\omega\left(s_j^E\right)$ participation in the spot market, where $\bm{M_n}=\left\{\alpha_1,...,\alpha_i,...,\alpha_{|\omega\left(s_j^E\right)|}\right\}$ is a vector and denoted as the ${n}^\text{th}$ case of MUs' participation in a transaction. For example, suppose that $s_j^E$ has pre-signed forward contracts with two MUs, all the possible cases of these MUs taking part in the spot market can be expressed as $\bm{\mathcal{M}_j}=\left\{\bm{M_1},\bm{M_2},\bm{M_3},\bm{M_4}\right\}=\left\{\left\{0,0\right\},\left\{0,1\right\},\left\{1,0\right\},\left\{1,1\right\}\right\}$. Besides, let $X_2=\sum_{\alpha_i\in \bm{M_n}}\alpha_{i}$ refer to the number of contractual MUs who attend in a practical transaction. According to $X_2$, the calculation of $\text{E}[\lambda_{i,j}]$ can involve the following two conditions:

\noindent
$\bullet$ \textit{the number of participated MUs in a practical transaction is lower than or equal to the resource supply of $s^E_j$ (i.e., $X_2\le X_1$)}, we have $\text{E}[\lambda_{i,j}]=0$;

\noindent
$\bullet$ \textit{the number of participated MUs in a practical transaction exceeds the resource supply (i.e., $X_2>X_1$).} Due to the selfishness and resource constraint, an ES will consistently choose MUs that contribute the lowest expected utility as volunteers. Specifically, the expected utility that $ b_i\in \bm{B} $ brings to $ s_j^E $ can be calculated as
\begin{equation}\label{key}{\small
	\begin{aligned}
		\overline{u^U}(s_j^E,b_i)=a_{i}\left( \text{E}[v_{i,j}]-p_{i,j}^{U\rightarrow E}\right)+(1-a_{i})q^{U\rightarrow E},
	\end{aligned}}
\end{equation}
where $  \text{E}[v_{i,j}]$ is given by (40). For notational simplicity, we use a set $\bm{\mathcal{M}^U_{i,j}}$ to denote all the cases that $b_i$ attends in a transaction, which, unfortunately, is selected as a volunteer by $s^E_j$ (e.g., $b_i$ is one of the ($ X_2-X_1$) MUs that bring the lowest expected utility to $s_j^E$).

Accordingly, the probability of MU $ b_i $ being determined by $s_j^E$ as a volunteer is given by
\begin{equation}\label{key}{\small
	\begin{aligned}
	\text{Pr}(\lambda_{i,j}=1)=\sum_{\bm{M_n}\in\bm{\mathcal{M}^U_{i,j}}}{\prod_{\alpha_i\in \bm{M_n}}}{\text{E}[\alpha_i]},		
	\end{aligned}}
\end{equation}
and the expected value of $\lambda_{i,j}$ can be calculated by
\begin{equation}{\small
	\begin{aligned}
	&\text{E}[\lambda_{i,j}]=\text{Pr}\left(\lambda_{i,j}=0\right)\times 0+\text{Pr}\left(\lambda_{i,j}=1\right)\times 1\\&=\text{Pr}(\lambda_{i,j}=1)=\sum_{\bm{M_n}\in\bm{\mathcal{M}^U_{i,j}}}{\prod_{\alpha_i\in \bm{M_n}}}{\text{E}[\alpha_i]},
\end{aligned}}
\end{equation}
where $\text{E}[\alpha_i]=a_i$.

\textbf{Derivation associated with (5).}
The $\text{E}[v_{i,j}]$ and $\text{E}[\lambda_{i,j}]$ of (5) are given by (40) and (43), respectively.

\textbf{Derivations associated with (28c) and (28d).}
In optimization problem $ \bm{\mathcal{F}^U} $ given by (28), constraint (28c) represents a probabilistic expression, making its close form non-trivial to be obtained. To resolve such an issue, we transform (28c) into a tractable one by exploiting a set of bounding techniques. First, (28c) can be rewritten as
\begin{equation}\label{key}{\small
	\begin{aligned}
			R^{U}_1\left( b_{i},\omega \left( b_{i} \right) \right)\leq \rho_1 \Rightarrow \text{Pr}\left(u(b_i,\omega\left(b_i\right))\geq u_{min}\right)>1-\rho_1.
	\end{aligned}}
\end{equation}

To obtain a tractable form for (44), we can have the upper-bound of its left-hand side by using Markov inequality \cite{9301243}, as the following (45).
		\begin{equation}\label{key}{\small
		\begin{aligned}
		\text{Pr}\left(u(b_i,\omega\left(b_i\right)\geq u_{min}\right)\leq \frac{ \text{E}[u(b_i,\omega\left(b_i\right))]}{u_{min}}	
	\end{aligned}}
\end{equation}
Combining (44) and (45), we can then get a tractable form for (28c):
\begin{equation}\label{key}{\small
	\begin{aligned}
			\frac{ \text{E}[u(b_i,\omega\left(b_i\right))]}{u_{min}}>1-\rho_1,
	\end{aligned}}
\end{equation}
where the value of $\text{E}[u(b_i,\omega\left(b_i\right))]$ is given by (5).
Besides, constraint (28d) can be rewritten as
\begin{equation}\label{key}{\small
	\begin{aligned}
			R^{U}_2\left( b_{i},\omega \left( b_{i} \right) \right)\leq \rho_2 \Rightarrow \text{Pr}(\lambda_{i,j}=1)\leq \rho_2,
	\end{aligned}}
\end{equation}
where $\text{Pr}(\lambda_{i,j}=1)$ is detailed by (42).
\subsection{Derivations associated with ESs}
\textbf{Mathematical Expectation of $\beta_{i,j,k}$.}
Note that $ \text{Pr}(\beta_{i,j,k}=1) $ stands for the possibility of ES $ s_j^E $ fulfills the contract with CS $ s_k^C $ about the task of $b_i \in \mathbb{B}_{j,k}$ during each practical transaction (i.e., (9)). Apparently, we are unaware of the specific tasks that can be offloaded to CSs for process during this time (e.g., during forward contract design), we use $ r^{max}_j=\max \left\{r_i^U~|~\forall b_i \in \omega \left(s_j^E\right)\right\} $ to estimate the amount of resources (per offloaded task) required by $ s_j^E $ from CSs. Accordingly, the expectation of $\beta_{i,j,k}$ depends on both trading between MUs and ESs, as well as that between ESs and CSs. Accordingly, the resource demand of MUs associated with $s_j^E$ is generally uncertain and can be lower than its resource supply during a practical transaction. In this case, $ s_j^E$ will break some contracts with its matched CSs to avoid further losses and wasting resources. Given the selfishness of $ s_j^E $, it will give up some contracts with CSs that bring lower values of its expected utility.

To facilitate better analysis, we use $\bm{\mathcal{M}^E_{j,k}}$ to denote the set of all the cases that can lead to $s_j^E$ breaks the contract with CS $s_k^C$ about $b_i$'s task in a transaction. Note that tasks in $ \mathbb{B}_{j,k}$ are supposed to be homogeneous during this time (due to the same $r_j^{max}$), so we randomly give up a contract with CS $s_k^C$ regarding $b_i$ (i.e., $\beta_{i,j,k}=0$).
Accordingly, the probability of $s_j^E$ breaks the contract with CS $s_k^C$ regarding $b_i$ in a practical transaction is expressed as
		\begin{equation}\label{key}{\small
		\begin{aligned}
			\text{Pr}(\beta_{i,j,k}=0)=\sum_{\bm{M_n}\in\bm{\mathcal{M}^E_{j,k}}}{\prod_{\alpha_i\in \bm{M_n}}}{\text{E}[\alpha_i]},
	\end{aligned}}
\end{equation}
while we can thereby have the expected value of $\beta_{i,j,k}$:
\begin{equation}{\small
	\begin{aligned}
		&\text{E}[\beta_{i,j,k}]=\text{Pr}\left(\beta_{i,j,k}=0\right)\times 0+\text{Pr}\left(\beta_{i,j,k}=1\right)\times 1\\&=\text{Pr}(\beta_{i,j,k}=1)=1-\text{Pr}(\beta_{i,j,k}=0)\\&=1-\sum_{\bm{M_n}\in\bm{\mathcal{M}^E_{j,k}}}{\prod_{\alpha_i\in \bm{M_n}}}{\text{E}[\alpha_i]}
	\end{aligned}}
\end{equation}

\textbf{Derivations associated with (29f) and (29g).}
According to (48), (29f) can be rewritten as 
\begin{equation}\label{key}{\small
	\begin{aligned}
	&R_1^E\left(s_j^E,\omega\left(s_j^E\right)\right) \leq\rho_3 \Rightarrow \text{Pr}\left(\beta_{i,j,k}=0\right) \leq\rho_3 \\& \Rightarrow \sum_{\bm{M_n}\in\bm{\mathcal{M}^E_{j,k}}}{\prod_{\alpha_i\in \bm{M_n}}}{\text{E}[\alpha_i]} \leq\rho_3,
\end{aligned}}
\end{equation}
and we use $\bm{\mathcal{M}^O_{i,j}}$ to denote all the cases where the number of MUs in $\omega(s_j^E)$ who have practically participated in a transaction exceeds the resource supply of $s_j^E$.
Thus, constraint (29g) can be reformulated as
\begin{equation}\label{key}{\small
	\begin{aligned}
		&R_2^E\left(s_j^E,\omega\left(s_j^E\right)\right) \leq \rho_4 \Rightarrow \\& \text{Pr}\left(\sum_{b_i\in\omega\left(s_j^E\right)}\alpha_i>\sum_{s_k^C\in\varphi\left(s_j^E\right)} |\mathbb{B}_{j,k}|+G_j^E\right)\leq \rho_4 \\& \Rightarrow \sum_{\bm{M_n}\in\bm{\mathcal{M}^O_{i,j}}}{\prod_{\alpha_i\in\bm{M_n}}}{\text{E}[\alpha_i]} \leq\rho_4
	\end{aligned}}
\end{equation}

\subsection{Derivations associated with CSs}
\textbf{Mathematical expectation of $ \varepsilon_k$.}
$\varepsilon_k$ refers to the resource demand of $s_k^C$'s inherent requestors, following a Poisson distribution denoted by $ \varepsilon_k\sim \mathbf{P}(\sigma_k) $ with mathematical expectation $ \sigma_k $, the expected value of $\varepsilon_k$ is simply expressed as	$\text{E}[\varepsilon_k]=\sigma_k$.

\textbf{Mathematical expectations of $ N_k$.}
According to (49) and $\text{E}[\varepsilon_k]$, we can show the expectation of $N_k $ as
\begin{equation}\label{key}{\small
	\begin{aligned}
		 &\text{E}[N_k]=\sum_{s_j^E\in\varphi\left(s_k^C\right)}\sum_{b_i\in \mathbb{B}_{j,k}} \text{E}[\beta_{i,j,k}]+ \text{E}[\varepsilon_k] - G_k^C\\&=\sum_{s_j^E\in\varphi\left(s_k^C\right)}\sum_{b_i\in \mathbb{B}_{j,k}} \left(1-\sum_{\bm{M_n}\in\bm{\mathcal{M}^E_{j,k}}}{\prod_{\alpha_i\in \bm{M_n}}}{\text{E}[\alpha_i]}\right)\\&+ \sigma_k - G_k^C 
	\end{aligned}}
\end{equation}

\textbf{Mathematical Expectation of $ \vartheta_k$.}
According to (17) and (18), we rewrite (17) as
\begin{equation}\label{key}{\small
	\vartheta_k = \left\{ \begin{matrix}
		{0~~~~~, N_k \leq 0} \\
		{1~~~~~, N_k>0} \\
	\end{matrix} \right.  }
\end{equation}
while its expectation can be computed as (54), the value of $\text{Pr}(\varepsilon_k=h)$ can be calculated as (55).
\begin{figure*}[t!]
		\centering
	\begin{equation}{\small
			\begin{aligned}
				&\text{E}[\vartheta_k]=0\times \text{Pr}(\vartheta_k=0)+1\times \text{Pr}(\vartheta_k=1)=\text{Pr}(\vartheta_k=1)=\text{Pr}(N_k>0)
				\\&=\text{Pr}\left(\sum_{s_j^E\in\varphi\left(s_k^C\right)}\sum_{b_i\in \mathbb{B}_{j,k}}\beta_{i,j,k}>G_k^C\right)\text{Pr}(\varepsilon_k=0)+\text{Pr}\left(\sum_{s_j^E\in\varphi\left(s_k^C\right)}\sum_{b_i\in \mathbb{B}_{j,k}}\beta_{i,j,k}>G_k^C-1\right)\text{Pr}(\varepsilon_k=1)+...\\&+\text{Pr}\left(\sum_{s_j^E\in\varphi\left(s_k^C\right)}\sum_{b_i\in \mathbb{B}_{j,k}}\beta_{i,j,k}>G_k^C-h\right)\text{Pr}(\varepsilon_k=h)+\text{Pr}\left(\sum_{s_j^E\in\varphi\left(s_k^C\right)}\sum_{b_i\in \mathbb{B}_{j,k}}\beta_{i,j,k}>0\right)\text{Pr}(\varepsilon_k=G_k^C)
				\\&=\sum_{h=0}^{G_k^C}\text{Pr}\left(\sum_{s_j^E\in\varphi\left(s_k^C\right)}\sum_{b_i\in \mathbb{B}_{j,k}}\beta_{i,j,k}>G_k^C-h\right)\text{Pr}(\varepsilon_k=h)
		\end{aligned}}
	\end{equation}
		\hrulefill
\end{figure*}
\begin{equation}\label{key}{\small
	\begin{aligned}
			\text{Pr}(\varepsilon_k=h)=e^{-\varepsilon_k}\frac{\varepsilon_k^h}{h!} ,~h\in [0,G_k^C].
	\end{aligned}}
\end{equation}

Accordingly, we use $\bm{\mathcal{M}^C_{j,k}\left\langle h \right\rangle}$ to denote all the cases where the resource demand of ESs in $\varphi(s_k^C)$ in a transaction is equal to $G_k^C-h$. Thus, we can obtain
\begin{equation}{\small
\begin{aligned}
	&\text{Pr}\left(\sum_{s_j^E\in\varphi\left(s_k^C\right)}\sum_{b_i\in \mathbb{B}_{j,k}}\beta_{i,j,k}=G_k^C-h\right)\\&=\sum_{\bm{M_n}\in \bm{\mathcal{M}^C_{j,k}\left\langle h \right\rangle}}{\prod_{\alpha_i\in \bm{M_n}}}{\text{E}[\alpha_i]},
\end{aligned}}
\end{equation}
and $\text{Pr}\left(\sum_{s_j^E\in\varphi\left(s_k^C\right)}\sum_{b_i\in \mathbb{B}_{j,k}}\beta_{i,j,k}>G_k^C-h\right)$ can further be calculated as
\begin{equation}{\small
	\begin{aligned}
	&\text{Pr}\left(\sum_{s_j^E\in\varphi\left(s_k^C\right)}\sum_{b_i\in \mathbb{B}_{j,k}}\beta_{i,j,k}>G_k^C-h\right)\\&=\sum_{h^\prime=0}^{h-1}\text{Pr}\left(\sum_{s_j^E\in\varphi\left(s_k^C\right)}\sum_{b_i\in \mathbb{B}_{j,k}}\beta_{i,j,k}=G_k^C-h^\prime\right)\\&=\sum_{h^\prime=0}^{h-1}~\sum_{\bm{M_n}\in \bm{\mathcal{M}^C_{j,k}\left\langle h^\prime \right\rangle}}{\prod_{\alpha_i\in \bm{M_n}}}{\text{E}[\alpha_i]},
\end{aligned}}
\end{equation}

To this end, the mathematical expectations of $ \vartheta_k$ can be calculated as
\begin{equation}\label{key}\small{
	\begin{aligned}
	& \text{E}[\vartheta_k]=\text{Pr}(N_k>0)=\\&=\sum_{h=0}^{G_k^C}\text{Pr}\left(\sum_{s_j^E\in\varphi\left(s_k^C\right)}\sum_{b_i\in \mathbb{B}_{j,k}}\beta_{i,j,k}>G_k^C-h\right)\text{Pr}(\varepsilon_k=h)\\&=\sum_{h=0}^{G_k^C}e^{-\varepsilon_k}\frac{\varepsilon_k^h}{h!}\sum_{h^\prime=0}^{h-1}~\sum_{\bm{M_n}\in \bm{\mathcal{M}^C_{j,k}\left\langle h^\prime \right\rangle}}{\prod_{\alpha_i\in \bm{M_n}}}{\text{E}[\alpha_i]}.
\end{aligned}}
\end{equation}
\section{Property Analysis of OA-CLM}
\begin{lem}
	(Convergence of MU-ES matching of OA-CLM) Phase 1 of Algorithm 1 converges within finite rounds.
\end{lem}
\begin{proof}
As the MU-ES matching refers to a M2O matching, we utilize the Gale-Shapley algorithm to solve the matching game \cite{9667258,9687261}. After a finite number of rounds, each MU's payment can either be accepted or reach its maximum payment while considering an acceptable level of risk $R^U_1$ (e.g., line 16, Algorithm 1), which thereby supports the property of convergence.
\end{proof}
\begin{lem}
	(Convergence of ES-CS matching of OA-CLM) Phase 2 of Algorithm 1 converges within finite rounds.
\end{lem}
\begin{proof}
	Since the ES-CS matching refers to as a M2M matching, we adopt the Gale-Shapley algorithm to solve the matching game \cite{9667258,9687261}. After a finite number of rounds, each ES's payment can be either accepted or reaches its maximum payment (e.g., line 36, Algorithm 1), and thus supports the property of convergence.
\end{proof}
\begin{lem}
(Individual rationality of MU-ES matching of OA-CLM) All the MUs and ESs are individual rational in the futures market.
\end{lem}
\begin{proof}
We offer the analysis on proving the individual rationality of MUs and ESs.

\textbf{Individual rationality of MUs.} For each MU $ b_i\in\bm{B} $, the payment from $ b_i $ remains unchanged when any of the following conditions is met: \textit{i)} $ b_i $ has been accepted by an ES $ s_j^E $; \textit{ii)} $ b_i $'s current payment $ p_{i,j}^{U\rightarrow E}\left\langle m \right\rangle $ (e.g., in the $m^\text{th}$ round , Algorithm 1) equals to its expected valuation $ \text{E}[v_{i,j} ]$; \textit{iii)} the risk is out of control, e.g., raising the payment can bring an unacceptable risk on utility (e.g., line 18, Algorithm 1). Such a consideration ensures that constraint (28c) can be satisfied. Moreover, thanks to the well-designed risk analysis, e.g., constraint (28d), each MU $ b_i $ can also decide whether to sign a forward contract with the matched ES, upon evaluating the acceptable risk of being selected as a volunteer (e.g., lines 44-45, Algorithm 1), which thereby guarantees that each contractual MU can be served by the corresponding ES in each practical transaction, at a high probability. 

\textbf{Individual rationality of ESs.} Owing to overbooking, each ES $s_j^E$ regards $ (1+\tau)K_j^E $ as the maximum capacity for serving MUs (e.g., the maximum number of MUs that can access to the ES simultaneously, line 12, Algorithm 1), and the actual number of matched MUs of $s_j^E$ will definitely not exceed its overbooked resource supply $  (1+\tau)\left(G_j^E+\sum_{s_k^C\in\varphi\left(s_j^E\right)} |\mathbb{B}_{j,k}|\right) $ (e.g., line 43, Algorithm 1). In addition, the payment offered from MU $ b_i $ can cover the service cost of $s^E_j$, while staying below its expense paid to a CS (i.e., lines 16 and 36, Algorithm 1), which thus supporting that its non-negative expected utility.

As a summary, MUs and ESs are individual rationality in our proposed MU-ES matching of OA-CLM.
\end{proof}
\begin{lem}
No blocking pair can exist in the MU-ES matching of OA-CLM
\end{lem}
\begin{proof}
We offer the proof to show there is no blocking pair of either Type 1 or Type 2, as following:

\noindent 
$\bullet$ \textbf{There is no Type 1 blocking pair in the MU-ES matching of OA-CLM.} We offer the proof by considering contradiction.

 Under a given matching $ \omega $, MU $ b_i $ and ES $ s_j^E $ form a Type 1 blocking pair $ \left(b_i;s_j^E\right) $.
If MU $ b_i $ does not sign a forward contract with ES $ s_j^E $, when any of the following conditions is met: \textit{i)} the final payment offered by MU $ b_i $ (e.g., payment $ p_{i,j}^{U\rightarrow E}$ during the last round) equals to its expected valuation; and \textit{ii)} the risk is out of control (e.g., constraint (28c)). For analytical simplicity, we use $ p_{i,j}^{max} $ to denoted the maximum payment from $b_i$ to $s_j^E$ under an accepted risk $ R_1^U $. Thus, the final payment $ p_{i,j}^{U\rightarrow E}$ can only refer to $\text{E}[v_{i,j}]$ or $p_{i,j}^{max} $, shown by (59) and (60).
\begin{equation}\label{key}{\small
	\begin{aligned}
		p_{i,j}^{U\rightarrow E} =  \text{min}\left\{\text{E}[v_{i,j}],p_{i,j}^{max}\right\},
	\end{aligned}}
\end{equation}
\begin{equation}\label{key}{\small
		\begin{aligned}
			\overline{u^{U\leftrightarrow E}}\left(s_j^E,\omega\left(s_j^E\right)\backslash \widetilde{\omega^\prime}\left(s_j^E\right)\cup \left\{b_i\right\}\right)<\overline{u^{U\leftrightarrow E}}\left(s_j^E,\omega\left(s_j^E\right)\right).	
	\end{aligned}}
\end{equation}

If ES $ s_j^E $ selects MU $ b_i $, we have $ p_{i,j}^{U\rightarrow E}\left\langle m^{*} \right\rangle\leq p_{i,j}^{U\rightarrow E}\left\langle m \right\rangle =  \text{min}\left\{\text{E}[v_{i,j}],p_{i,j}^{max}\right\} $ and the following (61)
\begin{equation}\label{key}{\small
	\begin{aligned}
		&\overline{u^{U\leftrightarrow E}}\left(s_j^E,\omega\left(s_j^E\right)\backslash \widetilde{\omega^\prime}\left(s_j^E\right)\cup \left\{b_i\right\}\right) \geq\\& \overline{u^{U\leftrightarrow E}}\left(s_j^E,\omega\left(s_j^E\right)\backslash \widetilde{\omega^{\prime\prime}}\left(s_j^E\right)\cup \left\{b_i\right\}\right),\\
	\end{aligned}}
\end{equation}
where $ 
\widetilde{\omega^{\prime\prime}}\left(s_j^E\right) \subseteq \widetilde{\omega^{\prime}}\left(s_j^E\right) $. From (60) and (61), we can get
\begin{equation}\label{key}\small{
	\begin{aligned}
		\overline{u^{U\leftrightarrow E}}\left(s_j^E,\omega\left(s_j^E\right)\right)> \overline{u^{U\leftrightarrow E}}\left(s_j^E,\omega\left(s_j^E\right)\backslash \omega^{\prime\prime}\left(s_j^E\right)\cup \left\{b_i\right\}\right),
	\end{aligned}}
\end{equation}
which is contrary to (22), which thus ensures the inexistence of Type 1 blocking pairs.

\noindent 
$\bullet$ \textbf{There is no Type 2 blocking pair in the MU-ES matching of OA-CLM.}
We conduct the proof by considering cases of contradiction\footnote{In this scenario, we conduct the proof by considering cases of contradiction, that is, the ES still possesses sufficient resources (i.e., VMs). This assumption is constrained by the number of orthogonal subcarriers, thereby excluding the presumption of proof under infinite bandwidth. Moreover, we do not consider transmission outages in our analysis, which can be supported by several recent studies \cite{8815852,9154594,10321730}.}. 

Under a given matching $ \omega $, MU $ b_i $ and ES $ s_j^E $ form a Type 2 blocking pair $ \left(b_i;s_j^E\right) $, as shown by (23).
If MU $ b_i $ is rejected by ES $ s_j^E $, the final payment of $ b_i $ can be set by $ p_{i,j}^{U\rightarrow E} = \text{min}\left\{\text{E}[v_{i,j}],p_{i,j}^{max}\right\} $, where the only reason of such a rejection is that $ s_j^E $ has no surplus resources. However, the coexistence of (23) shows that ES $ s_j^E $ has adequate resource supply to serve MUs, which contradicts our previous assumption. Therefore, we prove that there is no Type 2 blocking pair.

As a summary, no blocking pair can exist in our proposed MU-ES matching in OA-CLM. 
\end{proof}

\begin{lem}
(Individual rationality of ES-CS matching of OA-CLM) All the ESs and CSs are individual rational in the futures market.
\end{lem}
\begin{proof}
The individual rationality of each ES and CS are proved respectively, as following: 

\noindent 
$\bullet$ \textbf{Individual rationality of ESs.} Note that we conduct risk analysis $ R_1^E $ to support that each ES only pre-signs forward contracts with an acceptable risk, otherwise, it will participate in spot market (e.g., lines 44-45, Algorithm 1). In addition, the payment offered by MU $ b_i $ can cover the service cost of $s^E_j$, while does not exceed its payment to a CS in serving this task (i.e., lines 16 and 36, Algorithm 1), which thus guarantees a non-negative expected utility.

\noindent 
$\bullet$ \textbf{Individual rationality of CSs.} Since we have $ p_{i,j,k}^{E\rightarrow C}\le c^C_{i,j,k} $ (e.g., constraint (30b)), making that the final payment from ES will definitely stay above the corresponding service cost of a CS, ensuring a non-negative expected utility. Moreover, our designed risk estimation and control, e.g., constraint (30c), encourages each CS $s_k^C$ to make sure the risk of resource requests exceeding its supply to be within a reasonable range (line 32, Algorithm 1).

As a summary, ESs and CSs are individual rational in our proposed ES-CS matching of OA-CLM.
\end{proof}
\begin{lem}
No blocking coalition can exist in the ES-CS matching of OA-CLM.
\end{lem}
\begin{proof}
We offer the proof to show there is no blocking coalition of either Type 1 or Type 2, as following: 

\noindent 
$\bullet$ \textbf{There is no Type 1 blocking coalition in the ES-CS matching of OA-CLM.} We offer the proof by analyzing the cases of contradiction. 

Under a given matching $ \varphi $, CS $ s_k^C $ and an ES set $ \mathbb{S}$ ($ \mathbb{S}\subseteq\bm{S^{C}}$) form a Type 1 blocking coalition $ \left(s_k^C;\mathbb{S}\right) $, as shown by (24) and (25).
If ES $ s_j^E $ does not sign a forward contract with $ s_k^C $, the payment offered by $ s_j^E $ per task during the last round can only be set by the maximum payment $ p^{max}_j $, as given by the following (63) and (64).
\begin{equation}\label{key}{\small
		\begin{aligned}
		p_{i,j,k}^{E\rightarrow C}\left\langle n \right\rangle = p^{max}_j,	
	\end{aligned}}
\end{equation}
\begin{equation}\label{key}{\small
	\begin{aligned}
	&\overline{u^{E\leftrightarrow C}}\left(s_j^E,\varphi\left(s_j^E\right)\backslash\widetilde{\varphi^\prime}\left(s_j^E\right)\cup\left\{s_k^C\right\}\right)<\\&\overline{u^{E\leftrightarrow C}}\left(s_j^E,\varphi\left(s_j^E\right)\right).
\end{aligned}}
\end{equation}

If $ s_j^E $ selects $ s_k^C $ as a possible resource provider, we have $ p_{j,k}^{E\rightarrow C}\left\langle n^{*} \right\rangle \leq p_{j,k}^{E\rightarrow C}\left\langle n \right\rangle = p^{max}_j $ and the following (65)
\begin{equation}\label{key}{\small
	\begin{aligned}
		&\overline{u^{E\leftrightarrow C}}\left(s_j^E,\varphi\left(s_j^E\right)\backslash\widetilde{\varphi^\prime}\left(s_j^E\right)\cup\left\{s_k^C\right\}\right) \geq\\& \overline{u^{E\leftrightarrow C}}\left(s_j^E,\varphi\left(s_j^E\right)\backslash\widetilde{\varphi^{\prime\prime}}\left(s_j^E\right)\cup\left\{s_k^C\right\}\right),\\
	\end{aligned}}
\end{equation}
where $ 
\widetilde{\varphi^{\prime\prime}}\left(s_j^E\right) \subseteq \widetilde{\varphi^{\prime}}\left(s_j^E\right) $. From (64) and (65), we can get
\begin{equation}\label{key}{\small
	\begin{aligned}
		&\overline{u^{E\leftrightarrow C}}\left(s_j^E,\varphi\left(s_j^E\right)\right)>\\& \overline{u^{E\leftrightarrow C}}\left(s_j^E,\varphi\left(s_j^E\right)\backslash\widetilde{\varphi^{\prime\prime}}\left(s_j^E\right)\cup\left\{s_k^C\right\}\right),
	\end{aligned}}
\end{equation}
which is contrary to (25), and can thereby prove the inexistence of Type 1 blocking coalition.

\noindent 
$\bullet$ \textbf{There is no Type 2 blocking coalition in the ES-CS matching of OA-CLM.}
Similarly, we also conduct the proof upon having contradiction.

 Under a given matching $ \varphi $, CS $ s_k^C $ and an ES set $ \mathbb{S}$ ($ \mathbb{S}\subseteq\bm{S^{E}}$) form a Type 2 blocking coalition $ \left(s_k^C;\mathbb{S}\right) $, as given by (26) and (27).
If $ s_k^C $ rejects ES $ s_j^E $, the expense offered by $ s_j^E $ during the last round should be $ p_{i,j,k}^{E\rightarrow C} = p_j^{max} $, where the only reason of the rejection is that $s_k^C$ has no surplus resources. However, the coexistence of (26) and (27) shows that CS $ s_k^C $ has adequate VMs to serve ESs, which contradicts the aforementioned assumption. Therefore, our proposed matching does not allow any Type 2 blocking coalition.

All in all, there is no blocking coalitions in our proposed ES-CS matching of OA-CLM.
\end{proof}
\begin{thm}
(Strong stability of OA-CLM) OA-CLM is strongly stable.
\end{thm}
\begin{proof}
	Since the matching result of Algorithm 1 holds Lemma 3, Lemma 4, Lemma 5, and Lemma 6, according to Definition
	9, our proposed OA-CLM in the futures market is strongly stable.
\end{proof}

\begin{thm}
	(Competitive equilibrium associated with resource trading between MUs and ESs in OA-CLM) The trading between MUs and ESs can reach a competitive equilibrium.
\end{thm}
\begin{proof}
	To prove this theorem, we discuss that the three conditions introduced by Definition 10 (given by Sec. 3.5) can be held in MU-ES trading. First, we set $ p_{i,j}^{U\rightarrow E} \geq c^E_{i,j} $, indicating that the service cost will be covered by the payment from MUs in each round (e.g., constraint (29c)). We demonstrate next that when MU $ b_i $ enters into a forward contract with an ES $ s_j^E $, $ b_i $ achieves maximum expected utility. This is attributed to the fact that $ b_i $ selects the ES based on its preference list $ L_i^U $ (e.g., line 5, Algorithm 1), ensuring the attainment of the maximum expected utility for $ b_i $. Then, if $ b_i $ is not matched to any ES $ s_j^E \in \bm{C_i} $, its payment can be equal to its expected valuation $ \text{E}[v_{i,j}] $ or the maximum payment $ p_{i,j}^{max} $ that $b_i$ can tolerate, under an accepted risk $ R_1^U $ (lines 15-16, Algorithm 1). 
	According to Definition 10, we can verify that the considered MU-ES trading in futures market can reach a competitive equilibrium.
\end{proof}
\begin{thm}
	(Competitive equilibrium associated with resource trading between ESs and CSs in OA-CLM) The trading between ESs and CSs can reach a competitive equilibrium.
\end{thm}
\begin{proof}
		To prove this theorem, we also prove that the three conditions in Definition 11 (given in Sec. 3.5) can be held in ES-CS trading. First, we set $ p_{i,j,k}^{E\rightarrow C} \geq c^C_{i,j,k} $ to make sure that the service cost of the CS will be covered by its income paid from ESs in each round (i.e., constrain (29c)). We next demonstrate that when ES $ s_j^E $ engages in a forward contract for a task with a CS $ s_k^C $, $ s_j^E $ attains maximum expected utility. This is attributed to the fact that $ s_j^E $ selects the CS based on its task-specific preference list $ L_{i,j}^E $ (e.g., line 24, Algorithm 1), ensuring the achievement of the maximum expected utility for $ s_j^E $. Then, if a task of $ s_j^E $ has not been matched to any CS $ s_k^C $, the payment from $ s_j^E $ should be equal to its maximum payment $ p_j^{max} $ (lines 35-36, Algorithm 1). 
Based on Definition 11, we thereby verify that the considered ES-CS trading in futures market can achieve competitive equilibrium.
\end{proof}

\begin{thm}
(Weak Pareto optimality of OA-CLM) The proposed OA-CLM provides a weak Pareto optimality.
\end{thm}
\begin{proof}
	Review the design of OA-CLM, each participant (e.g., MU, ES, CS) makes decisions according to its preference list. If the subsequent choice ranks higher in the participant's preference list, they will switch their matching target in the following round. Such a switch operation indicates that backing to the previous choice will not bring it with any larger expected utility. For a MU $ b_i $, if there exists an ES $ s_j^E $ that can offer a higher expected utility than its current matched ES, $ b_i $ and $ s_j^E $ are more inclined to establish a matching relationship, this, however, will form a blocking pair. Since Theorem 1 verifies that our proposed OA-CLM is stable while allowing no blocking pairs. There exists no Pareto improvement, when the procedure of MU-ES matching terminates. Similarly, we can infer that there is no Pareto improvement in ES-CS matching (e.g., Lemma 6 and Theorem 1). As a summary, our studied OA-CLM game is weak Pareto optimal.
\end{proof}

\section{Details of OS-CLM Game}
\subsection{Problem Formulation}
In the spot market, we are interested in the M2O matching $ \omega^{\prime}(.) $ between an ES and several MUs, and the M2M matching $ \varphi^{\prime}(.) $ between ESs and CSs. Note that each MU, ES, and CS aims to \textit{maximize its overall practical utility}, which can be mathematically formulated as the following optimization problems.

\noindent
$\bullet$ \textbf{Optimization of MUs' practical utility:} 
\begin{subequations}\label{key}{\small
	\begin{align}
		\bm{\mathcal{F}^{U{\prime}}}:~~&\underset{{\omega^{\prime}\left(b_i\right)}}{\max}~u^{U\prime}\left(b_i,\omega^{\prime}\left(b_i\right)\right)\tag{68}\\
		\text{s.t.}~~~~&
		\omega^{\prime}\left(b_i\right)=\left\{\left\{s_j^E\right\},\varnothing \right\},~\forall s_j^E\subseteq \bm{C_i^{\prime}}\tag{68a}\\
		&p_{i,j}^{U\prime\rightarrow E\prime}\le v_{i,j},~\text{if}~\omega^{\prime}\left(b_i\right)=\left\{s_j^E\right\}\tag{68b}
	\end{align}}
\end{subequations}

\noindent 
where constraint (68a) guarantees that $b_i$ can only be mapped to one ES, and constraint (68b) limits the payment of a MU to an ES within a certain range, to hold the individual rationality.

\noindent
$\bullet$ \textbf{Optimization of ESs' practical utility:}
\begin{subequations}\label{key}{\small
	\begin{align}
		\bm{\mathcal{F}^{E{\prime}}}:~~&\underset{{\omega^{\prime}\left(s_j^E\right),\varphi^{\prime}\left(s_j^E\right)}}{\max} u^{E\prime}\left(s_j^E,\omega^{\prime}\left(s_j^E\right),\varphi^{\prime}\left(s_j^E\right)\right)\tag{69}\\
		\text{s.t.}~~~~&
		\omega^{\prime}\left(s_j^E\right)\subseteq \bm{B^{\prime}},~
		\varphi^{\prime}\left(s_j^E\right)\subseteq \bm{S^{C{\prime}}}\tag{69a}\\
		&p_{i,j}^{U\prime\rightarrow E\prime}\geq c_{i,j}^E, ~\forall b_i \in \omega^{\prime}\left(s_j^E\right)\tag{69b}\\
		&p_{i,j,k}^{E\prime\rightarrow C\prime}\le p_{i,j}^{U\prime\rightarrow E\prime},~\forall s_k^C \in \varphi^{\prime}\left(s_j^E\right)\tag{69c}\\
		&|\omega^{\prime}\left(s_j^E\right)|+\sum_{b_i\in\omega\left(s_j^E\right)}\alpha_i\le\sum_{s_j^E\in\varphi\left(s_k^C\right)}\sum_{b_i\in \mathbb{B}_{j,k}}\beta_{i,j,k}\nonumber\\&+G_j^E+\sum_{s_j^E\in\varphi^\prime\left(s_k^C\right)}{|\mathbb{B}^\prime_{j,k}|}\le K_j^E\tag{69d}
	\end{align}}
\end{subequations}

\noindent 
where $ |\mathbb{B}^\prime_{j,k}| $ denotes the number of tasks offloaded from $ s_j^E $ to $ s_k^C $. In problem $\bm{\mathcal{F}^{E{\prime}}}$, constraint (69a) ensures that $ \omega^\prime\left(s_j^E\right) $ and $ \varphi^\prime\left(s_j^E\right) $ belong to set $ \bm{B^\prime} $ and $ \bm{S^{C{\prime}}} $; constraints (69b) and (69c) guarantee that ES $ s_j^E $ can obtain a non-negative utility; while constraint (69d) describes the limited number of MUs served by ES $ s_j^E $.

\noindent
$\bullet$ \textbf{Optimization of CSs' practical utility:} 
\begin{subequations}\label{key}{\small
	\begin{align}
		\bm{\mathcal{F}^{C^{\prime}}}:~~&\underset{{\varphi^\prime\left(s_k^C\right)}}{\max}~ u^{C\prime}\left(s_k^C,\varphi^{\prime}\left(s_k^C\right)\right)\tag{70}\\
		\text{s.t.}~~~~&
		\varphi^{\prime}\left(s_k^C\right){\subseteq\bm{{S}^{{E}\prime}}}\tag{70a}\\
		&p_{i,j,k}^{E\prime\rightarrow C\prime}\geq c_{i,j,k}^C,~\forall s_j^E \in \varphi^\prime\left(s_k^C\right)\tag{70b}\\
		&\sum_{s_j^E\in\varphi\left(s_k^C\right)}\sum_{b_i\in \mathbb{B}_{j,k}}{\beta_{i,j,k}}+\sum_{s_j^E\in\varphi^\prime\left(s_k^C\right)}{|\mathbb{B}^\prime_{j,k}|}+\varepsilon_k\le G_k^C\tag{70c}
	\end{align}}
\end{subequations}

\noindent 
where constraint (70a) shows that $ \omega^\prime\left(s_k^C\right) $ belongs to set $ \bm{{S}^{E\prime}} $, constraint (70b) ensures a non-negative utility for each CS, and constraint (70c) limits the number of requestors served by CS $ s_k^C $.

\subsection{Key Definitions}
Service provisioning in our designed spot market over CAMENs is modeled by an OS-CLM game. Similar to OA-CLM, it comprises of two matching types: \textit{i)} MU-ES matching, and \textit{ii)} ES-CS matching. We start with MU-ES matching, which is in the form of a M2O matching tailored to the characteristics of our designed spot market, based on a series of definitions given below. 
\begin{Defn}(Many-to-one matching in the spot market)
	A M2O matching $ \omega^\prime $ in the spot market constitutes a mapping between MU set $ \bm{B^\prime} $ and ES set $\bm{S^{E{\prime}}} $, while satisfying the following conditions:
	
	\noindent
	$\bullet$ for each MU $ b_i\in\bm{B^\prime} $, $ \omega^\prime\left(b_i\right){\subseteq\bm{C^\prime_i}} $, $\bm{C^\prime_i}\subseteq\bm{S^{E{\prime}}} $, $ \left|\omega^\prime\left(b_i\right)\right|=1 $;
	
	\noindent
	$\bullet$ for each ES $ s_j^E\in\bm{C_i^\prime} $, $ \omega^\prime\left(s_j^E\right)\subseteq\bm{B^\prime} $;
	
	\noindent
	$\bullet$ for MU $ b_i $ and ES $ s_j^E $, $ b_i\in\omega^\prime\left(s_j^E\right) $ if and only if $ s_j^E\in\omega^\prime\left(b_i\right) $.
\end{Defn}

As the middle layer of CAMENs, ESs need to trading with MUs and CSs, to the convenience of expression, we divide (38) into two parts: \textit{i)} $ u^{U\prime\leftrightarrow  E\prime}\left(s_j^E,\omega^\prime\left(s_j^E\right)\right) = \sum_{b_i\in\omega^{\prime}\left(s_j^E\right)}\left(p_{i,j}^{E\prime \rightarrow C\prime}-c_{i,j}^E\right) $;
 and \textit{ii)} $ u^{U\prime\leftrightarrow E\prime}\left(s_j^E,\varphi^\prime\left(s_j^E\right)\right)=-\sum_{s_k^C\in\varphi^{\prime}\left(s_j^E\right)}\sum_{b_i\in\mathbb{B}^\prime_{j,k}}\left(p_{i,j,k}^{E\prime\rightarrow C\prime}-c_{i,j}^E\right) $. We next define the concept of \textit{blocking pair}, representing a significant factor that may lead to instability of a M2O matching. 

\begin{Defn}(Blocking pairs of MU-ES matching in OS-CLM)
	Under a given M2O matching $ \omega^\prime $, MU $ b_i $ and ES $ s_j^E $ form a blocking pair $ \left(b_i;s_j^E\right) $, for which we consider two types:

\noindent 
	\textbf{Type 1 blocking pair}: Type 1 blocking pair satisfies the following condition:
	\begin{equation}\label{key}{\small
		\begin{aligned}
		&{u^{U\prime\leftrightarrow E\prime}}\left(s_j^E,\omega^\prime\left(s_j^E\right)\backslash \widetilde{\omega^\prime}\left(s_j^E\right)\cup \left\{b_i\right\}\right)
		\\&>{u^{U\prime\leftrightarrow E\prime}}\left(s_j^E, \omega^\prime\left(s_j^E\right)\right),
	\end{aligned}}
	\end{equation}
	which indicates that $ s_j^E $ can increase its utility by giving up some MUs, e.g., $ \widetilde{\omega^\prime}\left(s_j^E\right) $, while serving $ b_i $ instead.
	
	\noindent 
	\textbf{Type 2 blocking pair}: Type 2 blocking pair satisfies the following condition:
	\begin{equation}\label{key}{\small
			\begin{aligned}
				{u^{U\prime\leftrightarrow E\prime}}
			\left(s_j^E,\omega^\prime\left(s_j^E\right)\cup\left\{b_i\right\}\right)>{u^{U\prime\leftrightarrow E\prime}}\left(s_j^E,\omega^\prime\left(s_j^E\right)\right),	
		\end{aligned}}
	\end{equation}
	which makes a matching unstable since $ s_j^E $ can serve more MUs under its resource constraint, to improve its utility.
\end{Defn}

The ES-CS matching in the spot market can be formalized relying on the following definitions.
\begin{Defn}(Many-to-many matching in the spot market)
	A M2M matching $ \varphi^\prime $ in the spot market denotes a mapping between $ \bm{S^{E{\prime}}} $ and $ \bm{S^{C{\prime}}} $, while satisfying the following conditions:
	
	\noindent
	$\bullet$ for each CS $ s_k^C\in\bm{S^{C{\prime}}} $, $ \varphi^\prime\left(s_k^C\right){\subseteq\bm{S^{E{\prime}}}} $;
	
	\noindent
	$\bullet$ for each ES $ s_j^E\in\bm{S^{E{\prime}}} $, $ \varphi^\prime\left(s_j^E\right)\subseteq\bm{S^{C{\prime}}} $;
	
	\noindent
	$\bullet$ for CS $ s_k^C $ and ES $ s_j^E $, $ s_k^C\in\varphi^\prime\left(s_j^E\right) $ if and only if $ s_j^E\in\varphi^\prime\left(s_k^C\right) $.
\end{Defn}

We next introduce the concept of \textit{blocking coalition}, which is a significant factor that can make a M2M matching unstable.

\begin{Defn}(Blocking coalition of ES-CS matching in OS-CLM)
	Given a M2M matching $ \varphi^\prime $, CS $ s_k^C $ and ES set $ \mathbb{S^\prime}\subseteq\bm{S^{E{\prime}}}$ form a blocking pair $ \left(s_k^C;\mathbb{S^\prime}\right) $, for which we consider two types:
	
	\noindent	
	\textbf{Type 1 blocking coalition}: Type 1 blocking coalition can be incurred when the following conditions are met:
	\begin{equation}\label{key}{\small
			\begin{aligned}
			{u^{C\prime}}\left(s_k^C,\mathbb{S^\prime}\right)>{u^{C\prime}}\left(s_k^C,\varphi^\prime\left(s_k^C\right)\right)	
		\end{aligned}}
	\end{equation}
	\begin{equation}\label{key}{\small
		\begin{aligned}
		&u^{E\prime\leftrightarrow C\prime}\left(s_j^E,\varphi^\prime\left(s_j^E\right)\backslash\widetilde{\varphi^\prime}\left(s_j^E\right)\cup\left\{s_k^C\right\}\right)>\\&{u^{E\prime\leftrightarrow C\prime}}\left(s_j^E,\varphi^\prime\left(s_j^E\right)\right).	
		\end{aligned}}
	\end{equation}
	
	\noindent
	\textbf{Type 2 blocking coalition}: Type 2 blocking coalition can be incurred when the following conditions are met:
	\begin{equation}\label{key}{\small
			\begin{aligned}
			{u^{C\prime}}\left(s_k^C,\mathbb{S^\prime}\right)>{u^{C\prime}}\left(s_k^C,\varphi^\prime\left(s_k^C\right)\right).	
		\end{aligned}}
	\end{equation}
	\begin{equation}\label{key}{\small
			\begin{aligned}
					{u^{E\prime\leftrightarrow C\prime}}\left(s_j^E,\varphi^\prime\left(s_j^E\right)\cup\left\{s_k^C\right\}\right)>{u^{E\prime\leftrightarrow C\prime}}\left(s_j^E,\varphi^\prime\left(s_j^E\right)\right).
		\end{aligned}}
	\end{equation}
\end{Defn}
It can be construed that Type 1 blocking coalition leads to the unstability of a matching, since the ES is incentivized to choose a another set of CSs to achieve a higher utility. Similarly, Type 2 blocking coalition can also make a matching unstable, due to the ES can trade with more CSs to increase its utility.
\subsection{Algorithm Design}
\begin{algorithm*}[!t]
	{\small
	\setstretch{0.1} 
	\caption{Proposed OS-CLM in the spot market}
	\LinesNumbered 
	\textbf{Initialization: : $ m\leftarrow1 $, $ n\leftarrow1 $, $ p_{i,j}^{U\prime\rightarrow E\prime}\left\langle 1 \right\rangle\leftarrow p_{i,j}^{min}$, $ p_{j,k}^{E\prime\rightarrow C\prime}\left\langle 1 \right\rangle\leftarrow p_{i,j,k}^{min} $, $ flag_i\leftarrow 1 $, $ flag_j\leftarrow 1 $, for $\forall b_i\in\bm{B^{\prime}}$, $\forall s_j^E\in\bm{{S}^{{E}\prime}}$, $\forall s_k^C\in\bm{{S}^{{C}\prime}}$ }\\
	\textit{\textbf{\% Phase 1. MU-ES M2O matching game}}
	
	\While{$ {flag}_{i} $}{
		\textbf{Calculate} $ L_i^{U\prime}$
		
		\textbf{Determine} $ \mathbb{W}^\prime(b_i)\leftarrow $ choose the top ES in $ L_i^{U\prime} $
		
		\textbf{$ {flag}_{i} \leftarrow 0 $}
		
		\If{$ \forall\mathbb{W}^\prime\left( b_{i} \right) \neq \varnothing $}{
			\For{$ 
				s^E_{j} \leftarrow \mathbb{W}^\prime\left( b_{i} \right) $}{$ b_{i} $ sends a proposal including $ p_{i,j}^{U\prime\rightarrow E\prime}\left\langle m \right\rangle $ and $ r^U_{i} $ to $ s^E_j $}
			\While{$ \Sigma_{b_{i}\in\bm{B^\prime}}{flag}_{i} > 0 $}{
				Collect proposals from the MUs in $ \bm{B^\prime} $, e.g., using $ \widetilde{\mathbb{W}^\prime}\left(s_j^E\right) $ to include the MUs that send proposals to $ s_j^E $
				
				$ \mathbb{W}^\prime\left(s_j^E\right)\leftarrow $ choose MUs from $ \widetilde{\mathbb{W}^\prime}\left(s_j^E\right) $ to maximize its utility under limited $ \left(K_j^E-\sum_{b_i\in\omega\left(s_j^E\right)}\alpha_i\right) $ constraint
				
				$ s_j^E $ temporally accept the MUs in $ \mathbb{W}^\prime\left(s_j^E\right) $, and rejects the others
			}
			\For{
				$ b_i \in \mathbb{W}^\prime\left(s_j^E\right) $
			}{
				\If{$ b_i $ is rejected by $ s_j^E $ and $ p_{i,j}^{U\prime\rightarrow E\prime}\left\langle m \right\rangle<v_{i,j} $}{
					$ p_{i,j}^{U\prime\rightarrow E\prime}\left\langle m+1 \right\rangle\leftarrow \text{min}\left\{p_{i,j}^{U\prime\rightarrow E\prime}\left\langle m \right\rangle+\mathrm{\Delta}p_i,v_{i,j}\right\} $}
				\Else{$ p_{i,j}^{U\prime\rightarrow E\prime}\left\langle m+1 \right\rangle\leftarrow p_{i,j}^{U\prime\rightarrow E\prime}\left\langle m \right\rangle $}
			}
			
			\If{there exists $
				p_{i,j}^{U\prime\rightarrow E\prime}\left\langle m+1 \right\rangle\neq p_{i,j}^{U\prime\rightarrow E\prime}\left\langle m \right\rangle $, $\forall b_{i} \in \mathbb{W}^\prime\left(s_j^E\right) $}{
				$ {flag}_i\leftarrow 1 $,
				$ m\leftarrow m+1 $
			}
		}
	}
	\textit{\textbf{\% Phase 2. ES-CS M2M matching game}}
	
	\While{$ flag_j $}{
		\textbf{Calculate} $ L_{i,j}^{E\prime} $
		
		\textbf{Determine} $ \mathbb{Y}^\prime\left(s_j^E\right)\leftarrow $ choose $ s_j^E $'s interested CSs from $ L_{i,j}^{E\prime} $
		
		\textbf{$ {flag}_{j} \leftarrow 0 $}\
		
		\If{$ \forall \mathbb{Y}^\prime \left(s_k^C \right)\neq \varnothing $}{
			\For{$ s_k^C\in \mathbb{Y}^\prime\left(s_j^E\right) $}{$ s_j^E $ sends the payments $ p_{i,j,k}^{E\prime\rightarrow C\prime} $ and corresponding requested amount of resource $ r_{i}^U $ to $ s_k^C $}

			\While{$ \Sigma_{s^E_{j}\in\bm{S^{E{\prime}}}}{flag}_{j} > 0 $}{
				Collect proposals from the ESs in $ \bm{S^{E{\prime}}} $, e.g., using $ \widetilde{\mathbb{Y}^\prime} \left(s_k^C\right)$ to include ESs' tasks that send proposals to $ s_k^C $;
				
				$ \mathbb{Y}^\prime \left(s_k^C\right)\leftarrow $ choose CSs from $ \widetilde{\mathbb{Y}^\prime}  \left(s_k^C\right) $ to maximize the CS's utility under limited $\left(G_k^C-\sum_{s_j^E\in\varphi\left(s_k^C\right)}\sum_{b_i\in \mathbb{B}_{j,k}}{\beta_{i,j,k}}-\sum_{s_j^E\in\varphi^\prime\left(s_k^C\right)}{|\mathbb{B}^\prime_{j,k}|}-\varepsilon_k\right) $ VMs constraint.
				
				$ s_k^C $ temporally accept the ESs in $ \mathbb{Y}^\prime \left(s_k^C\right) $, and rejects the others
			}
			\For{
				$ s_j^E \in \mathbb{Y}^\prime \left(s_k^C\right) $
			}{
				\If{a task from $ s_j^E $ is rejected by $ s_k^C $ and $ p_{i,j,k}^{E\prime\rightarrow C\prime}\left\langle n \right\rangle<p_{i,j}^{U\prime\rightarrow E\prime}\left\langle m \right\rangle $}{
					$ p_{i,j,k}^{E\prime\rightarrow C\prime}\left\langle n+1 \right\rangle\leftarrow \text{min}\left\{p_{i,j,k}^{E\prime\rightarrow C\prime}\left\langle n \right\rangle+\mathrm{\Delta}p_j,p_{i,j}^{U\prime\rightarrow E\prime}\left\langle m \right\rangle \right\} $}
				\Else{$ p_{i,j,k}^{E\prime\rightarrow C\prime}\left\langle n+1 \right\rangle\leftarrow p_{i,j,k}^{E\prime\rightarrow C\prime}\left\langle n \right\rangle $}
			}
			
			\If{there exists $
				p_{i,j,k}^{E\prime\rightarrow C\prime}\left\langle n+1 \right\rangle\neq p_{i,j,k}^{E\prime\rightarrow C\prime}\left\langle n \right\rangle $, $\forall s_j^E \in \mathbb{Y}^\prime \left(s_k^C\right) $}{
				$ {flag}_j\leftarrow 1 $,
				$ n\leftarrow n+1 $
			}
		}
	}
	\% \textit{\textbf{Phase 3. Cross-layer interaction}}
	
	\If{$ |\omega^{\prime}\left(s_j^E\right)|+\sum_{b_i\in\omega\left(s_j^E\right)}\alpha_i>\sum_{s_j^E\in\varphi\left(s_k^C\right)}\sum_{b_i\in \mathbb{B}_{j,k}}\beta_{i,j,k}+G_j^E+\sum_{s_j^E\in\varphi^\prime\left(s_k^C\right)}{|\mathbb{B}^\prime_{j,k}|}$}{$ \mathbb{W}^\prime \left(s_j^E\right)\leftarrow $ $ s_j^E $ selects some MUs from $ \mathbb{W}^\prime\left(s_j^E\right) $ to maximize its utility based on ES's resource supply, i.e., $\sum_{s_j^E\in\varphi\left(s_k^C\right)}\sum_{b_i\in \mathbb{B}_{j,k}}\beta_{i,j,k}+G_j^E+\sum_{s_j^E\in\varphi^\prime\left(s_k^C\right)}{|\mathbb{B}^\prime_{j,k}|}-\sum_{b_i\in\omega\left(s_j^E\right)}\alpha_i $} 
$\omega^\prime\left(s_j^E\right)\leftarrow\mathbb{W}^\prime\left(s_j^E\right)$, $\omega^\prime(b_i)\leftarrow\mathbb{W}^\prime(b_i)$, $\varphi^\prime\left(s_j^E\right)\leftarrow\mathbb{Y}^\prime \left(s_j^E\right)$, $\varphi^\prime\left(s_k^C\right)\leftarrow\mathbb{Y}^\prime \left(s_k^C\right)$\\
\textbf{Return:} $\omega^\prime\left(s_j^E\right)$, $\omega^\prime(b_i)$, $\varphi^\prime\left(s_j^E\right)$, $\varphi^\prime\left(s_k^C\right)$
}
\end{algorithm*}
Our proposed OS-CLM facilitates two-way negotiations involving MUs, ESs, and CSs, where the trading between MUs and ESs, as well as between ESs and CSs, can significantly influence each other, resulting in an iterative and intricate trading process. To provide a clearer illustration of how the OS-CLM game operates, the entire matching procedure involves three key phases with multiple rounds. It specific details can be found in Algorithm 2. 

\textit{i) MU-ES matching game (Phase 1):} We first introduce the M2O matching between MUs and ESs and adopt the Gale-Shapley algorithm\cite{9667258,9687261} to construct a stable mapping in the considered spot market.

\noindent
\textbf{Step 1. Initialization:} At the beginning of each round, each MU $ b_i $ sets its payment as $ p_{i,j}^{U\prime\rightarrow E\prime}\left\langle 1 \right\rangle=p_{i,j}^{min}$ (line 1), where $ p_{i,j}^{min} $ denotes the initial payment from $ b_i $ to $ s_j^E $. In addition, each MU announces its requests to ESs according to its preference list (Definition 17).

\begin{Defn}(Preference list of MU) The preference list $ L_{i}^{U\prime} $ of a MU $ b_i $ regarding ESs represents a vector of $ s^E_j \in \bm{C_i^\prime} $, sorted by $ {u^{U\prime}}\left(b_i,s_j^E\right) $ under a non-ascending order:
	\begin{equation}\label{key}{\small
			\begin{aligned}
					L_{i}^{U\prime}= \left\{s_j^E~|~\text{{\rm non-ascending on }}  {u^{U\prime}}\left(b_i,s_j^E\right),\forall s_j^E\in \bm{C_i^\prime}\right\},	
		\end{aligned}}
	\end{equation}
\end{Defn}
\noindent
where we use $\mathbb{W}^\prime(b_i) \in L_{i}^{U\prime}$ to represent the favorite ES of $ b_i $ (e.g., $|\mathbb{W}^\prime(b_i) |=1$), and $\mathbb{W}^\prime\left(s_j^E\right)$ to indicate the set of MUs that are temporarily accepted by $ s_j^E $ during the matching procedure.

\noindent
\textbf{Step 2. Proposal of MUs:} At round $m$, each MU $ b_i $ chooses the top ES in $ L_{i}^{U\prime} $, and records it in $\mathbb{W}^\prime(b_i)$ (line 5). Then, $ b_i $ sends a proposal to ES in $\mathbb{W}^\prime(b_i)$ (denoted by $s^E_j$ for analytical simplicity), including its payments $ p_{i,j}^{U\prime\rightarrow E\prime}\left\langle m \right\rangle $ and the required amount of resource $ r_i^U $ (line 9).

\noindent
\textbf{Step 3. MU selection on ESs' side:} We use set $\widetilde{\mathbb{W}^\prime}\left(s_j^E\right)$ to collect the information from MUs, each ES $ s_j^E $ then determines a collection of temporarily MUs, as recorded
in set $\mathbb{W}^\prime\left(s_j^E\right)$, where $\mathbb{W}^\prime\left(s_j^E\right)\subseteq \widetilde{\mathbb{W}^\prime}\left(s_j^E\right)$ that enables the maximum utility under the maximum number of concurrently accessed users minus the number of contractual MUs (from OA-CLM) which have engaged in the current practical transaction (e.g., constraint (69d)). Then, each $ s_j^E $ informs MUs in set $\widetilde{\mathbb{W}^\prime}\left(s_j^E\right)$ about its determinations in the current round (lines 10-13).

\noindent
\textbf{Step 4. Decision-making on MUs' side:} After obtaining decisions from each ES $ s_j^E \in \mathbb{W}^\prime(b_i) $, $ b_i $ considers the following conditions:

\textbf{Condition 1.} The payment from $ b_i $ remains unchanged, when one of the following conditions is met (line 18): \textit{i)} $ b_i $ is accepted by an ES $ s_j^E $; \textit{ii)} $ b_i $'s current payment $ p_{i,j}^{U\prime\rightarrow E\prime}\left\langle m \right\rangle $ equals to its valuation $v_{i,j}$;

\textbf{Condition 2.} If $ b_i $ is rejected by an ES $ s_j^E $, its current payment $ p_{i,j}^{U\prime\rightarrow E\prime}\left\langle m \right\rangle $ stays below its valuation $ v_{i,j} $, $ b_i $ will put up its payment to $ s_j^E $ in the next round (line 16).

\noindent
\textbf{Step 5. Repeat:} If payments of all the MUs stay unchanged from the $ (m-1)^\text{th} $ round to the $ m^\text{th} $ round, the matching will be terminated at round $ m $ ($ \Sigma_{b_{i}\in\bm{B^\prime}}{flag}_{i} = 0 $, line 6). Otherwise, OS-CLM repeats the above steps (e.g., lines 3-20) in the next round.

\textit{ii) ES-CS matching game (Phase 2):} We next introduce the M2M matching between the ESs and CSs and adopt the Gale-Shapley algorithm\cite{9667258,9687261}, contributing to construct a stable mapping for purchasing cloud resources in our spot market, as follows:

\noindent
\textbf{Step 1. Initialization:} At the beginning of each round, each ES $ s_j^E $ sets its payment for a task as $ p_{i,j,k}^{E\prime \rightarrow C\prime}\left\langle 1 \right\rangle=p_{i,j,k}^{min}$, where $p_{i,j,k}^{min}$ denotes the initial payment from ES $ s_j^E $. We apply $\mathbb{Y}^\prime\left(s_j^E\right)$ to indicate the tasks that are temporarily accepted by CS $ s_j^E $ and $\mathbb{Y}^\prime\left(s_k^C\right)$ to represent CSs that are interested in the tasks of ES $s^E_j $, based on each task's (namely, MU's) preference list, according to the following definition.
\begin{Defn}(Preference list of ES) The preference list $ L_{i,j}^{E\prime} $ of an ES $s_j^E$ in handling the task of $b_i$ regarding CSs (e.g., which CS can the task of $b_i$ be offloaded to for further processing, from $s_j^E$) is a vector of $ s^C_k \in \bm{S^{C\prime}} $, sorted by $ {u^{E\prime\leftrightarrow C\prime}}\left (s_j^E,\varphi^\prime\left(s_j^E\right)\right ) $ and following a non-ascending order: 	
	\begin{equation}\label{key}{\small
		\begin{aligned}
			&L_{i,j}^{E\prime}=
			\\ & \left\{s_k^C~|~\text{{\rm non-ascending on }} {u^{E\prime\leftrightarrow C\prime}}\left (s_k^C,\varphi^\prime\left(s_j^E\right)\right ),\forall s^C_k \in \bm{S^{C\prime}}\right\}.
		\end{aligned}}
	\end{equation}
\end{Defn}

\noindent
\textbf{Step 2. Proposal of ESs:} At round $ n $, each ES that need to borrow cloud resource reports its information, including each task's payment $ p_{i,j,k}^{E\prime\rightarrow C\prime}\left\langle n \right\rangle $ and the corresponding required amount of resources $r_i^U $ to CSs in $\mathbb{Y}^\prime(s_j^E)$ (line 28). 

\noindent
\textbf{Step 3. ES selection on CSs' side:} After collecting the information from ESs in set $\widetilde{\mathbb{Y}^\prime}\left(s_k^C\right)$, each CS $ s_k^C $ determines a collection of temporarily ESs' tasks, as recorded by set $\mathbb{Y}^\prime\left(s_k^C\right)$, where $\mathbb{Y}^\prime\left(s_k^C\right)\subseteq \widetilde{\mathbb{Y}^\prime}\left(s_k^C\right)$, that enable the maximum utility under surplus resource supply (i.e., the resource supply of $s_k^C$ minus the practical demand of its contractual ESs (obtained from OA-CLM) and its inherent requesters during a practical transaction, denoted as $G_k^C-\sum_{s_j^E\in\varphi\left(s_k^C\right)}\sum_{b_i\in \mathbb{B}_{j,k}}{\beta_{i,j,k}}-\sum_{s_j^E\in\varphi^\prime\left(s_k^C\right)}{|\mathbb{B}^\prime_{j,k}|}-\varepsilon_k $) Then, each $ s_k^C $ informs ESs about its determinations in the current round (lines 29-32).

\noindent
\textbf{Step 4. Decision-making on ESs' side:} After obtaining decisions from CS $ s_k^C \in \mathbb{Y}^\prime\left(s_j^E\right) $, ES $ s_j^E $ considers the following conditions:

\textbf{Condition 1.} If the task from a MU associated with ES $ s_j^E $ is accepted by $ s_k^C $ or its current payment $ p_{i,j,k}^{E\prime\rightarrow C\prime}\left\langle n \right\rangle $ equals to its payment from $ b_i $ (i.e., $ p_{i,j}^{U\prime\rightarrow E\prime} $), the payment from $ s_j^E $ remains unchanged, i.e., $ p_{i,j,k}^{E\prime\rightarrow C\prime}\left\langle n+1 \right\rangle=p_{i,j,k}^{E\prime\rightarrow C\prime}\left\langle n \right\rangle $
(line 37);

\textbf{Condition 2.} If the task from a MU associated with ES $ s_j^E $ is rejected by $ s_k^C $ and its current payment $ p_{i,j,k}^{E\prime\rightarrow C\prime}\left\langle n \right\rangle $ is still below its payment from $b_i$, (i.e., $ p_{i,j}^{U\prime\rightarrow E\prime} $), $s_j^E$ increases its payment to $ s_k^C $ in the next round (line 35).

\noindent
\textbf{Step 6. Repeat:} If payments of all the ESs stay unchanged from the $ (n-1)^\text{th} $ round to the $ n^\text{th} $ round, the matching will be terminated at round $ n $ (e.g., $ \Sigma_{s^E_{j}\in\bm{S^{E{\prime}}}}{flag}_{j} = 0 $, line 25). Otherwise, OS-CLM repeats the above steps (e.g., lines 22-39) in the next round.

\textit{iii) Cross-layer interaction (Phase 3):} This step describes the mutual impacts between the above-discussed two phases. When ES $ s_j^E $ did not purchase sufficient resources from CSs in Phase 2, i.e., $ |\omega^{\prime}\left(s_j^E\right)|+\sum_{b_i\in\omega\left(s_j^E\right)}\alpha_i>\sum_{s_j^E\in\varphi\left(s_k^C\right)}\sum_{b_i\in \mathbb{B}_{j,k}}\beta_{i,j,k}+G_j^E+\sum_{s_j^E\in\varphi^\prime\left(s_k^C\right)}{|\mathbb{B}^\prime_{j,k}|}$, $ s_j^E $ will further choose some MUs from $ \mathbb{W}^\prime\left(s_j^E\right) $ to maximize its utility based on its resource supply, e.g., $\sum_{s_j^E\in\varphi\left(s_k^C\right)}\sum_{b_i\in \mathbb{B}_{j,k}}\beta_{i,j,k}+G_j^E+\sum_{s_j^E\in\varphi^\prime\left(s_k^C\right)}{|\mathbb{B}^\prime_{j,k}|}-\sum_{b_i\in\omega\left(s_j^E\right)}\alpha_i $ (lines 41-42), which makes the previous two phases mutually impact each other (i.e., the quantity of matched MUs for each ES during Phase 1 determines the resource demand for each ES in Phase 2, whereas the quantity of cloud resources acquired by an ES in Phase 2 directly influences the resource supply of that ES in Phase 1).
\subsection{Property Analysis}
We are interested in the following targets when designing OS-CLM mechanism. 

\begin{Defn}(Individual rationality of MU-ES matching in spot market) For both MUs and ESs, a matching $ \omega^\prime $ is individual rational when the following conditions are satisfied:
	
	\noindent
	$\bullet$ for MUs: each MU $ b_i\in\bm{B^\prime} $ receives a non-negative utility, i.e., constraint (68b) is satisfied.
	
	\noindent
	$\bullet$ for ESs: each ES $ s_j^E $ matched to a MU set $ \omega^\prime\left(s_j^E\right) $ can achieve a non-negative utility, i.e., constraint (69b) is satisfied, and then the number of MUs served by ES $ s_j^E $ should not exceed access $ K_j^E $, i.e., constraint (69d) is satisfied.
\end{Defn}

\begin{Defn}(Individual rationality of ES-CS matching) For both ESs and CSs, a matching $ \varphi^\prime $ is individual rational when the following conditions are satisfied:
	
	\noindent
	$\bullet$ for CSs: each CS $ s_k^C $ matched to an ES set $ \varphi^\prime\left(s_k^C\right) $ can achieve a non-negative utility, i.e.,
	\begin{equation}\label{key}{\small
			\begin{aligned}
				{u^{C\prime}}\left(s_k^C,\varphi^\prime\left(s_k^C\right)\right)>0,
		\end{aligned}}
	\end{equation}
	and its resource demand should not exceed its resource $ G_k^C $, i.e., constraint (70c) is satisfied.
	
	\noindent
	$\bullet$ for ESs: each ES $ s_j^E $ matched to a CS set $ \varphi^\prime\left(s_j^E\right) $ can achieve a non-negative utility, i.e., constraint (69c) is satisfied.
\end{Defn}

\begin{Defn}
	(Strong stability of OS-CLM) The proposed OS-CLM is strongly stable if MU-ES matching and ES-CS matching are individually rational and have no blocking pair or coalition.
\end{Defn}

\begin{Defn}(Competitive equilibrium associated with resource trading between MUs and ESs in OS-CLM) The trading between MUs and ESs reaches a competitive equilibrium if the following conditions are satisfied:
	
	\noindent
	$\bullet$ For each ES $ s_j^E \in \bm{S^{E{\prime}}} $, if $ s_j^E $ is associated with a MU $ b_i\in \bm{B^\prime} $, then $ c_{i,j}^E\leq p_{i,j}^{U\prime\rightarrow E\prime} $;
	
	\noindent
	$\bullet$ For each MU $ b_i\in \bm{ B^\prime }$, $ b_i $ is willing to trade with the ES that can bring it with the maximum utility;
	
	\noindent
	$\bullet$ For each MU $ b_i\in \bm{ B^\prime} $, if $ b_i $ is not associated with any ES, then the payment paid by $b_i$ is equal to it obtains valuation $ v_{i,j} $ through offloading the task to ES, i.e., $ v_{i,j} = p_{i,j}^{U\prime\rightarrow E\prime} $.
\end{Defn}

		\begin{Defn}(Competitive equilibrium associated with resource trading between ESs and CSs in OS-CLM) The trading between ESs and CSs reaches a competitive equilibrium if the following conditions are satisfied:
	
	\noindent
	$\bullet$ For each CS $ s_k^C \in \bm{ S^{C\prime}} $, if $ s_k^C $ is associated with an ES $ s_j^E\in \bm{S^{E{\prime}}} $, then $ c_{i,j,k}^C\leq p_{i,j,k}^{E\prime\rightarrow C\prime} $,
	
	\noindent
	$\bullet$ For each ES $ s_j^E\in \bm{S^{E{\prime}}} $, $ s_j^E $ is willing to trade with the CS that can bring it with the maximum utility,
	
	\noindent
	$\bullet$ For each ES $ s_j^E\in \bm{S^{E{\prime}}} $, if a task of $b_i$ associated with ES $ s_j^E $ is not trading with any CSs, then $ p_{i,j,k}^{E\prime\rightarrow C\prime} = p_{i,j}^{U\prime\rightarrow E\prime} $.
\end{Defn} 

\begin{Defn}(Weak Pareto optimality of OS-CLM) The proposed OS-CLM is weak Pareto optimal if there is no Pareto improvement.
\end{Defn}

We next examine the aforementioned property of OS-CLM, as outlined below:

\begin{lem}
	(Convergence of MU-ES matching of OS-CLM) MU-ES matching of Algorithm 2 converges within finite rounds.
\end{lem}
\begin{proof}
	As the MU-ES matching refers to a M2O matching, we utilize the Gale-Shapley algorithm to solve the matching game\cite{9667258,9687261}. After a finite number of rounds, each MU's payment can either be accepted or reach its obtains valuation $ v_{i,j} $ through offloading the task to ES $s_j^E$, which supports the property of convergence.
\end{proof}
\begin{lem}
	(Convergence of ES-CS matching of OS-CLM) ES-CS matching of Algorithm 2 converges within finite rounds.
\end{lem}
\begin{proof}
	Since the ES-CS matching refers to as a M2M matching, we adopt the Gale-Shapley algorithm to solve the matching game \cite{9667258,9687261}. After a finite number of rounds, the payment of $s_j^E$'s each task can be either accepted or reaches the payment from $ b_i $ to $ s_j^E $ (i.e., $ p_{i,j}^{U\prime\rightarrow E\prime} $), and thus supports the property of convergence.
\end{proof}
\begin{lem}
	(Individual rationality of MU-ES matching of OS-CLM) All the MUs and ESs are individual rational in the spot market.
\end{lem}
\begin{proof}	
	We offer the analysis on proving the individual rationality of MUs and ESs.
	
	\noindent 
	$\bullet$ \textbf{Individual rationality of MUs.} Since we have $ p_{i,j}^{U\prime\rightarrow E\prime}\left\langle m+1 \right\rangle\leftarrow \text{min}\left\{p_{i,j}^{U\prime\rightarrow E\prime}\left\langle m \right\rangle+\mathrm{\Delta}p_i,v_{i,j}\right\} $ in Algorithm 2, the final payment of each MU will definitely be lower than or equal to its valuation $ v_{i,j} $ through offloading the task to ES, which guarantees a non-negative utility. 

\noindent 
$\bullet$ \textbf{Individual rationality of ESs.} We set $ c_{i,j}^E\leq p_{i,j}^{U\prime\rightarrow E\prime} $ and $ p_{i,j,k}^{E\prime\rightarrow C\prime}\leq p_{i,j}^{U\prime\rightarrow E\prime} $ (e.g., (69b) and (69c)), ensuring that ES $ s_j^E $ obtains payment $p_{i,j}^{U\prime \leftarrow E\prime}$ from MU $ b_i $ dose exceed its service cost $c^E_{i,j}$ and dose exceed its payment $p_{i,j,k}^{E\prime \leftarrow C\prime}$ that offloaded task of $b_i$ to CS $ s_k^C $, which guarantees a non-negative utility for $ s_j^E $. And the line 12 in Algorithm 2, which ensures that the resource demand of $ s_j^E $ within a reasonable range.

As a summary, MUs and ESs are individual rationality in our proposed MU-ES matching of OS-CLM.
\end{proof}
\begin{lem}
 No blocking pair can exist in the MU-ES matching of OS-CLM.
\end{lem}
\begin{proof}
	We offer the proof to show there is no blocking pair of either Type 1 or Type 2, as following:
	
	\noindent 
	$\bullet$ \textbf{There is no Type 1 blocking pair in the MU-ES matching of OS-CLM.} We offer the proof by considering contradiction.
	
	 Under a given matching $ \omega^\prime $, MU $ b_i $ and ES $ s_j^E $ form a Type 1 blocking pair $ \left(b_i;s_j^E\right) $.
	If MU $ b_i $ does not trading with ES $ s_j^E $, the payment of MU $ b_i $ during the last round can only be its valuation $v_{i,j}$, as given by (79) and (80).
	\begin{equation}\label{key}{\small
		\begin{aligned}
			p_{i,j}^{U\prime\rightarrow E\prime}\left\langle m \right\rangle = v_{i,j},
		\end{aligned}}
	\end{equation}
	\begin{equation}\label{key}\small{
		\begin{aligned}
		{u^{U\prime\leftrightarrow E\prime}}\left(s_j^E,\omega^\prime\left(s_j^E\right)\backslash \widetilde{\omega^\prime}\left(s_j^E\right)\cup \left\{b_i\right\}\right)<{u^{U\prime\leftrightarrow E\prime}}\left(s_j^E,\omega^\prime\left(s_j^E\right)\right).
	\end{aligned}}
	\end{equation}
	
	If ES $ s_j^E $ selects MU $ b_i $, we have $ p_{i,j}^{U\prime\rightarrow E\prime}\left\langle m^{*} \right\rangle\leq p_{i,j}^{U\prime\rightarrow E\prime}\left\langle m \right\rangle = v_{i,j} $ and the following (81)
	\begin{equation}\label{key}{\small
		\begin{aligned}
			& {u^{U\prime\leftrightarrow E\prime}}\left(s_j^E,\omega^\prime\left(s_j^E\right)\backslash \widetilde{\omega^\prime}\left(s_j^E\right)\cup \left\{b_i\right\}\right) \geq\\& {u^{U\prime\leftrightarrow E\prime}}\left(s_j^E,\omega^\prime\left(s_j^E\right)\backslash \widetilde{\omega^{\prime\prime}}\left(s_j^E\right)\cup \left\{b_i\right\}\right),\\
		\end{aligned}}
	\end{equation}
	where $ 
	\widetilde{\omega^{\prime\prime}}\left(s_j^E\right) \subseteq \widetilde{\omega^{\prime}}\left(s_j^E\right) $. From (80) and (81), we can get
	\begin{equation}\label{key}{\small
		\begin{aligned}
			&{u^{U\prime\leftrightarrow E\prime}}\left(s_j^E,\omega^\prime\left(s_j^E\right)\right)>\\& {u^{U\prime\leftrightarrow E\prime}}\left(s_j^E,\omega\left(s_j^E\right)\backslash \widetilde{\omega^{\prime\prime}}\left(s_j^E\right)\cup \left\{b_i\right\}\right),
		\end{aligned}}
	\end{equation}
	which is contrary to (70), which thus ensures the inexistence of Type 1 blocking pair.
	
	\noindent 
	$\bullet$ \textbf{There is no Type 2 blocking pair in the MU-ES matching of OS-CLM.}
	We conduct the proof by considering cases of contradiction. 
	
	Under a given matching $ \omega^\prime $, MU $ b_i $ and ES $ s_j^E $ form a Type 2 blocking pair $ \left(b_i;s_j^E\right) $, as shown by (71).
	If ES $ s_j^E $ rejects MU $ b_i $, the payment of $ b_i $ during the last round can be set by $ p_{i,j}^{U\prime\rightarrow E\prime}\left\langle m \right\rangle = v_{i,j} $, where the only reason of such a rejection is that the overall resource demand exceeds its supply. However, the coexistence of (71) shows that ES $ s_j^E $ has an adequate resource to serve MUs, which contradicts our previous assumption. Therefore, we prove that there is no Type 2 blocking pair.
	
	As a summary, no blocking pair can exist in our proposed MU-ES matching in OS-CLM. 
\end{proof}

\begin{lem}
	(Individual rationality of ES-CS matching of OS-CLM) All the ESs and CSs are individual rational in the spot market.
\end{lem}
\begin{proof}
	The individual rationality of each ES and CS are proved respectively, as following:
	
\noindent 
$\bullet$ \textbf{Individual rationality of ESs.} Note that we have $ p_{j,k}^{E\prime\rightarrow C\prime}\left\langle n+1 \right\rangle\leftarrow \text{min}\left\{p_{i,j,k}^{E\prime\rightarrow C\prime}\left\langle n \right\rangle+\mathrm{\Delta}p_j,p_{i,j}^{U\prime\rightarrow E\prime}\left\langle m \right\rangle \right\} $ in Algorithm 2, the final payment of a task from ES $ s_j^E $ will definitely be lower than or at least equal to it obtains payment from $ b_i $ (i.e., $ p_{i,j}^{U\prime\rightarrow E\prime} $), which guarantees a non-negative utility. 

\noindent 
$\bullet$ \textbf{Individual rationality of CSs.} Since we set $ c_{i,j,k}^C\leq p_{i,j,k}^{E\prime\rightarrow C\prime} $, which ensures that CS $ s_k^C $ obtains payment from ES $ s_j^E $ stays above its corresponding service cost. This guarantees a non-negative utility for $ s_k^C $. And the line 31 in Algorithm 2, which ensures that the resource demand of $ s_k^C $ does not exceed its supply. 

As a result, our proposed M2M matching of OS-CLM in the futures market is individual rational for both ESs and CSs.
\end{proof}
\begin{lem}
	No blocking coalition can exist in the ES-CS matching of OS-CLM.
\end{lem}
\begin{proof}
We offer the proof to show there is no blocking coalition of either Type 1 or Type 2, as following:
	
	\noindent 
	$\bullet$ \textbf{There is no Type 1 blocking coalition in the ES-CS matching of OS-CLM.} We offer the proof by analyzing the cases of contradiction. 
	
	Under a given matching $ \varphi^\prime $, CS $ s_k^C $ and an ES set $ \mathbb{S}^\prime$, ($ \mathbb{S}^\prime\subseteq\bm{S^{E}}$) form a Type 1 blocking coalition $ \left(s_k^C;\mathbb{S}^\prime\right) $, as shown by (72) and (73).
	If a task of ES $ s_j^E $ does not trading with $ s_k^C $, the payment from ES $ s_j^E $ during the last round can only be equal to it obtains payment from $ b_i $, as given by (83) and (84).
	\begin{equation}\label{key}{\small
		\begin{aligned}
			p_{i,j,k}^{E\prime\rightarrow C\prime}\left\langle n \right\rangle = p_{i,j}^{U\prime\rightarrow E\prime}\left\langle m \right\rangle,
		\end{aligned}}
	\end{equation}
	\begin{equation}\label{key}{\small
		\begin{aligned}
		&{u^{E\prime\leftrightarrow C\prime}}\left(s_j^E,\varphi^\prime\left(s_j^E\right)\backslash\widetilde{\varphi^\prime}\left(s_j^E\right)\cup\left\{s_k^C\right\}\right)<\\&{u^{E\prime\leftrightarrow C\prime}}\left(s_j^E,\varphi^\prime\left(s_j^E\right)\right).
	\end{aligned}}
	\end{equation}
	
	If a task of $ s_j^E $ selects $ s_k^C $, we have $ p_{i,j,k}^{E\prime\rightarrow C\prime}\left\langle n^{*} \right\rangle \leq p_{i,j,k}^{E\prime\rightarrow C\prime}\left\langle n \right\rangle = p_{i,j}^{U\prime\rightarrow E\prime}\left\langle m \right\rangle $ and the following (85)
	\begin{equation}\label{key}{\small
		\begin{aligned}
			&{u^{E\prime\leftrightarrow C\prime}}\left(s_j^E,\varphi^\prime\left(s_j^E\right)\backslash\widetilde{\varphi^\prime}\left(s_j^E\right)\cup\left\{s_k^C\right\}\right) \geq\\& {u^{E\prime\leftrightarrow C\prime}}\left(s_j^E,\varphi^\prime\left(s_j^E\right)\backslash\widetilde{\varphi^{\prime\prime}}\left(s_j^E\right)\cup\left\{s_k^C\right\}\right),\\
		\end{aligned}}
	\end{equation}
	where $ 
	\widetilde{\varphi^{\prime\prime}}\left(s_j^E\right) \subseteq \widetilde{\varphi^\prime}\left(s_j^E\right) $. From (84) and (85), we can get
	\begin{equation}\label{key}{\small
		\begin{aligned}
			&{u^{E\prime\leftrightarrow C\prime}}\left(s_j^E,\varphi^\prime\left(s_j^E\right)\right)>\\&{u^{E\prime\leftrightarrow C\prime}}\left(s_j^E,\varphi^\prime\left(s_j^E\right)\backslash\widetilde{\varphi^{\prime\prime}}\left(s_j^E\right)\cup\left\{s_k^C\right\}\right),
		\end{aligned}}
	\end{equation}
	which is contrary to (73), and can thereby prove the inexistence of Type 1 blocking coalition.
	
	\noindent 
	$\bullet$ \textbf{There is no Type 2 blocking coalition in the ES-CS matching of OS-CLM.}
	Similarly, we also conduct the proof upon having contradiction.
	
	 Under a given matching $ \varphi^\prime $, CS $ s_k^C $ and an ES set $ \mathbb{S}^\prime$ ($ \mathbb{S}^\prime\subseteq\bm{S^{E{\prime}}}$) form a Type 2 blocking coalition $ \left(s_k^C;\mathbb{S}^\prime\right) $, as given by (74) and (75).
	If $ s_k^C $ rejects ES $ s_j^E $, the expense offered by $ s_j^E $ about task $b_i$ during the last round can only be $ p_{i,j,k}^{E\prime\rightarrow C\prime}\left\langle n \right\rangle = p_{i,j}^{U\prime\rightarrow E\prime}\left\langle m \right\rangle $, where the only reason of the rejection is that $s_k^C$ has no surplus resources. However, the coexistence of (74) and (75) shows that CS $ s_k^C $ has an adequate VMs to serve ESs, which contradicts the aforementioned assumption. Therefore, our proposed matching does not allow any Type 2 blocking coalition.
	
	All in all, there is no blocking coalition in our proposed ES-CS matching of OS-CLM.
\end{proof}
\begin{thm}
	(Strong stability of OS-CLM) OS-CLM is strongly stable.
\end{thm}
\begin{proof}
	Since the matching result of Algorithm 2 holds Lemma 9, Lemma 10, Lemma 11, and Lemma 12, according to Definition 21, our proposed OS-CLM in the spot market is strongly stable.
\end{proof}
\begin{thm}
	(Competitive equilibrium associated with resource trading between MUs and ESs in OS-CLM) The trading between MUs and ESs can reach a competitive equilibrium.
\end{thm}
\begin{proof}
	To prove this theorem, we also prove that the three conditions in Definition 22 (given in Appendix D.4) can be held in MU-ES matching. First, we have $ p_{i,j}^{U\prime\rightarrow E\prime} \geq c^E_{i,j} $ to make sure that service cost of the CS will be covered by its income paid from MUs in each round. We demonstrate next that when MU $ b_i $ trades with an ES $ s_j^E $, $ b_i $ achieves maximum utility. This is attributed to the fact that $ b_i $ selects the ES based on its preference list $ L_i^{U\prime} $ (e.g., line 5, Algorithm 2), ensuring the attainment of the maximum utility for $ b_i $. Then, if $ b_i $ has not been matched to any ES $ s_j^E \in \bm{C_i^\prime} $, the payment of $b_i$ should be equal to its valuation $ v_{i,j} $ (line 16, Algorithm 2). 
Based on Definition 22, we thereby verify that the considered MU-ES matching in the spot market can own the property of competitive equilibrium.
\end{proof}
\begin{thm}
	(Competitive equilibrium associated with resource trading between ESs and CSs in OS-CLM) The trading between ESs and CSs can reach a competitive equilibrium.
\end{thm}
\begin{proof}
		To prove this theorem, we also prove that three conditions in Definition 22 (given in Appendix D.4) can be held in ES-CS trading. First, we have $ p_{i,j,k}^{E\prime\rightarrow C\prime} \geq c^C_{i,j,k} $ to make sure that the service cost of the CS will be covered by its income paid from ESs in each round (i.e., constrain (70b)). We next demonstrate that when ES $ s_j^E $ chooses the trade with a CS $ s_k^C $, $ s_j^E $ attains maximum utility. This is attributed to the fact that $ s_j^E $ selects the CS based on its task-specific preference list $ L_{i,j}^{E\prime
		} $ (e.g., line 24, Algorithm 2), ensuring the achievement of the maximum utility for $ s_j^E $. Then, if a task of $ s_j^E $ has not been matched to any CS $ s_k^C $, the payment about the task from $ s_j^E $ (e.g., $ p_{i,j,k}^{E\prime\rightarrow C\prime} $ ) should be equal to it obtains payment from $ b_i $ (i.e., $ p_{i,j}^{U\prime\rightarrow E\prime} $) (line 34, Algorithm 2). 
Based on Definition 23, we thereby verify that the considered ES-CS trading in spot market can achieve competitive equilibrium.
\end{proof}

\begin{thm}
	(Weak Pareto optimality of OS-CLM) The proposed OS-CLM provides a weak Pareto optimalilty.
\end{thm}
\begin{proof}
	Review the design of OS-CLM, each participant (e.g., MU, ES, CS) makes decisions according to its preference list. If the subsequent choice ranks higher in the participant's preference list, they will switch their matching target in the following round. Such a switch operation indicates that backing to the previous choice will not bring it with any larger expected utility. For a MU $ b_i $, if there exists an ES $ s_j^E $ that can offer a higher utility than its current matched ES, $ b_i $ and $ s_j^E $ are more inclined to establish a matching relationship, this, however, will form a blocking pair. Since Theorem 5 verifies that our proposed OS-CLM is stable while allowing no blocking pairs. There exists no Pareto improvement when the procedure of MU-ES matching terminates. Similarly, we can infer that there is no Pareto improvement in ES-CS matching (e.g., Lemma 12 and Theorem 5). As a summary, our studied OS-CLM game is weak Pareto optimal.
\end{proof}

\section{EUA Dataset}
The Melbourne central business district area of EUA datasets as our simulation area in Fig. 6.
\begin{figure}[t!] \centering 
	\vspace{-0cm}
	\subfigbottomskip=-1pt
	\subfigcapskip=-180cm
	\setlength{\abovecaptionskip}{-0cm}
	\setlength{\belowcaptionskip}{-0.5cm}
	\subfigure[] { 
		\includegraphics[width=0.85\columnwidth]{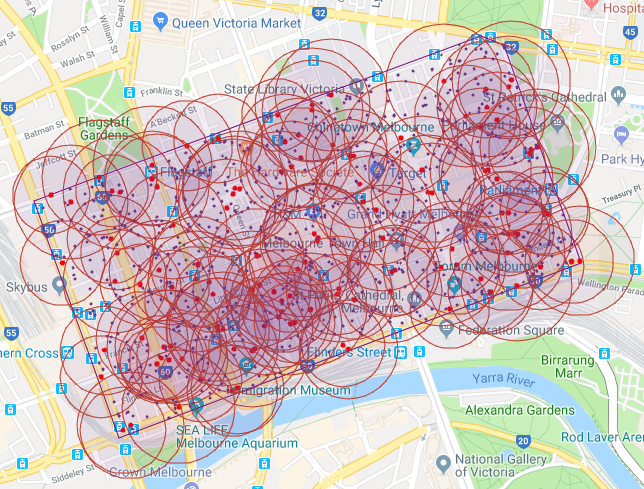} 
	} 
	
	\caption{Locations of base stations and mobile users within Melbourne central business district area in EUA Datasets.}  
	\label{fig} 
	\vspace{-0.3cm} 
\end{figure}

\section{Performance Evaluation on Utilities of MUs, ESs, and CSs}
\begin{figure} [b!]
	\centering    
	\vspace{-2.2cm}
	\subfigtopskip=2pt
	\subfigbottomskip=1pt
	\subfigcapskip=-2.0cm
	\setlength{\abovecaptionskip}{-1.6cm}
	\subfigure[] {
		\label{fig:a}     
		\includegraphics[width=0.495\columnwidth]{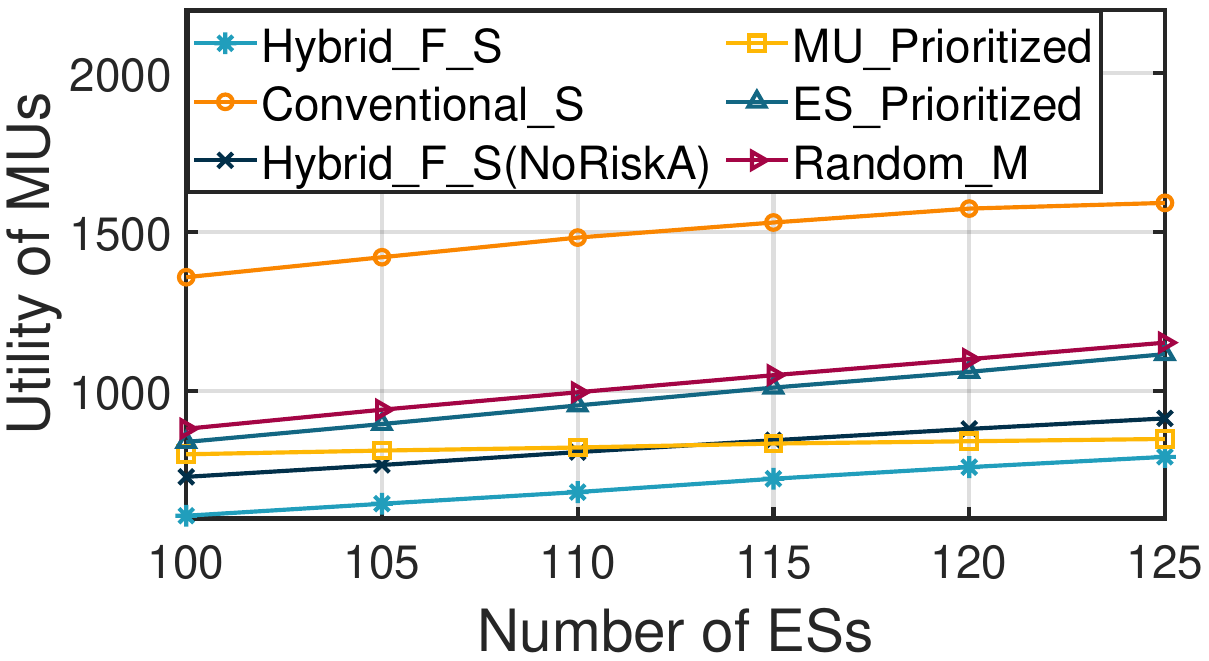}  
	}    \hspace{-6.0mm} \vspace{-35mm}
	\subfigure[] { 
		\label{fig:b}     
		\includegraphics[width=0.495\columnwidth]{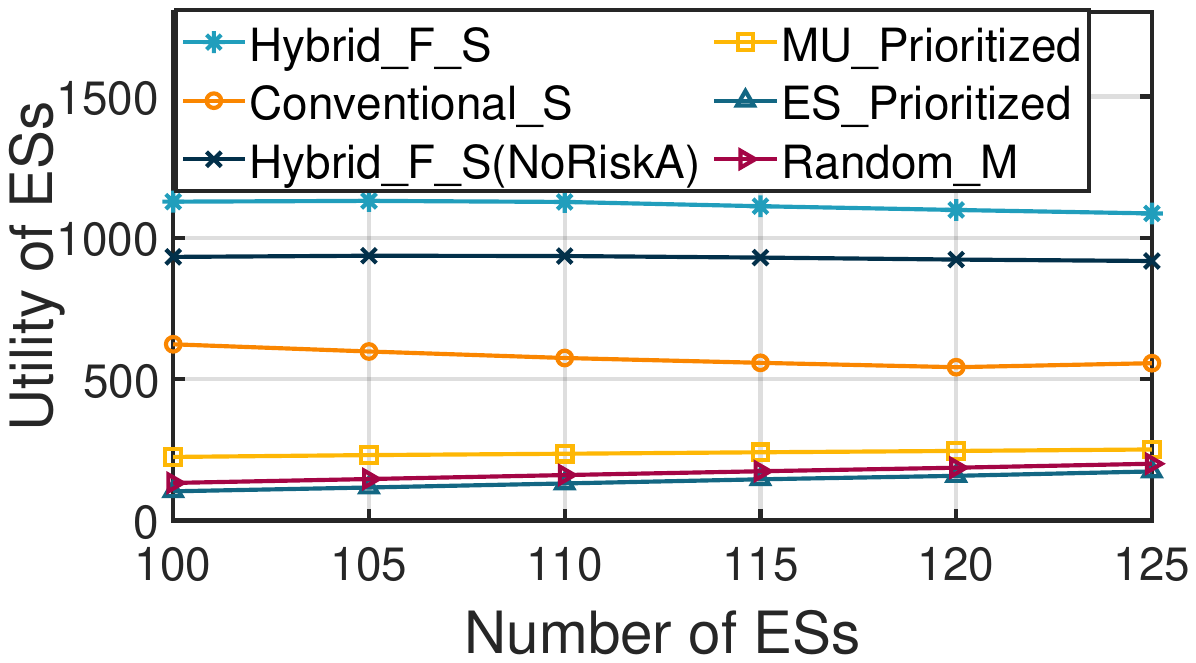}     
	}     
	\subfigure[] {
		\label{fig:c}     
		\includegraphics[width=0.495\columnwidth]{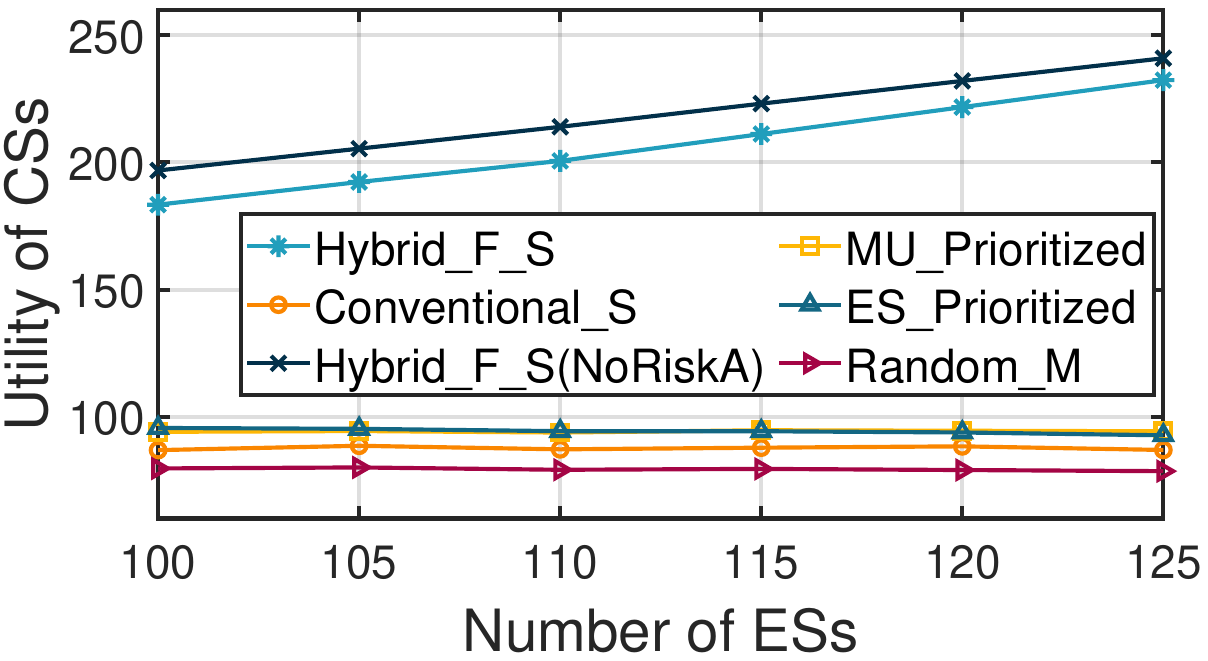}  
	}    \hspace{-6.0mm} \vspace{-35mm}
	\subfigure[] { 
		\label{fig:d}     
		\includegraphics[width=0.495\columnwidth]{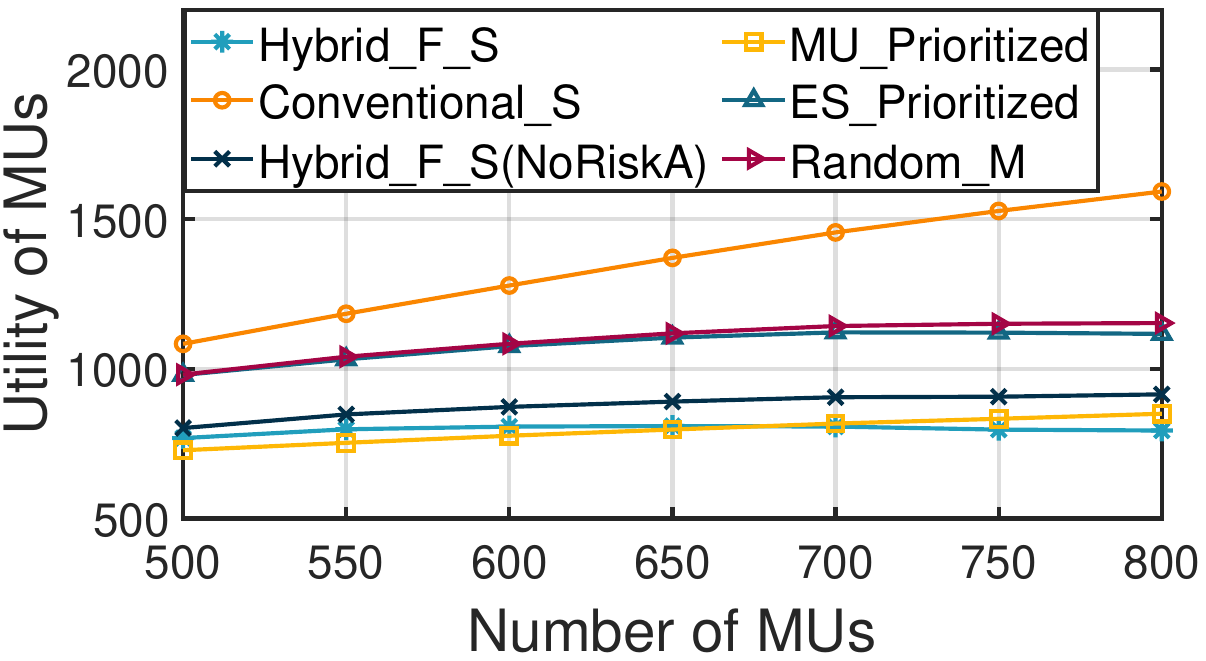}     
	}  
	\subfigure[] { 
		\label{fig:e}     
		\includegraphics[width=0.495\columnwidth]{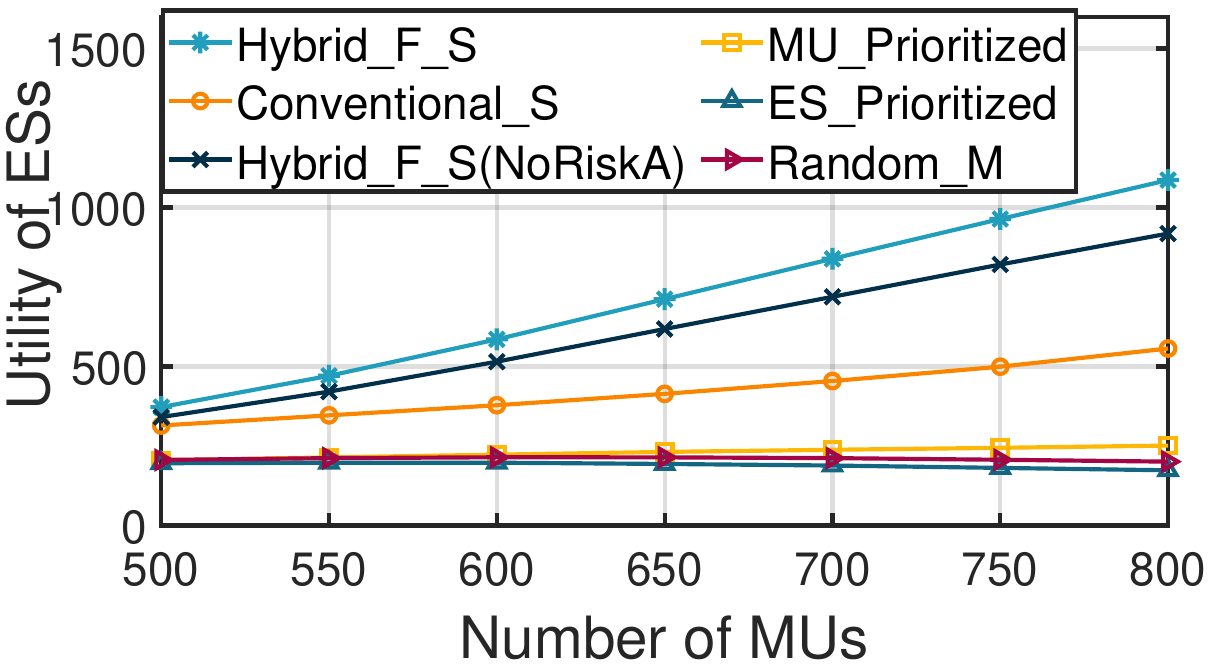}     
	}     \hspace{-6.0mm} 
	\subfigure[] {
		\label{fig:f}     
		\includegraphics[width=0.495\columnwidth]{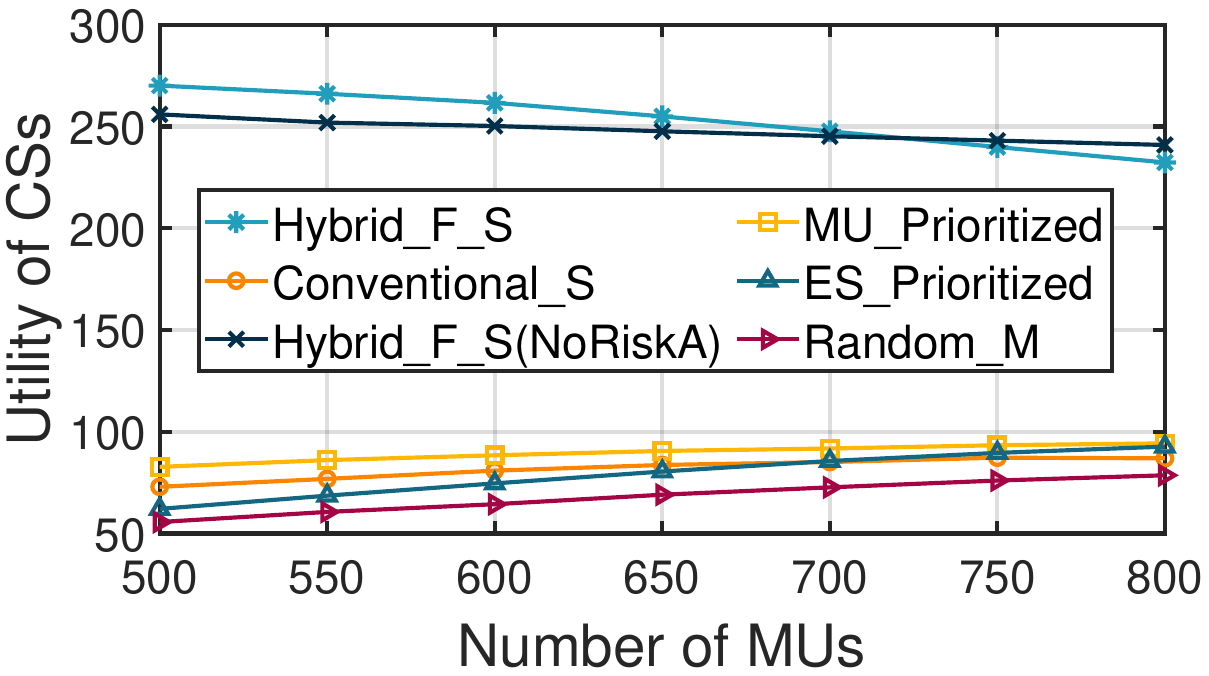}  
	}     
	\caption{Performance comparisons in terms of the utilities of MUs, ESs, CSs under different problem sizes, with (a)-(c) consider 800 MUs and 12 CSs, while (d)-(f) consider 125 ESs and 12CSs.}     
	\label{fig}   
	\vspace{-0.45 cm}  
\end{figure}

We study the individual utility of MUs, ESs, and CSs upon testing different numbers of ESs and MUs in Fig. 7, to better show how the social welfare are distributed among different parties.

Figs. 7(a)-(c) considers 800 MUs and 12 CSs, where in Fig. 7(a), the curves of all the methods show a rising trend since more ESs can bring more sufficient resources, while the demands of MUs can be met. Conventional\_S achieves the best performance on the utility of MUs due to its analysis of the current transaction conditions, which, in turn, suffers from excessive overhead (see Fig. 3). In our Hybrid\_F\_S, the forward contracts between MUs and ESs necessitates multiple rounds of negotiation to sign, which force MUs to continually increase their payment to enhance competitiveness, which, in turn, leads to the overall lowest utility of MUs. Even so, we introduce a good performance on social welfare (Fig. 2), while the risks of MUs have been well controlled although they may get rather low utilities. Namely, take part in our designed resource trading market can always get MUs with non-negative profits.
In Fig. 7(b), the curves of Conventional\_S, Hybrid\_F\_S, and Hybrid\_F\_S(NoRiskA) methods slightly decline. The underlying reason for this is that rather sufficient resources offers by the increasing number of ESs can weaken the competition between MUs, leading to a slight reduction of the utility of ESs. The proposed Hybrid\_F\_S outperforms other methods thanks to the well-designed overbooking and risk analysis. Besides, curves of ES\_Prioritized and Random\_M methods slightly increase with the number of ES raises, since they do not engage in bargaining. Then, the curve of MU\_Prioritized remains constant. This is because this method is puts MUS' interests at the first place, which will not impact the utility of ESs significantly.
As can be seen from Fig. 7(c), the utility of CSs for Hybrid\_F\_S(NoRiskA) remains higher than that of other methods, since Hybrid\_F\_S(NoRiskA) does not consider risk analysis, resulting in a high incidence of default by ESs, necessitating substantial compensation to CSs. Furthermore, the curve of our Hybrid\_F\_S increases with the number of ESs, owing to that the raising ESs can enable more ESs to sign forward contracts with CSs. Besides, the curves of other methods remain rather stable since they rely on a spot trading mode. Namely, under given demand of MUs, the resource requirements of CS resources can stay unchanged.

Figs. 7(d)-7(f) involves 125 ESs and 12 CSs. Particularly, Fig. 7(d) shows that curves of MUs' utility for all methods increase with a raising number of MUs, which is expected since the existence of more MUs implies a larger demand of resources. Conventional\_S achieves the best performance thanks to the analysis on current market conditions. However, in our proposed Hybrid\_F\_S, MUs and ESs need to conduct risk assessments, leading to more negotiations to reach a stable matching, as well as continuous competition among MUs. This finally results in lower utility of MUs in comparison with others. Similar to our analysis regarding Fig. 7(a), although we get non-ideal MUs' utility, our Hybrid\_F\_S outperforms baseline methods in other significant evaluation indicators (e.g., social welfare, time efficiency).
Fig. 7(e) reveals that the curves of Conventional\_S, Hybrid\_F\_S, and Hybrid\_F\_S(NoRiskA) methods show a raising trend. The reason is that these methods allow for bargains. In other words, a growing number of MUs can intensify competition among them, and thus bring more payments to ESs. In addition, the curves of other methods remain relatively stable due to the lack of bargain.
In Fig. 7(f), the curves of Conventional\_S, MU\_Prioritized, ES\_Prioritized, and Random\_M methods slightly increase with the increasing number of MUs, because larger resource demands encourage ESs to purchase services from CSs. Additionally, the utility of CSs of Hybrid\_F\_S and Hybrid\_F\_S(NoRiskA) slowly declines, due to that ESs have a wider selection of reliable MUs, and MUs are more likely to participate in practical transactions. This phenomenon, in turn, leads to the likelihood of ESs fulfilling their contracts with CSs. Furthermore, thanks to risk analysis, the penalty from ESs to CSs rapidly decreases with more MUs join in our considered market.
\end{document}